\documentclass[aap,preprint]{imsart}

\RequirePackage{amsthm,amsmath,amsfonts,amssymb}
\RequirePackage[numbers]{natbib}
\RequirePackage{graphicx}
\graphicspath{ {Images/} }

\RequirePackage{subcaption}
\RequirePackage{sidecap}

\usepackage[svgnames]{xcolor}
\usepackage{listings}

\usepackage{algpseudocode}
\usepackage{algorithm}

\RequirePackage{wrapfig}
\RequirePackage{enumitem}
\RequirePackage{hyperref}
\RequirePackage{relsize}
\RequirePackage{bigints}
\RequirePackage{verbatim}
\RequirePackage{float}
\RequirePackage{afterpage}
\RequirePackage{lscape}

\RequirePackage{array}
\RequirePackage{amsmath}

\RequirePackage{amsfonts}
\RequirePackage{amsthm}
\RequirePackage{amssymb}
\RequirePackage{bbm}
\RequirePackage{bm}
\RequirePackage{mathtools}

\RequirePackage{pgf,tikz}
\RequirePackage{mathrsfs}
\usetikzlibrary{arrows.meta} 
\RequirePackage{pgfplots}

\startlocaldefs
\numberwithin{equation}{section}
\theoremstyle{plain}
\newtheorem{theorem}{Theorem}[section]

\newtheorem{lemma}[theorem]{Lemma}

\newtheorem{assumption}[theorem]{Assumption}
\theoremstyle{remark}

\newcommand{\dif}{\mathrm{d}}

\newcommand{\Sperp}{\mathbb{S}^d \perp \mathbb{S}^d}

\newcommand{\opnorm}[1]{{\lvert\kern-0.25ex\lvert\kern-0.25ex\lvert #1 \rvert\kern-0.25ex\rvert\kern-0.25ex\rvert}}

\DeclareFontFamily{U}{mathx}{}
\DeclareFontShape{U}{mathx}{m}{n}{<-> mathx10}{}
\DeclareSymbolFont{mathx}{U}{mathx}{m}{n}
\DeclareMathAccent{\widecheck}{0}{mathx}{"71}

\allowdisplaybreaks

\pgfplotsset{compat=1.15}

\definecolor{grey}{rgb}{0.6274509803921569,0.6274509803921569,0.6274509803921569}

\endlocaldefs

\begin{document}

\begin{frontmatter}
\title{Adaptive Stereographic MCMC}

\begin{aug}
\author[A,B]{\fnms{Cameron}~\snm{Bell}},
\author[A]{\fnms{Krzysztof}~\snm{{\L}atuszy{\'n}ski}},
\and
\author[A]{\fnms{Gareth O.}~\snm{Roberts}}
\address[A]{Department of Statistics, University of Warwick, United Kingdom}
\address[B]{CEREMADE, Universit\'e Paris Dauphine-PSL, France}
\end{aug}

\begin{abstract}
In order to tackle the problem of sampling from heavy-tailed, high-dimensional distributions via Markov Chain Monte Carlo (MCMC) methods, \cite{yang2022stereographic} introduces the stereographic projection as a tool to compactify $\mathbb{R}^d$ and transform the problem into sampling from a density on the unit sphere $\mathbb{S}^d$. However, the improvement in algorithmic efficiency, as well as the computational cost of the implementation, is significantly impacted by the parameters used in this transformation.

To address this, we introduce adaptive versions of three stereographic MCMC algorithms - the Stereographic Random Walk (SRW), the Stereographic Slice Sampler (SSS), and the Stereographic Bouncy Particle Sampler (SBPS) - which automatically update the parameters of the algorithms as the run progresses. The adaptive setup allows for the power of the stereographic projection to be better exploited, even when the target distribution is neither centred nor homogeneous. Unlike Hamiltonian Monte Carlo (HMC) and other off-the-shelf MCMC samplers, the resulting algorithms are robust to starting far from the mean in heavy-tailed, high-dimensional settings. To prove convergence properties, we develop a novel framework for the analysis of adaptive MCMC algorithms over collections of simultaneously uniformly ergodic Markov operators, which is applicable to continuous-time processes, such as SBPS. This framework allows us to obtain $\mathcal{L}^2$ and almost sure convergence results, and a CLT for our adaptive stereographic algorithms.
\end{abstract}



\end{frontmatter}


\section{Introduction}
Markov chain Monte Carlo (MCMC) algorithms are used to approximate a given target distribution $\pi$ on $\mathbb{R}^d$, which usually arises in the context of Bayesian inference. For this purpose, a Markov process is simulated that admits $\pi$ as its stationary distribution, and we can use its empirical distribution to approximate $\pi$. It is therefore crucial that our algorithms quickly converge to stationarity and efficiently explore the entire target distribution. 

From the Random Walk Metropolis algorithm (RWM) \cite{metropolis1953equation,hastings1970monte} to Hamiltonian Monte Carlo (HMC) \cite{neal2012mcmc}, many MCMC algorithms rely on local moves to explore the space. Even more sophisticated algorithms, such as the Zig-Zag algorithm \cite{bierkens2019zig} or the Bouncy Particle Sampler (BPS) \cite{bouchard2018bouncy}, only move at a fixed speed, although versions of these algorithms exist with non-constant speeds \cite{bierkens2020boomerang,vasdekis2023speed}. If $\pi$ is heavy-tailed, a large amount of probability mass will be spread far from the mode of the target. All of the above algorithms will struggle to efficiently mix when targeting such distributions because they have a tendency to get lost in one corner of the tails, follow near-Brownian dynamics, and take an eternity to return to the mode. These problems are only made worse by the ``curse of dimensionality'': as the dimension $d$ increases, more and more of the volume in $\mathbb{R}^d$ is away from the mode, so there is ``more'' tail that needs efficiently exploring, whilst still needing to return quickly to the centre of the target. Hence, heavy-tailed densities are a major obstacle for efficient posterior sampling. Another is multimodality, though we will primarily discuss the case where $\pi$ is unimodal and refer to many other works attempting to address the challenges of sampling from multimodal targets (e.g.\ \cite{tawn2021annealed,pompe2020framework, MR4412989} and references therein).

With this motivation, one solution is to attempt to transform the sample space onto a compact set, effectively removing the possibility of getting lost in the tails, then sample from the transformed target distribution on the new support. \cite{yang2022stereographic} achieves this via the stereographic projection: this map transforms Euclidean space $\mathbb{R}^d$ onto $\mathbb{S}^d/N$, the unit sphere with the North Pole $N = (0,\dots,0,1)$ removed. In this setting, $N$ can be seen as the ``image of $\infty$'' under the map. Although $\mathbb{S}^d/N$ is not compact, it is relatively compact and can easily be extended to the compact set $\mathbb{S}^d$. This allows algorithms to reach anywhere in the state space in bounded time using only local moves. The algorithms introduced in their paper are shown to be uniformly ergodic for a wide range of target distributions, including heavy-tailed targets, and even exhibit a ``blessing of dimensionality'' in ideal settings, converging to stationarity faster as $d$ increases. 

Other papers, such as \cite{johnson2012variable, MR3911112}, also introduce transformations of the state space to improve sampling properties, but these transformations do not yield the same geometric benefits as the stereographic projection. These papers discuss only the impact of the transformations on mixing in the tails of the target, and do not discuss how they may cause the geometry in the high-probability regions to become more irregular. This is in stark contrast to the stereographic projection and its blessing of dimensionality.

However, having a compact support does not immediately lead to incredible sampling properties. If the target distribution $\pi$ is poorly preconditioned, the probability mass will be concentrated on a very small part of the sphere, which locally looks to our algorithms like a very small version of $\mathbb{R}^d$ and the potential benefits of using the stereographic projection are lost. We therefore parametrise the stereographic projection in order to attempt to evenly distribute the probability mass around $\mathbb{S}^d$. This is equivalent to preconditioning $\pi$ to be centred and appropriately scaled before we apply the transformation. 

With optimally chosen parameters, the probability mass becomes concentrated and uniformly spread around the equator of the sphere, which in turn improves the convergence and mixing of the MCMC algorithms. In practice, however, we will not know the optimal values for these parameters before running the process. A natural way of addressing this challenge is therefore to automatically update the parameters based on the history of the chain, and use these new, hopefully improved parameters in future transitions. This framework is known as adaptive MCMC, and is a relatively well-studied area \cite{MR1649199, MR2260070, roberts2007coupling, MR2461882, roberts2009examples, MR2759732, MR3012408, chimisov2018adapting, haario2001adaptive}.

To address the issues of potentially poorly specified parameters in the stereographic algorithms, in this paper, we create adaptive frameworks to update the parameters as we run our processes:
\begin{itemize}
    \item we present adaptive versions of the 3 stereographic MCMC algorithms we discuss;
    \item we provide a unifying theorem giving appropriate conditions for a Strong Law of Large Numbers (SLLN), $\mathcal{L}^2$ convergence and a CLT in each case;
    \item we prove this by showing that such a theorem holds when creating an adaptive version of any uniformly ergodic Markov process, whether discrete or continuous-time;
    \item we demonstrate the benefits and robustness of the adaptive scheme on two synthetic examples.
\end{itemize}

In Section \ref{section-stereographic-mcmc}, we formally introduce the stereographic projection, then the three algorithms: the Stereographic Random Walk (SRW), the Stereographic Slice Sampler (SSS) and the Stereographic Bouncy Particle Sampler (SBPS). The SRW and SBPS were first presented in \cite{yang2022stereographic}, and the SSS in \cite{habeck2023geodesic} under the name geodesic slice sampler.

In Section \ref{section-adaptive-stereographic}, we present our adaptive versions of each of the algorithms. Our algorithms are based on the Adapting Increasingly Rarely (AIR) MCMC setup from \cite{chimisov2018air,hofstadler2024almost}. We state the main convergence results in Theorem \ref{thm-fhat-asymptotics}. To theoretically justify the application of the AIR MCMC framework to our setting, we create a novel auxiliary process in Section \ref{section-segment-chain}, which we dub the segment chain. We show a SLLN, $\mathcal{L}^2$ convergence and a CLT for the segment chain, and show that any uniformly ergodic Markov process, whether discrete or continuous-time, inherits these results.

Finally, in Section \ref{section-simulations}, we present two synthetic examples demonstrating the ability of our algorithms to adapt the parameters of the transformation and the improvement this yields in sampling properties. The second example in particular demonstrates the stereographic algorithms' ability to start deep in the tails of heavy-tailed, high-dimensional target distributions with poor initial parameter choices and still find the modal region. By comparison, HMC fails to make any progress in finding the mode.

\section{Stereographic MCMC}
\label{section-stereographic-mcmc}

We start by defining the stereographic projection, then the three algorithms: the Stereographic Random Walk (SRW), the Stereographic Slice Sampler (SSS), and the Stereographic Bouncy Particle Sampler (SBPS).

\subsection{The Stereographic Projection}

The stereographic projection is a diffeomorphism from the punctured unit sphere $\mathbb{S}^d / \{N\}$, where $N = (0,\dots,0,1)$ is the ``North Pole'', to $\mathbb{R}^d$. Figure \ref{fig-stereo-project} presents the geometric intuition of the stereographic projection in the case $d=1$.

\begin{figure}[t!]
    \centering
        \begin{subfigure}{2.36in}
            \centering
            \begin{tikzpicture}[line cap=round,line join=round,>=Latex,x=1.25cm,y=1.25cm]
                \draw[->,color=black] (-2.,0.) -- (2.,0.);
                \foreach \x in {-2.,-1.5,-1.,-0.5,0.5,1.,1.5}
                \draw[shift={(\x,0)},color=black] (0pt,-2pt);
                \draw[->,color=black] (0.,-1.5) -- (0.,1.5);
                \foreach \y in {-1.5,-1.,-0.5,0.5,1.}
                \draw[shift={(0,\y)},color=black] (-2pt,0pt);
                \clip(-2.,-1.5) rectangle (2.,1.5);
                \draw [line width=0.8pt] (0.,0.) circle (1.25cm);
                \draw [line width=0.8pt,domain=-2.0:0.0] plot(\x,{(-0.5-1.*\x)/-0.5});
                \begin{scriptsize}
                    \draw [fill=red] (-0.5,0.) circle (2.5pt);
                    \draw [fill=black] (-0.8,-0.6) circle (2.5pt);
                \end{scriptsize}
            \end{tikzpicture}
        \end{subfigure}
        \begin{subfigure}{2.36in}
            \centering
            \begin{tikzpicture}[line cap=round,line join=round,>=Latex,x=1.25cm,y=1.25cm]
                \draw[->,color=black] (-2.,0.) -- (2.,0.);
                \foreach \x in {-2.,-1.5,-1.,-0.5,0.5,1.,1.5}
                \draw[shift={(\x,0)},color=black] (0pt,-2pt);
                \draw[->,color=black] (0.,-1.5) -- (0.,1.5);
                \foreach \y in {-1.5,-1.,-0.5,0.5,1.}
                \draw[shift={(0,\y)},color=black] (-2pt,0pt);
                \clip(-2.,-1.5) rectangle (2.,1.5);
                \draw [line width=0.8pt] (0.,0.) circle (1.25cm);
                \draw [line width=0.8pt,domain=0.0:2.0] plot(\x,{(--1.5-1.*\x)/1.5});
                \begin{scriptsize}
                    \draw [fill=red] (1.5,0.) circle (2.5pt);
                    \draw [fill=black] (0.9230769230769231,0.3846153846153846) circle (2.5pt);
                \end{scriptsize}
            \end{tikzpicture}
        \end{subfigure}
    \caption{The stereographic projection between $\mathbb{R}$ and $\mathbb{S}^1$: given $x \in \mathbb{R}$, we draw a line between $x$ (red) and the North Pole $N = (0,1)$, and define the projected point $z \in \mathbb{S}^1$ (black) to be point where the ray intersects the circle. We see that as $x \rightarrow \pm \infty$, $z \rightarrow N$.}
    \label{fig-stereo-project}
\end{figure}

Generalising to higher dimensions, and given a vector $\mu \in \mathbb{R}^d$ and a positive definite matrix $\Sigma \in \mathbb{R}^{d \times d}$, the (preconditioned) stereographic projection of a point $z \in \mathbb{S}^d$ gives
\begin{equation}
    x = \Sigma^{1/2} \left( \frac{z_1}{1- z_{d+1}}, \dots , \frac{z_d}{1 - z_{d+1}} \right) + \mu,
    \label{equ-stererographic-projection-x-preconditioned}
\end{equation}
and
\begin{equation}
    \begin{gathered}
        z_{1:d} = \frac{2 \Sigma^{-1/2}(x - \mu)}{\lVert \Sigma^{-1/2}(x-\mu) \rVert^2 +1}, \\[4pt]
        z_{d+1} = \frac{\lVert \Sigma^{-1/2}(x-\mu) \rVert^2 -1}{\lVert \Sigma^{-1/2}(x-\mu) \rVert^2 +1}.
    \end{gathered}
    \label{equ-stererographic-projection-z-preconditioned}
\end{equation}
This can be thought of as mapping from $z$ to $\Tilde{x} = \frac{z_{1:d}}{1-z_{d+1}}$, then preconditioning to obtain $x = \Sigma^{1/2} \Tilde{x} + \mu$. In other words, before we draw our ray from $N$ to $x$, we perform an affine transformation by translating $x$ by $-\mu$ and scaling via $\Sigma^{-1/2}$. Although mathematically equivalent, it is preferable to apply the transformations to $\mathbb{R}^d$, rather than turn $\mathbb{S}^d$ into an ellipsoid, so that the symmetry and constant curvature of the unit sphere are retained for simplicity of the resulting MCMC implementations. However, when communicating geometric intuition, talking instead about the equivalent transformation of adjusting the shape and location of the sphere is often convenient.

Setting $\gamma = (\mu, \Sigma) \in \Gamma$ as the parameter of the transformation and $\Gamma$ as full parameter space, we write 
\begin{equation}
    x = \text{SP}_\gamma(z) \quad \text{and} \quad z = \text{SP}_\gamma^{-1}(x) \label{equ-stereo-short}
\end{equation}
for the stereographic projection map \eqref{equ-stererographic-projection-x-preconditioned} and its inverse \eqref{equ-stererographic-projection-z-preconditioned}.

Under the stereographic projection, $N$ can be thought of as the image of $\infty$ in $\mathbb{S}^d$, folded into a single point, because as $\lVert x \rVert \rightarrow \infty$ in any way, $z = \text{SP}_\gamma^{-1}(x) \rightarrow N$. Moreover, we can travel vast distances in $\mathbb{R}^d$ by taking a very small step within the vicinity of the North Pole in $\mathbb{S}^d$. This will allow our MCMC algorithms to quickly explore the tails of the target distribution and then easily return to the high-probability region.

Using the Jacobian of the stereographic projection, $J_{\text{SP}_{\gamma}}(x) \propto  ( 1 + \lVert \Sigma^{-1/2}(x - \mu) \rVert^2 )^d,$ we transform the target density $\pi(x)$ on $\mathbb{R}^d$ to the following density on $\mathbb{S}^d$:
\begin{equation}
    \begin{split}
        \pi_\gamma(z) \propto & \pi(x) \left( 1 + \lVert \Sigma^{-1/2}(x - \mu) \rVert^2 \right)^d \\
        \propto & \pi(x) (1 - z_{d+1})^{-d},
    \end{split}
    \label{equ-target-preconditioned}
\end{equation}
where $x$ and $z$ are related by \eqref{equ-stereo-short}. For example, if $\pi(x) \propto (d + \lVert x \rVert^2)^d$ is a multivariate $t$-distribution (MtD) with $\nu = d$ degrees of freedom (DoF), then $\pi_\gamma$ is uniform on $\mathbb{S}^d$ when $\mu = 0_d$ and $\Sigma = dI_d$.

\subsection{Markov Processes on the Sphere}

Given the transformed density $\pi_\gamma$ on $\mathbb{S}^d$, we now discuss algorithms to efficiently sample from densities on the unit sphere. Since $\mathbb{S}^d$ is compact, we will see in Section \ref{section-uniform-ergodicity} that our stereographic MCMC algorithms can be uniformly ergodic, even in cases where $\pi$ has polynomial tails in $\mathbb{R}^d$.

\subsubsection{The Stereographic Random Walk}

We start with the SRW, a RWM algorithm on the sphere. This algorithm was introduced in \cite{yang2022stereographic} under the name stereographic projection sampler. We rename it to SRW for more consistent terminology and easier differentiation of acronyms. 

Given a position $z \in \mathbb{S}^d$, we propose the next point by taking a Gaussian step in the hyperplane tangent to $\mathbb{S}^d$ at $z$, then normalising the vector to return to $\mathbb{S}^d$. This proposal scheme is shown in Figure \ref{fig-srw-proposal}.

\begin{figure}[b!]
    \centering
    \begin{tikzpicture}[line cap=round,line join=round,>=Latex,x=1.6666666666666667cm,y=1.6666666666666667cm]
        \draw[->,color=black] (-1.5,0.) -- (1.5,0.);
        \foreach \x in {-1.4,-1.2,-1.,-0.8,-0.6,-0.4,-0.2,0.2,0.4,0.6,0.8,1.,1.2,1.4}
        \draw[shift={(\x,0)},color=black] (0pt,-2pt);
        \draw[->,color=black] (0.,-1.5) -- (0.,1.5);
        \foreach \y in {-1.4,-1.2,-1.,-0.8,-0.6,-0.4,-0.2,0.2,0.4,0.6,0.8,1.,1.2,1.4}
        \draw[shift={(0,\y)},color=black] (-2pt,0pt);
        \clip(-1.5,-1.5) rectangle (1.5,1.5);
        \draw [line width=0.8pt] (0.,0.) circle (1.6666666666666667cm);
        \draw [line width=1.2pt,dash pattern=on 2pt off 2pt,color=red] (1.1619944964788085,-0.6085544393252934)-- (0.9755921556337759,0.2195904047672446);
        \draw [line width=1.2pt,dash pattern=on 2pt off 2pt,color=blue] (0.8858657967534184,-0.4639415805275825)-- (1.1619944964788085,-0.6085544393252934);
        \begin{scriptsize}
            \draw [color=black] (0.9755921556337759,0.2195904047672446)-- ++(-2.5pt,-2.5pt) -- ++(5.0pt,5.0pt) ++(-5.0pt,0) -- ++(5.0pt,-5.0pt);
            \draw[color=black] (1.0892874940443935,0.30662150717376896) node {$z$};
            \draw [fill=red] (1.1619944964788085,-0.6085544393252934) circle (1.5pt);
            \draw[color=red] (1.1942178386527317,-0.7384253758797918) node {$z + dz$};
            \draw[color=red] (1.1950782601939334,-0.1371461297370395) node {$dz$};
            \draw [color=blue] (0.8858657967534184,-0.4639415805275825)-- ++(-2.0pt,-2.0pt) -- ++(4.0pt,4.0pt) ++(-4.0pt,0) -- ++(4.0pt,-4.0pt);
            \draw[color=blue] (0.7704037648141462,-0.38735554203781447) node {$z^\prime$};
        \end{scriptsize}
    \end{tikzpicture}
    \caption{An SRW proposal}
    \label{fig-srw-proposal}
\end{figure}

To implement this, given $Z_n = z$, we sample $\Tilde{dz} \sim \mathcal{N}(0_{d+1}, h^2 I_{d+1})$ and set $dz = \Tilde{dz} - (z~\cdot~\Tilde{dz})z$ to be the Gaussian step orthogonal to $z$. Here $a \cdot b$ denotes the standard inner product between vectors $a$ and $b$. We then let our proposal point be $z^\prime = \frac{z + dz}{\lVert z + dz \rVert}$. The symmetrical nature of the Gaussian distribution ensures the proposal kernel $q(z,z^\prime)$ is reversible with respect to the uniform measure on~$\mathbb{S}^d$, i.e.
\begin{equation}
    q(z,z^\prime) = q(z^\prime,z).
    \label{equ-reversible-srw-proposal}
\end{equation}
We then accept $z^\prime$ as the new position with probability $\text{min}\left( \frac{\pi_\gamma(z^\prime)}{\pi_\gamma(Z_n)}, 1\right)$, since the proposal is reversible. The overall algorithm is given in Algorithm \ref{alg-srw}.

\begin{algorithm}[t!]
    \begin{itemize}
        \item[\textbf{Input:}]
        Target density $\pi$ on $\mathbb{R}^d$, $X_0 \in \mathbb{R}^d$, parameters $\gamma \in \Gamma$, $h>0$, $Z_0 = \text{SP}^{-1}_\gamma( X_0 )$ 
        \item[\textbf{Output:}]
        $\{( X_n, Z_n )\}_{n \in \mathbb{N}}$ [2pt]
        \item[\textbf{For:}] $n = 0,1,\dots$\textbf{:}
        \item Sample $\Tilde{dz} \sim \mathcal{N}(0_{d+1},h^2I_{d+1})$ in $\mathbb{R}^{d+1}$, and set $dz = \Tilde{dz} - (\Tilde{dz} \cdot Z_n)Z_n$
        \item Set $z^\prime = \frac{Z_n + dz}{\lVert Z_n + dz \rVert}$, and set $Z_{n+1} = z^\prime$ with probability
        \begin{equation*}
            \text{min}\left( \frac{\pi_\gamma(z^\prime)}{\pi_\gamma(Z_n)}, 1\right) = \text{min}\left( \frac{\pi(x^\prime)(1-z^\prime_{d+1})^{-d}}{\pi(X_n)(1-Z_{n,d+1})^{-d}}, 1\right),
        \end{equation*}
        where $x^\prime = \text{SP}_\gamma(z^\prime)$. Otherwise, set $Z_{n+1} = Z_n$
        \item Set $X_{n+1} = \text{SP}_\gamma(Z_{n+1})$
    \end{itemize}
    \caption{The Stereographic Random Walk}\label{alg-srw}
\end{algorithm}

From Equation \eqref{equ-reversible-srw-proposal}, we see that this is indeed a RWM algorithm, and therefore has invariant measure $\pi_\gamma$ on $\mathbb{S}^d$. The projection of the sample path onto $\mathbb{R}^d$ will therefore have stationary distribution $\pi$. Furthermore, the SRW is simultaneously uniformly ergodic, even when targeting densities as heavy-tailed as a MtD with at least $d$ DoF (see Lemma \ref{lemma-srw-uniform-ergodicity}).

\subsubsection{The Stereographic Slice Sampler}

\cite{habeck2023geodesic} presents the geodesic slice sampler as a method for sampling from distributions on $\mathbb{S}^d$. We repurpose it as the SSS and apply it to $\pi_\gamma$.

From a position $z \in \mathbb{S}^d$, we start by sampling $v$ uniformly from the set
\begin{equation}
    \{z\}^\perp = \{ v \in \mathbb{S}^d : v \cdot z = 0\},
    \label{equ-z-perp}
\end{equation}
which we write as $v \sim p(\cdot \mid z)$. This defines a geodesic of the form $\{ z\cos(\theta) + v\sin(\theta)$ : $\theta \in [0;2\pi)\}$. We also define the space of orthonormal pairs $(z,v)$ as
\begin{equation}
    \Sperp = \{ (z,v) \in \mathbb{S}^d \times \mathbb{S}^d : z \cdot v =0\}.
    \label{equ-sperp}
\end{equation}

We then aim to perform a slice sampling step targeting the measure on the one-dimensional geodesic with density proportional to $\pi_\gamma$. We start by sampling $t \sim U(0,\pi_\gamma(z))$, then wish to sample a new position $z^\prime$ uniformly from the superlevel set
\begin{equation}
    L_{z,v}(t) = \{ z^\prime = z \cos(\theta) + v \sin(\theta) : \theta \in [0,2\pi), \pi_\gamma(z^\prime) > t \}.
    \label{equ-level-set}
\end{equation}

In practice, exact sampling from the uniform distribution on $L_{z,v}(t)$ requires rejection sampling and can be very inefficient. To improve efficiency, at each step \cite{habeck2023geodesic} instead uses an adaptive rejection sampling algorithm, referred to as the shrinkage procedure =(this is the same procedure as used in the elliptical slice sampler \cite{murray2010elliptical}). This process is explicitly stated in Algorithm \ref{alg-sss-shrinkage}, and depicted geometrically in Figure \ref{fig-sss-shrinkage}.

\begin{algorithm}[b!]
    \begin{itemize}
        \item[\textbf{Input:}] Target density $p$ on $\mathbb{S}^d$, orthonormal pair $(z,v) \in \Sperp$, level $t \in (0,p(z))$ 
        \item[\textbf{Output:}] $z^\prime \in L_{z,v}(t)$
        \item[\textbf{Initialisation:}] Sample $\theta \sim U(0,2\pi)$, and set $\theta_\text{max} = \theta$, $\theta_\text{min} = \theta - 2\pi$ 
        \item[\textbf{While}] $p( z\cos(\theta) + v\sin(\theta)) \leq t$\textbf{:}
        \item \textbf{If} $\theta <0$: Set $\theta_\text{min} = \theta$
        \item \textbf{Else}: Set $\theta_\text{max} = \theta$
        \item Resample $\theta \sim U(\theta_\text{min},\theta_\text{max})$
        \item[\textbf{Return:}] $z^\prime = z\cos(\theta) + v\sin(\theta)$.
    \end{itemize}
    \caption{The Shrinkage Procedure for the Stereographic Slice Sampler}\label{alg-sss-shrinkage}
\end{algorithm}

\begin{figure}[t!]
    \centering
    \begin{subfigure}{2.36in}
        \centering
        \begin{tikzpicture}[line cap=round,line join=round,>=Latex,x=1cm,y=1cm]
            \draw[->,color=black] (-2.,0.) -- (2.,0.);
            \draw[->,color=black] (0.,-1.5) -- (0.,1.5);
            \draw [line width=1pt,color=grey] (0,0) circle (1cm);
            \draw [shift={(0,0)},line width=1pt,color=red]  plot[domain=-0.0657487157210559:0.3919475254685133,variable=\t]({1*1*cos(\t r)+0*1*sin(\t r)},{0*1*cos(\t r)+1*1*sin(\t r)});
            \draw [shift={(0,0)},line width=1pt,color=red]  plot[domain=3.9124990779217423:5.494837299830893,variable=\t]({1*1*cos(\t r)+0*1*sin(\t r)},{0*1*cos(\t r)+1*1*sin(\t r)});
            \draw [shift={(0,0)},line width=1pt,color=red]  plot[domain=2.302674221733029:2.9144869629557855,variable=\t]({1*1*cos(\t r)+0*1*sin(\t r)},{0*1*cos(\t r)+1*1*sin(\t r)});
            \begin{scriptsize}
                \draw [fill=red] (1,0) circle (1.5pt);
                \draw[color=red] (1.2,0.2) node {$z$};
                \draw [fill=blue] (-0.46702917587107556,0.884241906315904) circle (1.5pt);
                \draw[color=blue] (-0.45,1.2) node {$\theta_0$};
            \end{scriptsize}
        \end{tikzpicture}
        \caption{We start with an initial proposal $\theta_0$, targeting the set $L_{z,v}(t)$ (shown in red).}
    \end{subfigure}
    \hfill
    \begin{subfigure}{2.36in}
        \centering
        \begin{tikzpicture}[line cap=round,line join=round,>=Latex,x=1cm,y=1cm]
            \draw[->,color=black] (-2.,0.) -- (2.,0.);
            \draw[->,color=black] (0.,-1.5) -- (0.,1.5);
            \draw [line width=1pt,color=grey] (0,0) circle (1cm);
            \draw [shift={(0,0)},line width=1pt,color=red]  plot[domain=-0.0657487157210559:0.3919475254685133,variable=\t]({1*1*cos(\t r)+0*1*sin(\t r)},{0*1*cos(\t r)+1*1*sin(\t r)});
            \draw [shift={(0,0)},line width=1pt,color=red]  plot[domain=3.9124990779217423:5.494837299830893,variable=\t]({1*1*cos(\t r)+0*1*sin(\t r)},{0*1*cos(\t r)+1*1*sin(\t r)});
            \draw [shift={(0,0)},line width=1pt,color=red]  plot[domain=2.302674221733029:2.9144869629557855,variable=\t]({1*1*cos(\t r)+0*1*sin(\t r)},{0*1*cos(\t r)+1*1*sin(\t r)});
            \begin{scriptsize}
                \draw [fill=red] (1,0) circle (1.5pt);
                \draw[color=red] (1.2,0.2) node {$z$};
                \draw [fill=blue] (-0.46702917587107556,0.884241906315904) circle (1.5pt);
                \draw[color=blue] (-0.45,1.2) node {$\theta_0$};
                \draw [fill=blue] (-0.8610897474504623,-0.5084529937326547) circle (1.5pt);
                \draw[color=blue] (-1,-0.7) node {$\theta_1$};
            \end{scriptsize}
        \end{tikzpicture}
        \caption{If we reject $\theta_0$, propose a new point from the geodesic.}
    \end{subfigure}
    \\[10pt]
    \centering
    \begin{subfigure}{2.36in}
        \centering
        \begin{tikzpicture}[line cap=round,line join=round,>=Latex,x=1cm,y=1cm]
            \draw[->,color=black] (-2.,0.) -- (2.,0.);
            \draw[->,color=black] (0.,-1.5) -- (0.,1.5);
            \draw [shift={(0,0)},line width=1pt,color=grey]  plot[domain=-2.6082053862013876:2.0567243708247496,variable=\t]({1*1*cos(\t r)+0*1*sin(\t r)},{0*1*cos(\t r)+1*1*sin(\t r)});
            \draw [shift={(0,0)},line width=1pt,color=red]  plot[domain=-0.0657487157210559:0.3919475254685133,variable=\t]({1*1*cos(\t r)+0*1*sin(\t r)},{0*1*cos(\t r)+1*1*sin(\t r)});
            \draw [shift={(0,0)},line width=1pt,color=red]  plot[domain=3.9124990779217423:5.494837299830893,variable=\t]({1*1*cos(\t r)+0*1*sin(\t r)},{0*1*cos(\t r)+1*1*sin(\t r)});
            \begin{scriptsize}
                \draw [fill=red] (1,0) circle (1.5pt);
                \draw[color=red] (1.2,0.2) node {$z$};
                \draw [fill=blue] (-0.46702917587107556,0.884241906315904) circle (1.5pt);
                \draw[color=blue] (-0.45,1.2) node {$\theta_0$};
                \draw [fill=blue] (-0.8610897474504623,-0.5084529937326547) circle (1.5pt);
                \draw[color=blue] (-1,-0.7) node {$\theta_1$};
                \draw [fill=blue] (0.7191229524372548,0.6948828529168967) circle (1.5pt);
                \draw[color=blue] (0.95,0.9) node {$\theta_2$};
            \end{scriptsize}
        \end{tikzpicture}
        \caption{If we reject $\theta_1$, remove the segment which does not contain $z$. Sample $\theta_2$ uniformly from the new interval.}
    \end{subfigure}
    \hfill
    \begin{subfigure}{2.36in}
        \centering
        \begin{tikzpicture}[line cap=round,line join=round,>=Latex,x=1cm,y=1cm]
            \draw[->,color=black] (-2.,0.) -- (2.,0.);
            \draw[->,color=black] (0.,-1.5) -- (0.,1.5);
            \draw [shift={(0,0)},line width=1pt,color=grey]  plot[domain=-2.6082053862013876:0.7682569852588388,variable=\t]({1*1*cos(\t r)+0*1*sin(\t r)},{0*1*cos(\t r)+1*1*sin(\t r)});
            \draw [shift={(0,0)},line width=1pt,color=red]  plot[domain=-0.0657487157210559:0.3919475254685133,variable=\t]({1*1*cos(\t r)+0*1*sin(\t r)},{0*1*cos(\t r)+1*1*sin(\t r)});
            \draw [shift={(0,0)},line width=1pt,color=red]  plot[domain=3.9124990779217423:5.494837299830893,variable=\t]({1*1*cos(\t r)+0*1*sin(\t r)},{0*1*cos(\t r)+1*1*sin(\t r)});
            \begin{scriptsize}
                \draw [fill=red] (1,0) circle (1.5pt);
                \draw[color=red] (1.2,0.2) node {$z$};
                \draw [fill=blue] (-0.8610897474504623,-0.5084529937326547) circle (1.5pt);
                \draw[color=blue] (-1,-0.7) node {$\theta_1$};
                \draw [fill=blue] (0.7191229524372548,0.6948828529168967) circle (1.5pt);
                \draw[color=blue] (0.95,0.9) node {$\theta_2$};
                \draw [fill=green] (0.503114016785377,-0.8642200449619779) circle (1.5pt);
                \draw[color=green] (0.7,-1.1) node {$\theta_3$};
            \end{scriptsize}
        \end{tikzpicture}
        \caption{Repeat this process, iteratively shrinking the size of the interval around $z$, until we propose a point in $L_{z,v}(t)$.}
    \end{subfigure}
    
    \caption{The Shrinkage Procedure for the Stereographic Slice Sampler: consider an initial position $z$ and a geodesic defined by some $v$. The set $L_{z,v}(t)$ is shown in red, and the search interval sequentially shrinks based on the rejected points.}
    \label{fig-sss-shrinkage}
\end{figure}

The shrinkage procedure, which we denote $\text{Shrink}(z,v,t)$, is then incorporated into a full Markov kernel: given a current position $Z_n = z \in \mathbb{S}^d$, we sample $t \sim U(0, \pi_\gamma(z))$, $v \sim p(\cdot \mid z)$, then sample $Z_{n+1} \sim \text{Shrink}(z,v,t)$. We summarise this in Algorithm \ref{alg-sss}. \cite{habeck2023geodesic} prove in Proposition 15 that the shrinkage procedure is reversible with respect to the uniform distribution on $L_{z,v}(t)$, and therefore that the SSS kernel is $\pi_\gamma$ reversible, and so has the correct stationary distribution. Furthermore, the SSS is simultaneously uniformly ergodic, even when targeting densities as heavy-tailed as a MtD with at least $d$ DoF (see Lemma \ref{lemma-sss-uniform-ergodicity}).

\begin{algorithm}[t!]
    \begin{itemize}
    \item[\textbf{Input:}]
    Target density $\pi$ on $\mathbb{R}^d$, $X_0 \in \mathbb{R}^d$, parameter $\gamma \in \Gamma$, $Z_0 = \text{SP}^{-1}_\gamma( X_0 )$ 
    \item[\textbf{Output:}]
    $\{( X_n, Z_n )\}_{n \in \mathbb{N}}$
    \item[\textbf{For}] $n = 0,1,\dots$\textbf{:}
    \item Sample $T_n \sim U(0,\pi_\gamma(Z_n))$, and $V_n \sim p(\cdot \mid Z_n)$
    \item Sample $Z_{n+1} \sim \text{Shrink}(Z_n,V_n,T_n)$ according to the shrinkage procedure in Algorithm \ref{alg-sss-shrinkage}
    \item Set $X_{n+1} = \text{SP}_\gamma(Z_{n+1})$
    \end{itemize}
    \caption{The Stereographic Slice Sampler}\label{alg-sss}
\end{algorithm}

When comparing the SSS to the Elliptical Slice Sampler (ESS) \cite{murray2010elliptical}, both behave almost identically when on the equator. The ESS considers an elliptical path through $X$ and some $V \sim \mathcal{N}(0_d, I_d)$, then uses an identical shrinkage procedure to sample from the appropriate superlevel set. This results in a very similar move to the SSS.

On the other hand, Figure \ref{fig-sss-ess-proposals} shows the difference in the proposal mechanisms of the SSS and ESS when $\lVert X \rVert$ is large. In particular, the projection of the SSS proposal onto $\mathbb{R}^d$ does not put a uniform density on the circle it is considering, making it significantly more likely to propose moves near the mode, even if the current position is far in the tails.

\begin{figure}[b]
    \centering
    \begin{subfigure}{0.45\textwidth}
        \centering
        \vspace{-10pt}        
        \begin{tikzpicture}[line cap=round,line join=round,>=Latex,x=1cm,y=1cm]
            \clip(-2.7991631799163184, -1.8659688144266038) rectangle (2.7991631799163184,1.8746169596319737);
            \draw [line width=1pt] (0,0) circle (1cm);
            \draw [line width=1pt,color=red] (-1.05,0) circle (1.45cm);
            \draw [shift={(-1.05,0)},line width=1pt,color=blue]  plot[domain=0.7610127542247297:5.522172552954856,variable=\t]({1*1.45*cos(\t r)+0*1.45*sin(\t r)},{0*1.45*cos(\t r)+1*1.45*sin(\t r)});
            
            \begin{scriptsize}
                \draw [fill=grey] (0,0) circle (2pt);
                \draw[color=black] (-0.2,0.2) node {$O$};
                \draw [fill=blue] (-2.5,0) circle (2.5pt);
                \draw[color=blue] (-2.7,0.2) node {$X$};
                \draw [fill=grey] (0,1) circle (2pt);
                \draw[color=black] (0.2,1.2) node {$V$};
            \end{scriptsize}
            
        \end{tikzpicture}

        \vspace{-15pt}
        \caption{SSS ellipse of interest}
        \vspace{-5pt}
    \end{subfigure}
    \hfill
    \begin{subfigure}{0.45\textwidth}
        \centering
        \vspace{-10pt}       
        \begin{tikzpicture}[line cap=round,line join=round,>=Latex,x=1cm,y=1cm]
            \clip(-2.7991631799163184, -1.8659688144266038) rectangle (2.7991631799163184,1.8746169596319737);
            \draw [line width=1pt] (0,0) circle (1cm);
            \draw [line width=1pt,color=red]  plot[domain=-pi/2:pi/2,variable=\t]({2.5*cos(\t r)},{sin(\t r)});
            
            \draw [line width=1pt,color=blue]  plot[domain=pi/2:3*pi/2,variable=\t]({2.5*cos(\t r)},{sin(\t r)});
            
            \begin{scriptsize}
                \draw [fill=grey] (0,0) circle (2pt);
                \draw[color=black] (-0.2,0.2) node {$O$};
                \draw [fill=blue] (-2.5,0) circle (2.5pt);
                \draw[color=blue] (-2.7,0.2) node {$X$};
                \draw [fill=grey] (0,1) circle (2pt);
                \draw[color=black] (0.2,1.2) node {$V$};
            \end{scriptsize}
        \end{tikzpicture}
        \vspace{-15pt}
        \caption{ESS ellipse of interest}
        \vspace{-5pt}
    \end{subfigure}
    \caption[Proposals for SSS and ESS]{Proposal ellipses and densities for the SSS and ESS when $\lVert X \rVert^2$ is large. The blue and red regions have equal probability mass under the first proposal. The SSS will therefore more quickly propose points in the high-probability region.}
    \label{fig-sss-ess-proposals}
\end{figure}

\subsubsection{The Stereographic Bouncy Particle Sampler}

Our last algorithm, the SBPS, was introduced in \cite{yang2022stereographic}, and is a Piecewise Deterministic Markov Process (PDMP) targeting~$\pi_\gamma$. Unlike the previous two algorithms, it is defined as a continuous-time process over $\Sperp$, as defined in equation \eqref{equ-sperp}, where the position $z$ is considered to be our sample point, and the velocity $v$ is a latent variable used to explore the space. Given an initial pair $(z_0,v_0) \in \Sperp$, the process evolves deterministically along a geodesic, according to the dynamics
\begin{equation}
    \begin{split}
        z(t) &= z_0\ \cos(t) + v_0\ \sin(t), \\
        v(t) &= v_0\ \cos(t) - z_0\ \sin(t).
    \end{split}
    \label{equ-deterministic}
\end{equation}

We then introduce two types of random events at which we change the velocity $v$. Once an event occurs, we resume our deterministic dynamics according to Equation \eqref{equ-deterministic}, but with the new velocity.

Bounce events occur according to an inhomogeneous Poisson process with rate $\chi(t) = \lambda( z(t), v(t) )$, where
\begin{equation}
    \lambda(z,v) = \max\left[ 0, -v \cdot \nabla_z \log\pi_\gamma(z) \right],
    \label{equ-bounce-rate}
\end{equation}
which is equal to $\left(-\frac{\dif \log\pi_\gamma(z(t))}{\dif t}\right)^+$ at $(z(t),v(t))$. Bounce events, therefore, cannot occur if $\log\pi_\gamma(z(t))$ is increasing and are expected to occur sooner the faster $\log\pi_\gamma(z(t))$ is decreasing.

For a bounce event at $(z,v)$, we update the velocity $v$ by reflecting its component in the direction of the gradient of $\log\pi_\gamma$ to
\begin{equation}
    v^\prime = v - 2 \frac{v \cdot \Tilde{\nabla}_z \log\pi_\gamma(z)}{\lVert \Tilde{\nabla}_z \log\pi_\gamma(z) \rVert^2} \Tilde{\nabla}_z \log\pi_\gamma(z),
    \label{equ-bounce-velocity}
\end{equation}
where the $\Tilde{\nabla}_z $ operator refers to the portion of the gradient which is tangent to the sphere at $z$, and is expressed by
\begin{equation}
    \Tilde{\nabla}_z  U(z) = \nabla_z U(z) - (z \cdot \nabla_z U(z))z,
    \label{equ-tangent-gradient}
\end{equation}
for an arbitrary function $U$. Geometrically, the path of the particle around bounce events is described in Figure \ref{fig-bounce-event}. We let $R(z)$ be the matrix
\begin{equation}
    R(z) = I_{d+1} - 2 \frac{\Tilde{\nabla}_z \log\pi_\gamma(z) \Tilde{\nabla}_z \log\pi_\gamma(z)^T}{\lVert \Tilde{\nabla}_z \log\pi_\gamma(z) \rVert^2},
    \label{equ-bounce-matrix}
\end{equation}
so that $v^\prime = R(z)v$ for a bounce event at $(z,v)$.

\begin{figure}[t!]
    \centering
    \begin{tikzpicture}[line cap=round,line join=round,>=Latex,x=1cm,y=1cm]
        \clip(-5.7315873015873136,1.3176719576719547) rectangle (-0.06809523809523338,5.10179894179896);
        \draw [line width=0.5pt] (0,0) circle (5cm);
        \draw [line width=0.5pt,dash pattern=on 2pt off 2pt,domain=-5.7315873015873136:-0.06809523809523338] plot(\x,{(--24.985716326530707--3.6403862911543095*\x)/3.4253910693096747});
        \draw [{Stealth[length=5pt,width=5pt]}-,line width=0.5pt,color=green] (-2.141017949815245,3.3472903390951245)-- (-3.6403862911543095,3.4253910693096747);
        \draw [-{Stealth[length=5pt,width=5pt]},line width=0.5pt,color=red] (-3.4712698412698475,1.9335449735449741)-- (-3.5611918324218124,2.726784019865405);
        \draw [line width=0.5pt,color=red] (-3.5611918324218124,2.726784019865405) -- (-3.6403862911543095,3.4253910693096747);
        \draw [->,line width=0.5pt] (-3.6403862911543095,3.4253910693096747) -- (-4.246631501746869,3.9958324357167134);
        \begin{scriptsize}
            \draw[color=black] (-1.2,4.3) node {$U(z) = \text{constant}$};
            \draw[color=black] (-4.5728571428571483,4.3) node {$-\nabla U(z)$};
        \end{scriptsize}
    \end{tikzpicture}
    \caption{Illustration of a Bounce Event. Here, the target density is taken to be of the form $\pi_\gamma(z) \propto \exp(U(z))$ for some function $U(z)$. The particle initially travels ``downhill'' along the red path before bouncing at a specific time. At the event time, the particle lies on the black contour of $U(z)$ and ``bounces'' away from it.}
    \label{fig-bounce-event}
\end{figure}

Our second type of event, refreshment events, occurs according to a homogeneous Poisson process with constant rate $\lambda_\text{ref}>0$ independently of the current state of the process. For a refreshment event at $(z,v)$, we sample $v^\prime \sim p(\cdot \mid z)$ uniformly from $\{ z\}^\perp$, as in Equation \eqref{equ-z-perp}, therefore choosing a new geodesic to follow uniformly at random. As with the Euclidean BPS, including independent refreshments is necessary for irreducibility of the algorithm, as shown in Figure \ref{fig-bounce-reducibility}. 

\begin{algorithm}[b!]
    \begin{itemize}
        \item[\textbf{Input:}] Target density $\pi_\gamma$ on $\mathbb{S}^d$, $(Z_0,V_0) \in \Sperp$, refreshment rate $\lambda_\text{ref}$.
        \item[\textbf{Output:}] $\{( Z_t, V_t )\}_{t \in [0,\infty)}$.
        \item[\textbf{Initialisation:}] Set $(z^{(0)},v^{(0)}) = (Z_0,V_0)$ and $s = 0$. 
        \item[\textbf{For}] $i = 0,1,\dots$\textbf{:}
        \item Sample $\tau_\text{ref} \sim \text{Exp}(\lambda_\text{ref})$
        \item Sample $\tau_\text{bounce}$ according to the first event of a Poisson process with rate function $t \mapsto \lambda(z(t),v(t))$, with $(z(t),v(t))$ given by Equation \eqref{equ-deterministic} initialised at $(z^{(i)},v^{(i)})$, and $\lambda(z,v)$ as in Equation \eqref{equ-bounce-rate}
        \item Set $\tau = \text{min}(\tau_\text{ref},\tau_\text{bounce})$, and $Z(t+s) = z^{(i)} \cos(t) + v^{(i)} \sin(t),\ V(t+s) = v^{(i)} \cos(t) - z^{(i)} \sin(t)$ for all $t \in [0,\tau)$
        \item Set $s = s+\tau$, $z^{(i+1)} = z^{(i)} \cos(\tau) + v^{(i)} \sin(\tau)$ and $\widehat{v} = v^{(i)} \cos(\tau) - z^{(i)} \sin(\tau)$
        \item \textbf{If} $\tau_\text{ref} < \tau_\text{bounce}$: sample $v^{(i+1)} \sim p(\cdot \mid z^{(i+1)})$
        \item \textbf{Else}: set $v^{(i+1)} = R(z^{(i+1)})\widehat{v}$ as per Equation \eqref{equ-bounce-matrix}
    \end{itemize}
    \caption{The Stereographic Bouncy Particle Sampler}\label{alg-sbps}
\end{algorithm}

\begin{figure}[t!]
    \centering
    \begin{subfigure}{2.36in}
        \centering
        \includegraphics[width=2.35in]{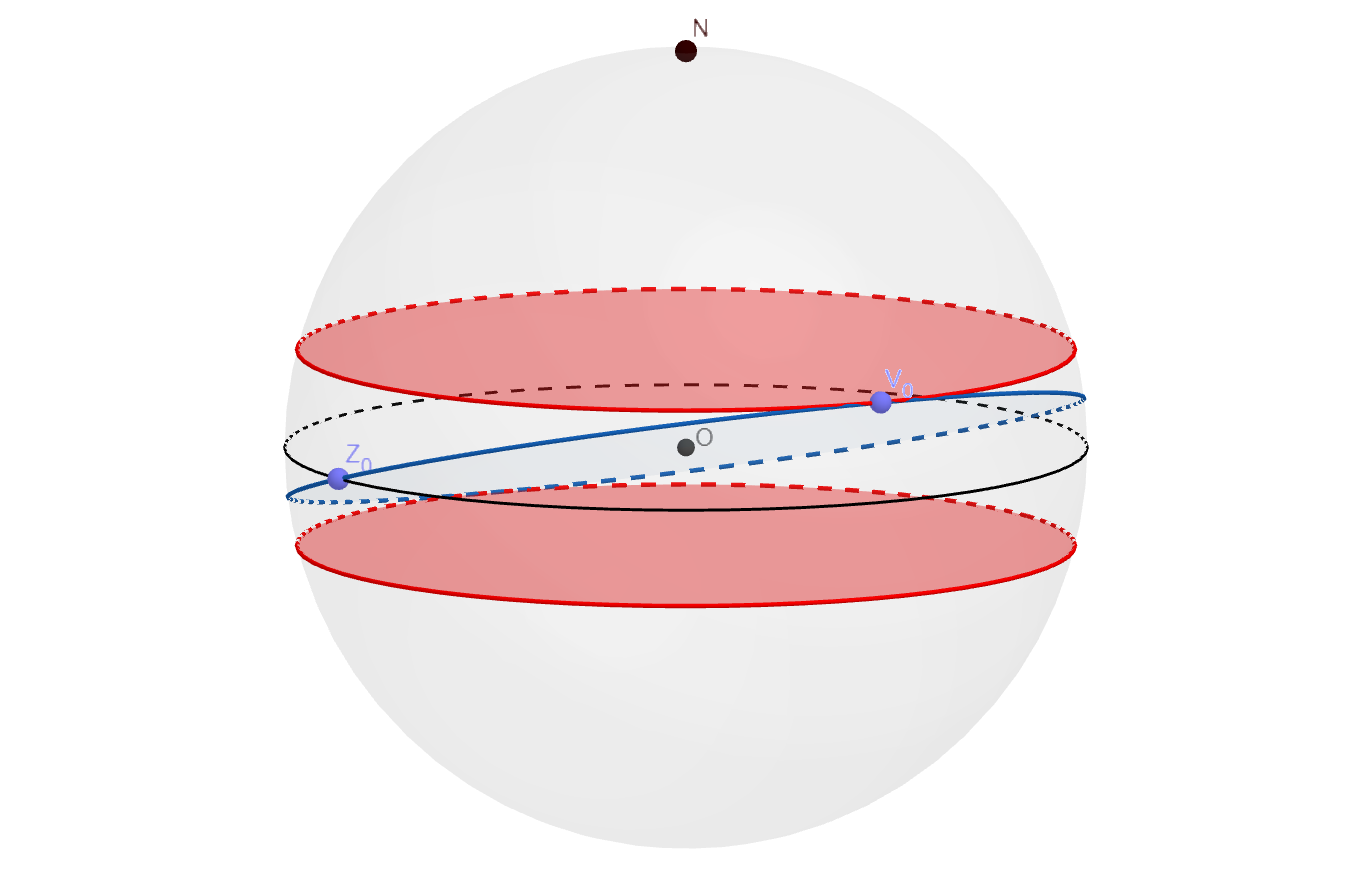}
        \caption{Any geodesic is constrained to a certain region around the equator.}
    \end{subfigure}
    \hfill
    \begin{subfigure}{2.36in}
        \centering
        \includegraphics[width=2.35in]{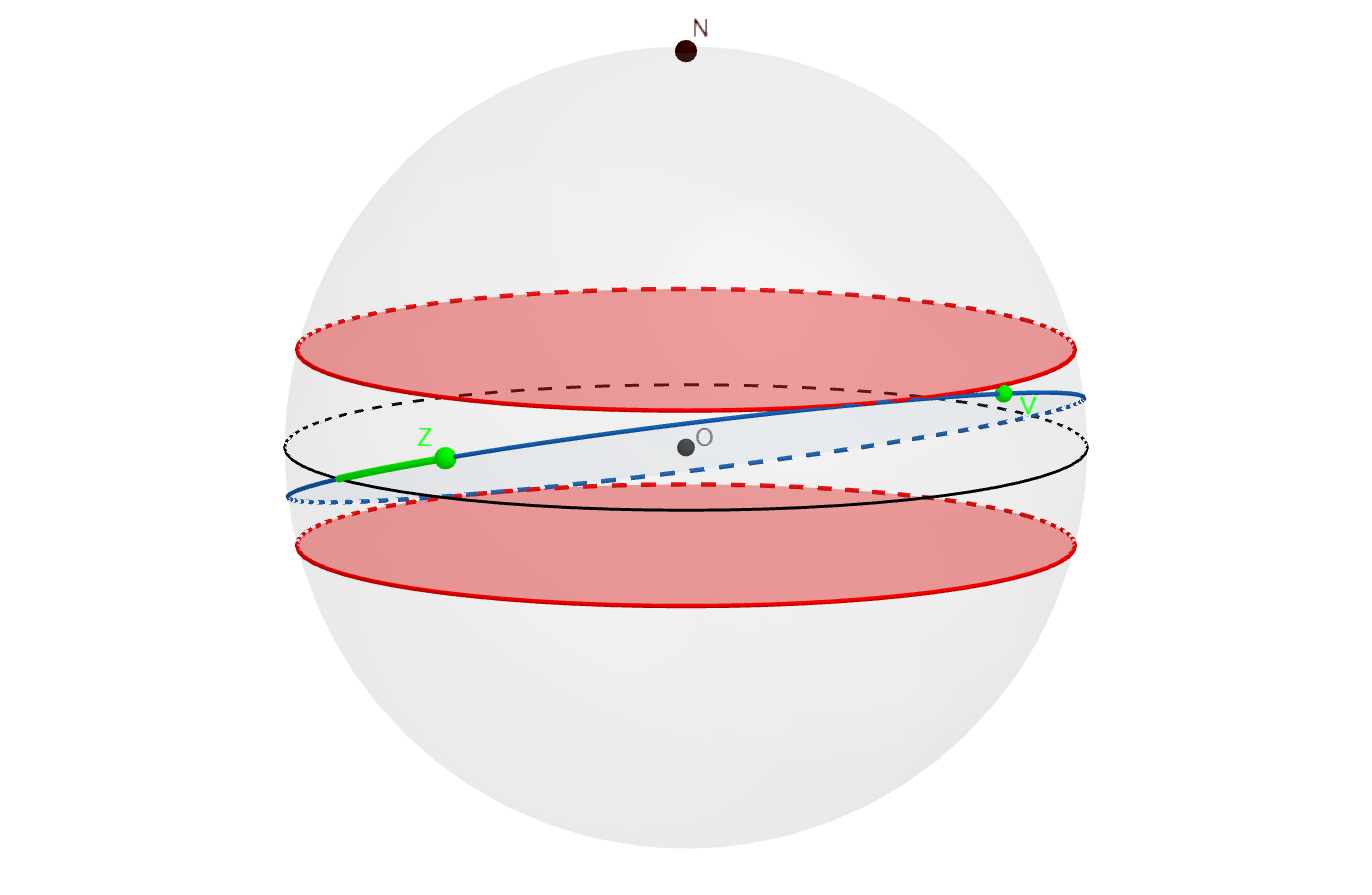}
        \caption{Running the dynamics SBPS keeps us there.}
    \end{subfigure}
    \\[10pt]
    \begin{subfigure}{2.36in}
        \centering
        \includegraphics[width=2.35in]{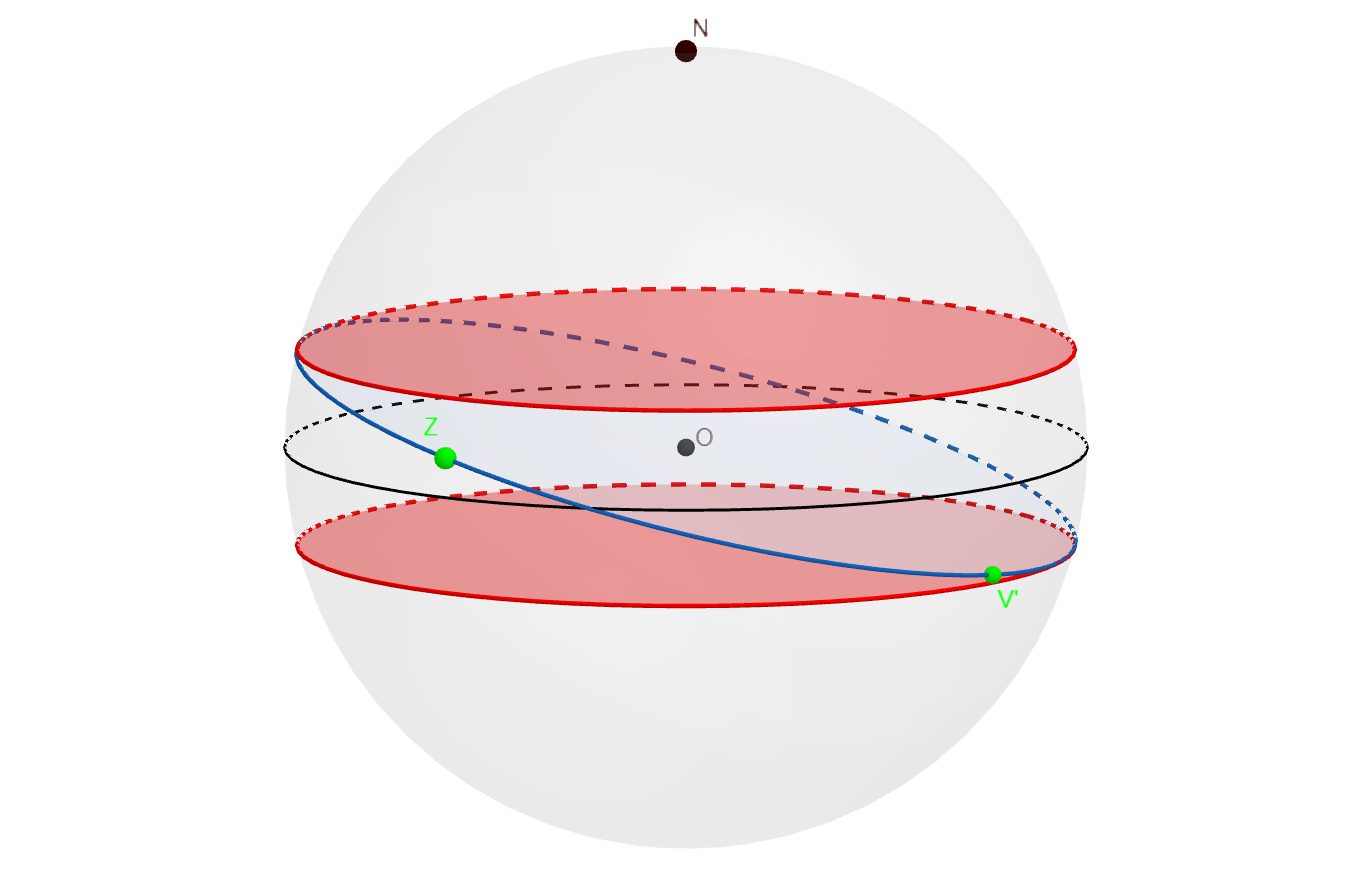}
        \caption{Bounce events cannot change these bounds on the latitude.}
    \end{subfigure}
    \hfill
    \begin{subfigure}{2.36in}
        \centering
        \includegraphics[width=2.35in]{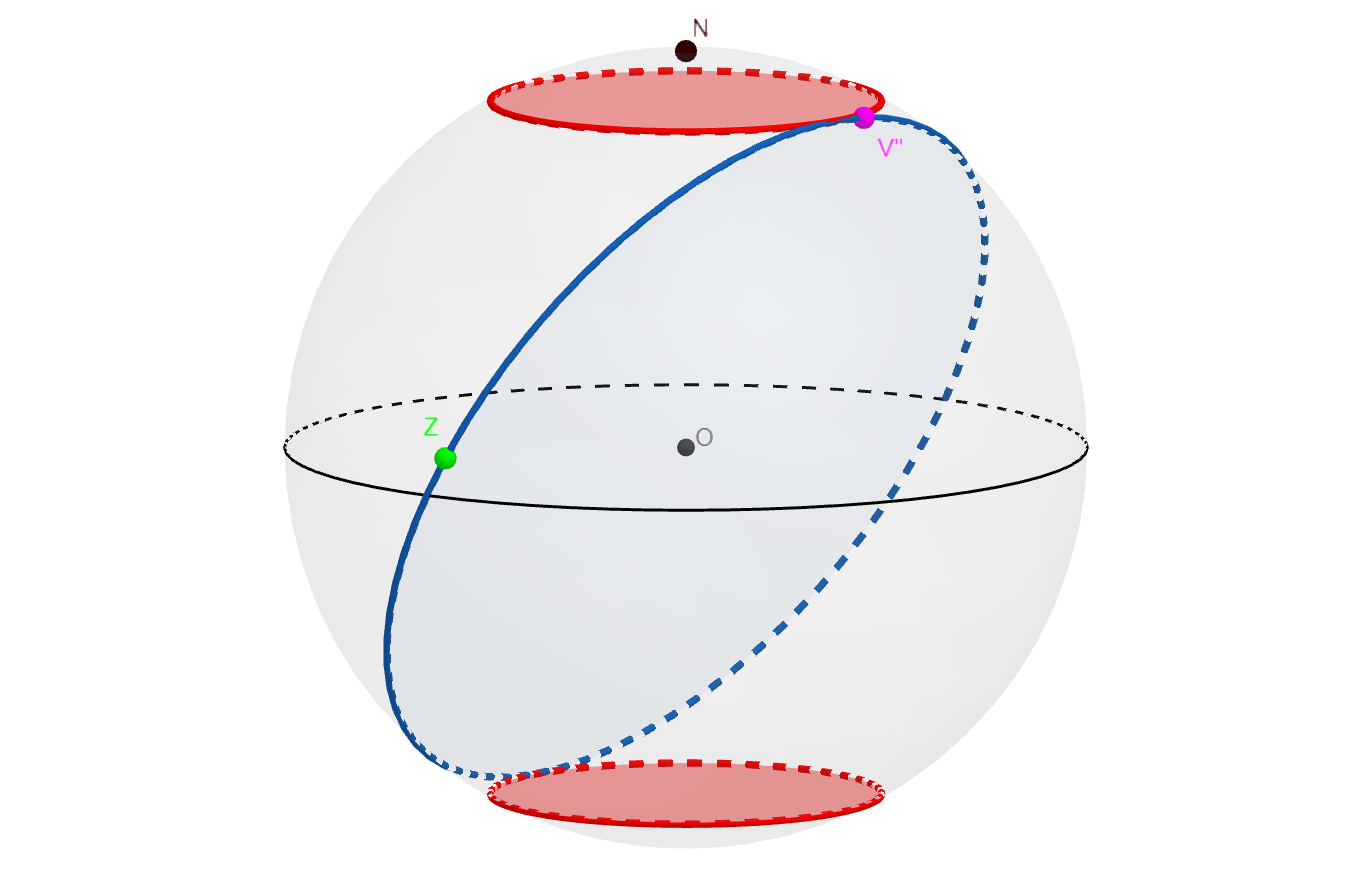}
        \caption{Refreshment events allow the process to access the entire sample space.}
    \end{subfigure}
    \caption{SBPS Bounce Events do not allow the path to leave the vicinity of the equator when targeting a spherically symmetrical distribution, such as a $\mathcal{N}(0_d,I_d)$. Refreshment events are then required to ensure irreducibility.}
    \label{fig-bounce-reducibility}
\end{figure}

We denote $(Z_t, V_t)$ the random position and velocity of the SBPS at time $t > 0$, and present the algorithm as a whole in Algorithm \ref{alg-sbps}. Note that, computationally speaking, one will not be able to store or work with a continuous sample path. One can choose to either output a skeleton of the chain, sampled at regular time intervals of short length $\delta$, or output the points at event times $\{(z^{(i)},v^{(i)})\}_{i \in \mathbb{N}}$, although it is vital to note that the positions $z^{(i)}$ at event times are not distributed according to $\pi_\gamma$ (see e.g. Equation (34.23) of \cite{davis1993markov}).

It is shown in \cite{yang2022stereographic} that the SBPS is ergodic, and its stationary distribution has density $\pi_\gamma(z) \times p(v \mid z)$ on $\Sperp$. Furthermore, the SBPS is simultaneously uniformly ergodic, even when $\pi$ is as heavy-tailed as a MtD with at least $d - 1/2$ DoF (see Lemma \ref{lemma-sbps-uniform-ergodicity}).

\section{Adaptive Stereographic Algorithms}
\label{section-adaptive-stereographic}

Each of the above algorithms is parametrised by $\gamma$, the parameters of the stereographic projection, and potentially additional parameters determining the dynamics of the Markov kernel on $\mathbb{S}^d$ ($h$ for the SRW and $\lambda_\text{ref}$ for the SBPS). The choices of these parameters can heavily impact algorithmic performance. In this section, we establish motivation for optimal choices of $\gamma$ and introduce our adaptive versions of each of the 3 algorithms. We then find sufficient conditions to prove that estimators produced by the adaptive algorithms satisfy a SLLN, $\mathcal{L}_2$ convergence, and a CLT.

\subsection{The Equator as a High-Probability Region}
\label{section-equator}

We will be interested in varying $\gamma$ in order to optimally position and scale the sphere to match properties of the target distribution~$\pi$. As discussed at the start of Section \ref{section-stereographic-mcmc}, this is equivalent to tuning the parameters of the affine preconditioning, which we perform on $\mathbb{R}^d$ before projecting onto $\mathbb{S}^d$. 

To motivate the optimal choice of parameters, consider $X \sim \prod_{i=1}^d f(x_i)$ an iid product with $\mathbb{E}_f(X) = 0$ and $\mathbb{E}_f(X^2) = 1$. Chebyshev's inequality then gives
\begin{equation}
    \frac{1}{d} \sum_{i=1}^d X_i^2 = 1 + \mathcal{O}_{\mathbb{P}}(d^{-\frac{1}{2}}).
    \label{equ-radius-big-o}
\end{equation}
Geometrically, as $d$ increases, we see that the distribution of $X$ will become more and more concentrated around a spherical shell of radius $\sqrt{d}$, as depicted in Figure \ref{fig-high-dim-equator}.

\begin{figure}[t!]
    \centering
    \includegraphics[width=2.6in]{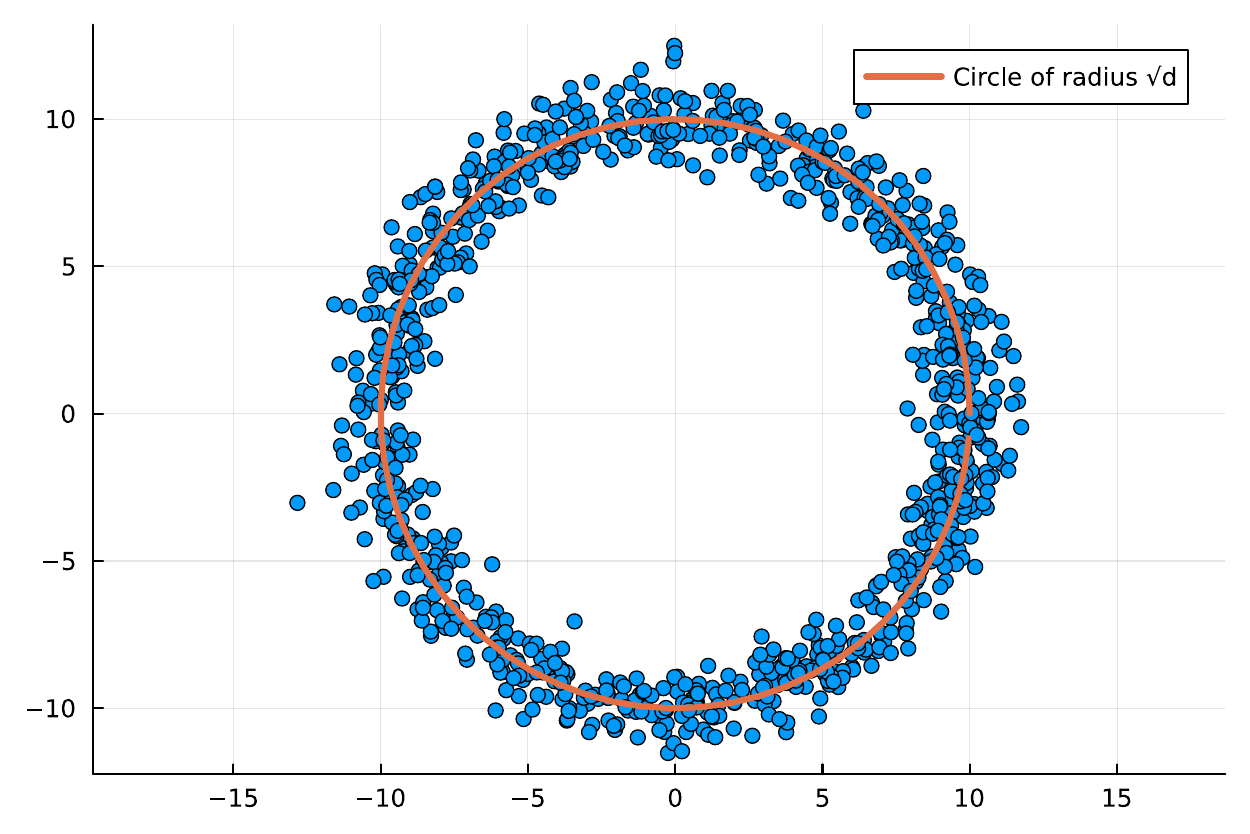}
    \caption{Visualisation of the norm of a $\mathcal{N}(0_d, I_d)$ distribution for $d = 100$. We project the points onto a 2-dimensional subspace to clearly show that points are likely to be at distance $\sqrt{d}$ from the origin. In high dimensions, the circle is in reality a hyperspherical shell of radius $\sqrt{d}$.}
    \label{fig-high-dim-equator}
\end{figure}

As a consequence, taking $\mu = 0_d$ and $\Sigma = d I_d$ gives
\begin{equation}
    Z_{d+1} = \frac{\frac{1}{d}\lVert X \rVert^2 - 1}{\frac{1}{d}\lVert X \rVert^2 + 1} = \mathcal{O}_\mathbb{P}(d^{-\frac{1}{2}}),
    \label{equ-latitude-big-o}
\end{equation}
so that our sample points under $\pi_\gamma$ will become concentrated around the equator. However, if $\gamma$ is chosen poorly, the density is likely to become concentrated around one of the poles, as shown in Figure \ref{fig-latitude-target}.

\begin{figure}[t!]
    \centering
    \begin{subfigure}{5in}
        \centering
        \includegraphics[width=2.35in]{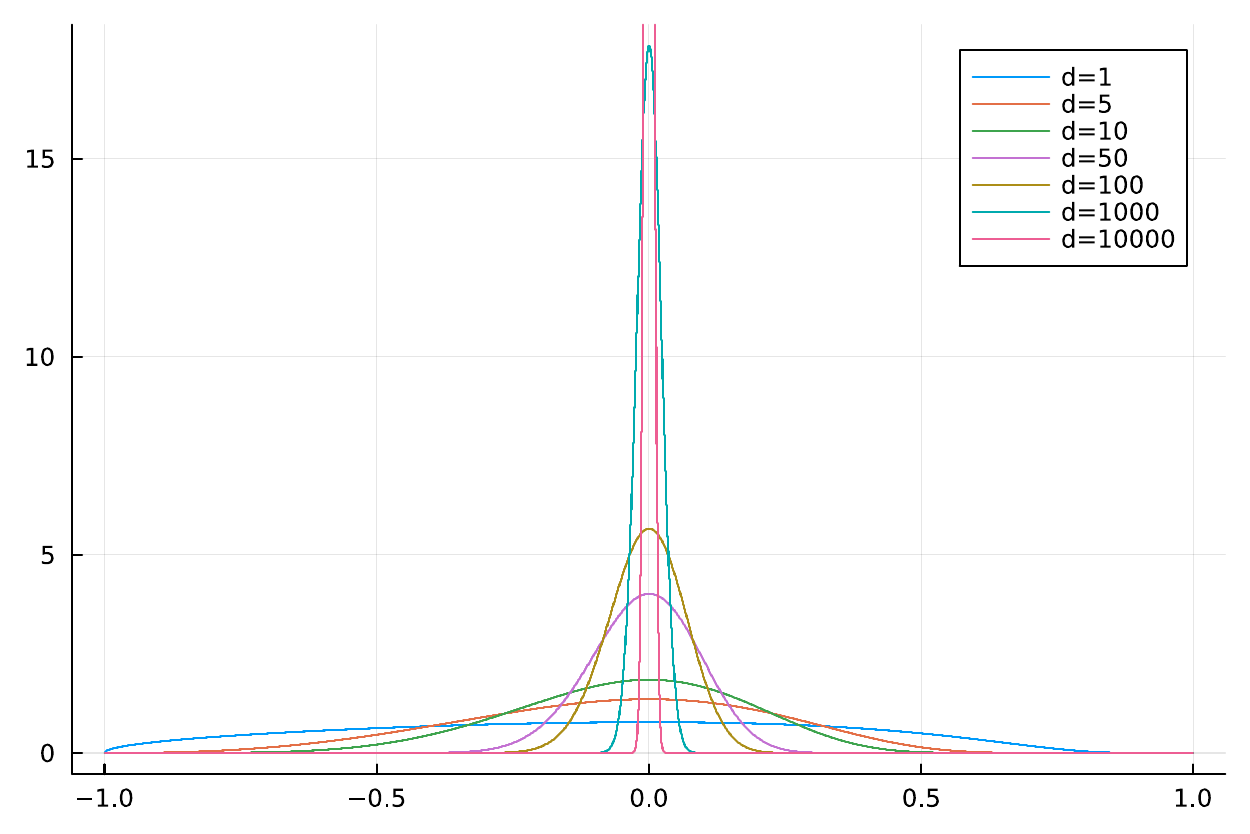}
        \caption{Taking $\Sigma = dI_d$, we see that $\pi_\gamma$ becomes concentrated around the equator $z_{d+1}=0$.}
    \end{subfigure}
    \\[10pt]
    \begin{subfigure}{2.36in}
        \centering
        \includegraphics[width=2.35in]{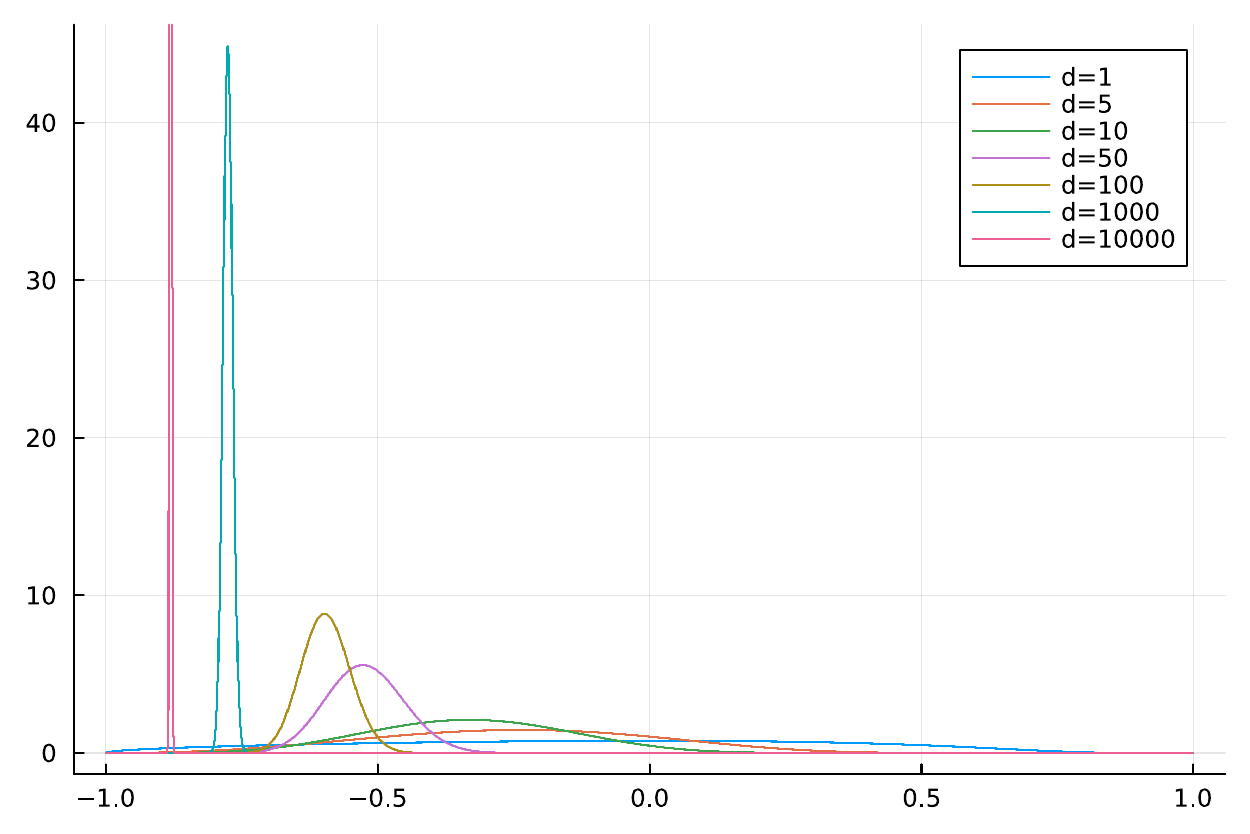}
        \caption{Taking $\Sigma = d^{1.3}I_d$, we see that the mass moves towards the South Pole. This is because the radius of the sphere we are using in the projection is too large.}
    \end{subfigure}
    \hfill
    \begin{subfigure}{2.36in}
        \centering
        \includegraphics[width=2.35in]{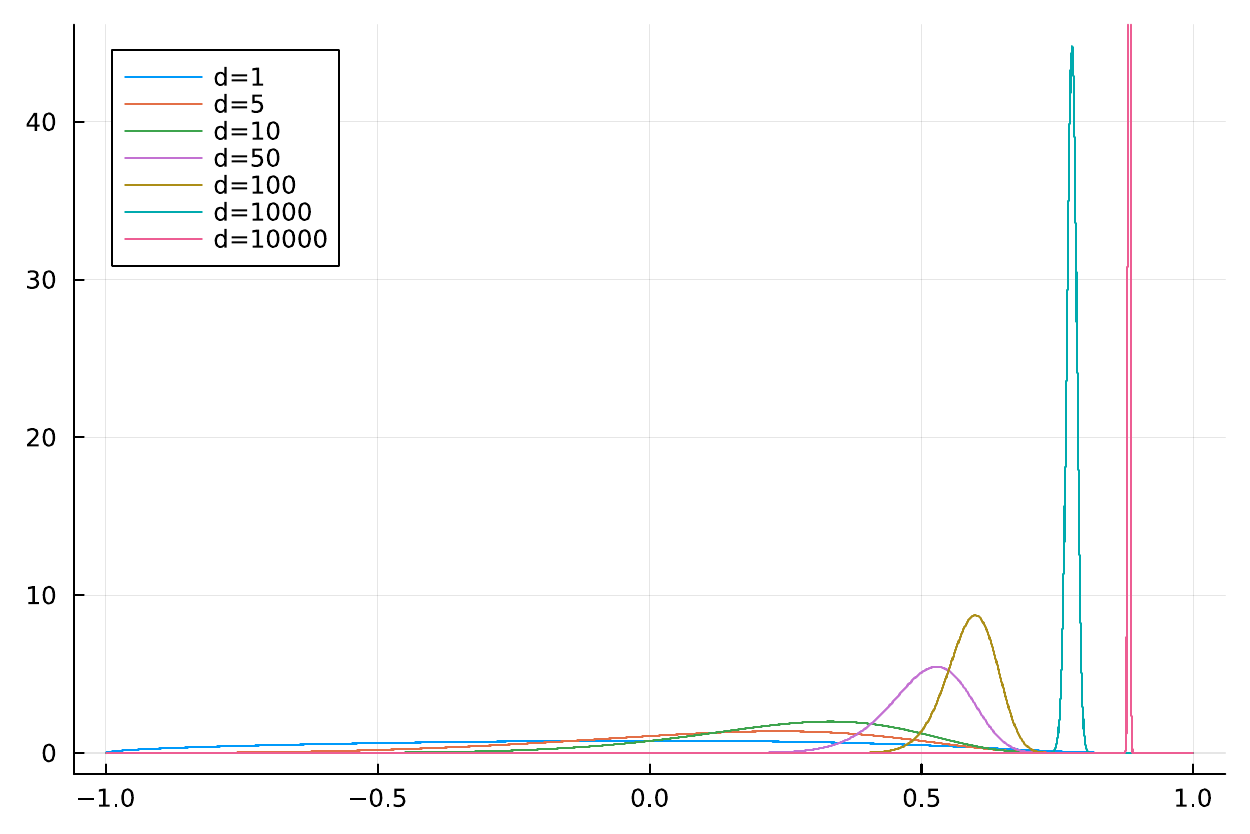}
        \caption{Taking $\Sigma = d^{0.7}I_d$, we see that the mass moves towards the North Pole. This is because the radius of the sphere we are using in the projection is too small.}
    \end{subfigure}
    \caption{Plots of the marginal density of $Z_{d+1}$ under $\pi_\gamma$, when $X \sim \mathcal{N}(0_d,I_d)$. We fix $\mu = 0$ and vary $\Sigma$.}
    \label{fig-latitude-target}
\end{figure}

We can therefore construct our algorithms to efficiently explore the equator, and tune our parameters to precondition the target by taking
\begin{equation}
    \mu = \mathbb{E}_\pi(X), \qquad \Sigma = d \times \mathbb{E}_\pi\left( (X-\mu)(X-\mu)^T \right).
    \label{equ-optimal-parameters}
\end{equation}
Geometrically, if we consider $\mu$ and $\Sigma$ the centre and shape of the sphere used in the stereographic projection, we are choosing our parameters so that the equator of the sphere intersects with $\mathbb{R}^d$ along the spherical shell of high probability depicted in Figure \ref{fig-high-dim-equator}.

These parameters are optimal as $d \rightarrow \infty$, as shown in Section 5.3 of \cite{yang2022stereographic}, but it may not be the case for finite $d$ that $\mathbb{E}_{\pi_\gamma}(Z_{d+1}) = 0$ when taking $\gamma$ as in Equation \eqref{equ-optimal-parameters}. For example, for a MtD with $d$ DoF, we want $\Sigma = d I_d = (d-2)\text{Var}_\pi(X)$ to get a uniform distribution on the sphere and a target which is spread evenly around the equator. We therefore take
\begin{equation}
    \Sigma = c \times \mathbb{E}_\pi\left( (X-\mu)(X-\mu)^T \right)
\end{equation}
for $c > 0$ such that $\mathbb{E}_{\pi_\gamma}(Z_{d+1}) \approx 0$.

Each of our algorithms can be seen to have better mixing properties when $\gamma$ is optimally chosen. For the SRW, the geometry of $\mathbb{S}^d$ causes the proposed moves to naturally stay on the equator if that is where the chain lies and to drift back towards the equator if the chain is currently near either pole. Section 5 of \cite{yang2022stereographic} presents several results showing that if $\pi$ is spherically symmetric, then the SRW can be superefficient when compared to the Euclidean RWM. It is shown that the acceptance probability goes to 1 as $d \rightarrow \infty$, even for a constant step size $h$. Since moves around the equator are then projected back onto $\mathbb{R}^d$ by a factor of roughly $\sqrt{d}$ via $\Sigma$, we obtain a ``blessing of dimensionality'' where the expected jump distance of the SRW in Euclidean space will be $\mathcal{O}(d^{1/2})$ (i.e.\ mixing improves as $d$ increases). By comparison, optimal Euclidean RWM steps will typically be $\mathcal{O}(d^{-1/2})$ as $d \rightarrow \infty$.

We give sketch arguments for the blessing of dimensionality in the case of the SSS or the SBPS. Though these are not rigorous statements, they give intuition as to how the algorithms work so well. Both work from proposing moves along geodesics, which are more and more likely to remain in close proximity to the equator as $d$ increases, and will always intersect the equator regardless of the current position. Thus, the algorithms will naturally stay near the equator if that is where they are, and will attempt to return to the equator in $\mathcal{O}(1)$ time even if the process is currently at one of the poles. If $\pi$ is spherically symmetric and $\gamma$ is appropriately chosen, we will therefore expect to observe the same blessing of dimensionality as we get for the SRW, where the expected distance travelled per step/unit time will also be $\mathcal{O}(d^{1/2})$.

If the parameters are poorly chosen and $\pi_\gamma$ is focused on a small subset of $\mathbb{S}^d$, we can scale the stepsize of the SRW accordingly, and the resulting algorithm will never perform worse than its Euclidean counterpart (see Corollary 5.1 of \cite{yang2022stereographic}). For the SSS and SBPS, however, one must additionally consider the consequences of poor parametrisation on the computational cost of simulating a step: the SBPS will require many bounce events to fight the drift back towards the equator, and the rejection sampling step in the SSS will reject many moves at every step before narrowing down the search interval. As a consequence, not only does each step provide slower convergence to stationarity, but they are also increasingly expensive as the parametrisation worsens. It is therefore all the more important for the parameter $\gamma$ to be chosen appropriately when using these algorithms.

Since we do not know the expectations in Equation \eqref{equ-optimal-parameters} in advance, it is of interest to create adaptive versions of any algorithm using the stereographic projection in order to automatically tune $\gamma$. It is worth mentioning that adaptations can also be performed for the other parameters, $h$ and $\lambda_\text{ref}$. In the case of the SRW, we have already discussed that if $\pi$ is spherically symmetric and $\gamma$ is chosen correctly, then for any step size $h$ the acceptance probability goes to 1. We therefore want to take $h$ as large as possible. \cite{yang2022stereographic} also prove that if $\pi$ is a non-Gaussian iid distribution, and $\gamma$ is chosen according to Equation \eqref{equ-optimal-parameters}, then it is optimal to tune $h$ to be $\mathcal{O}(d^{-1})$ such that we have an average acceptance rate of $0.234$, as is the case for the Euclidean RWM algorithm \cite{gelman1997weak}. It is therefore always appropriate to tune $h$ to achieve an average acceptance rate of $0.234$. Even if $\gamma$ is poorly chosen, this will lead to a target distribution focused on a small corner of $\mathbb{S}^d$, which locally behaves like $\mathbb{R}^d$, so we still want to aim for the same acceptance rate.

The optimal refreshment rate for the SBPS is an open problem, but it is noteworthy that \cite{bierkens2022high} shows that, for the Euclidean BPS, it is optimal in certain settings to have $78.1\%$ of events be refreshment events. However, in a spherically symmetrical setting where we have very few bounces, this may lead to very small refreshment rates, which could hurt the irreducibility of the process. We therefore recommend never taking a refreshment rate lower than $\frac{1}{\pi}$, mimicking the No-U-Turn sampler's intuition \cite{MR3214779} of following the dynamics until we ``turn around'' at the other end of the geodesic.

\subsection{The Algorithms}

We now construct adaptive versions of the stereographic algorithms which automatically tune the parameters as they run, in order to improve performance. We established in Equation \eqref{equ-optimal-parameters} what values we will be targeting with our estimators. However, as is the case for all adaptive MCMC algorithms, changing the transition scheme based on the full history of the sample path causes the Markov property to fail.

There is a vast literature for the construction and analysis of discrete-time adaptive MCMC algorithms with desirable asymptotic properties \cite{MR2260070, roberts2007coupling, MR2461882, roberts2009examples, MR2759732, MR3012408, chimisov2018adapting}. However, the literature on designing and studying continuous-time adaptive MCMC algorithms is lacking and the only such approach, \cite{MR4474548}, discretises the sample path to establish ergodicity.

In this paper, we use the Adapting Increasingly Rarely (AIR) MCMC framework originally proposed in \cite{chimisov2018air}, and show how it can be applied to continuous-time adaptive processes. Crucially, this makes the setup identical for each of our algorithms.

We assume the lags $t_k$ between adaptations are polynomially increasing, i.e.\ $\exists \beta >0$ and $c \geq 1$ such that
\begin{equation}
    \frac{1}{c} k^\beta \leq t_k \leq c k^\beta,
    \label{equ-poly-lags}
\end{equation}
and define the adaptation times $T_k = \sum_{i=1}^k t_i$ with $T_0 = t_0 = 0$.

We proceed by running the process with the parameter $\gamma_k$ fixed for $t \in [T_k, T_{k+1})$. At each time $T_{k+1}$, we update the parameter to $\gamma_{k+1}$ based on the sample path so far and the previous parameter values. Algorithm \ref{alg-air-srw} outlines the framework in the case of the SRW, with the corresponding algorithms for the SSS and SBPS being essentially identical, but replacing the SRW with their respective algorithms.

\begin{algorithm}[b!]
    \begin{itemize}
        \item[\textbf{Input:}] Target density $\pi$ on $\mathbb{R}^d$, $X^{(0)}_0 \in \mathbb{R}^d$, initial parameters $\gamma_0 \in \Gamma$, $h_0>0$, sequence of adaptation lags $\{ t_k \}_{k \in \mathbb{N}}$.
        \item[\textbf{Output:}] $\left\{\left( Z^{(k)}_n, X^{(k)}_n \right)\right\}_{0 \leq n \leq t_k}$, $\gamma_k$ and $h_k$, for $k \in \mathbb{N}$.
        \item[\textbf{For}] $k = 0,1,\dots$\textbf{:}
        \item Run the SRW (Algorithm \ref{alg-srw}) for $t_{k+1}$ time units to get
        \[\left\{\left( Z^{(k+1)}_n, X^{(k+1)}_n \right)\right\}_{0 \leq n \leq t_{k+1}} \sim \text{SRW}\left(\pi,X^{(k)}_{t_k},\gamma_k,h_k \right)\]
        \item Update parameters to $\gamma_{k+1}, h_{k+1}$ using $\left\{\left( Z^{(k+1)}_n, X^{(k+1)}_n \right)\right\}_{0 \leq n \leq t_{k+1}}$
    \end{itemize}
    \caption{AIR SRW}\label{alg-air-srw}
\end{algorithm}

These AIR schemes have multiple benefits over the more commonly used adaptation schemes of adapting every step, or every $k$ steps:
\begin{itemize}
    \item Practically speaking, calculating the new adaptive parameters can be expensive, and reducing the frequency of the updates can have little impact on the mixing of the chain. For the stereographic MCMC algorithms, we will be trying to obtain an estimator for the square root of the covariance matrix to feed into our algorithms. Even using a rank-one update for the Cholesky decomposition (as discussed in \cite{igel2006computational} or the ``cholupdate'' function in Matlab), this operation adds $\mathcal{O}(d^2)$ in computational cost every time we adapt. In high dimensions, this can become very costly, so it is sensible to adapt more rarely as the process goes on, and we expect the estimators to converge.
    \item Theoretically speaking, adaptive MCMC algorithms can be tricky to analyse because using the history of the chain to inform future transitions causes the Markov property to break down. By keeping the parameters fixed for increasing lengths of time, we obtain a sequence of epochs that conditionally behave like standard Markov chains, and these sample paths become easier to control the longer they run.
\end{itemize}

\subsection{Asymptotic Results}
\label{section-results}

We now present our main results for the convergence of estimators using the adaptive stereographic algorithms. Given any bounded function $f:\mathbb{R}^d \rightarrow \mathbb{R}$, we want to estimate $\pi (f) = \mathbb{E}_\pi(f(X))$. We are interested in the behaviour of the estimators
\begin{equation}
    \widehat{f}^\text{disc}_t = \frac{1}{t} \sum_{s=0}^{t-1} f(X_s), \qquad
    \widehat{f}^\text{cont}_t = \frac{1}{t} \int_{s=0}^t f(X_s) \,\dif s,
    \label{equ-estimator}
\end{equation}
depending on whether we are considering a discrete or continuous-time algorithm. Note also that for the sake of computation, we will generally evaluate estimators from SBPS sample paths by taking a discrete skeleton with a small mesh size. 

With this definition, it is clear why we need to move the path off of $\mathbb{S}^d$: since $f$ is fixed on $\mathbb{R}^d$, the target function on $\mathbb{S}^d$ would be
\begin{equation}
    f_\gamma(z) = f(\text{SP}^{-1}_\gamma(x)),
    \label{equ-f-gamma}
\end{equation}
which changes as we update $\gamma$.

To obtain convergence results for the AIR stereographic algorithms, we make several assumptions.
\begin{assumption}
    The adaptation scheme for the parameter updates keeps $\gamma_i = (\mu_i, \Sigma_i)$ in the compact set $\Gamma_{r,R}$ for all $i$, where
    \begin{equation*}
        \Gamma_{r,R} = \{ \gamma \in \Gamma : \lVert \mu \rVert \leq R, r^2 \leq \rho_i \leq R^2, \forall i \},
    \end{equation*}
    for $0 < r < R$, and $\rho_i$ are the eigenvalues of $\Sigma$.

    For the AIR SRW, also assume $r \leq h_i \leq R$ for each $h_i$. For the AIR SBPS, also assume $r \leq \lambda_i \leq R$ for each $\lambda_i$.
    \label{assump-parameter-space-compact}
\end{assumption}
This assumption restricts the parameters to a compact set, so that even if the adaptation scheme ``goes wrong'', the parameters cannot cause arbitrarily poor mixing. A condition along these lines is used in almost every adaptive MCMC theoretical result (see e.g.\ for notable examples \cite{roberts2007coupling,roberts2009examples,haario2001adaptive}).

As a benefit of the AIR framework, we can use any estimator for our parameters that respects Assumption \ref{assump-parameter-space-compact}, and will not require trickier conditions such as diminishing adaptations for our $\mathcal{L}^2$ or almost sure convergence.

\begin{assumption}
    The adaptation lags $t_k = \Theta(k^\beta)$ are chosen under mild conditions.
    \label{assump-adaptation-times}
\end{assumption}
This assumption allows the AIR setup to be used. The mild additional conditions arise through our method of proving the results and are simply conditions on the expression for the times $t_k$. These are of very little practical relevance, and we conjecture that simply assuming $t_k = \Theta(k^\beta)$ is sufficient. See Appendix \ref{section-times-assumptions} for more details.

For our assumptions on the target distribution, we will need a slightly different assumption for the SRW and SSS than the assumption for the SBPS. For the SRW and SSS, we need the following:
\begin{assumption}
    The target density $\pi$ is positive, continuous, and satisfies
    \begin{equation*}
        \limsup_{\lVert x\rVert \rightarrow \infty}  \left( \pi(x) (\lVert x \rVert^2 + 1)^d \right) < \infty.
    \end{equation*}
    
    \label{assump-boundedness-condition}
\end{assumption}
This assumption ensures that $\pi_\gamma$ is bounded over $\mathbb{S}^d$. Our Markov chains are then targeting a distribution with a bounded density and a compact support, leading to uniform ergodicity (see Lemmas \ref{lemma-srw-uniform-ergodicity} and \ref{lemma-sss-uniform-ergodicity}). If this condition fails, the chains could get stuck in the vicinity of the North Pole for arbitrarily long times.

Note that this condition is satisfied for distributions with relatively heavy tails, such as MtD with at least $d$ DoF.

Since the SBPS works off of $\nabla_z \log\pi_\gamma$, we need a different condition:
\begin{assumption}
    The target density $\pi$ is positive, continuously differentiable and satisfies
    \begin{equation*}
        \limsup_{\lVert x\rVert \rightarrow \infty}  \left( x \cdot \nabla_x \log\pi(x) + R \lVert \nabla_x \log\pi(x) \rVert \right) + 2d < \frac{1}{2},
    \end{equation*}
    for some $R>0$.
    \label{assump-derivatives-condition}
\end{assumption}
This assumption also ensures that the North Pole does not become an inescapable singularity, and is sufficient for uniform ergodicity of the SBPS (see Lemma \ref{lemma-sbps-uniform-ergodicity}). It is also satisfied for distributions with relatively heavy tails, such as MtD with more than $d-\frac{1}{2}$ DoF. The fact that the SBPS can handle an ``extra $\frac{1}{2}$ DoF'' compared to the discrete algorithms covered by Assumption \ref{assump-boundedness-condition} is quite surprising. One can interpret this discrepancy as a sign that the deterministic drift of the SBPS towards the equator can allow it to escape a singularity that the SRW or SSS do not have enough drift to quickly escape.

Although these are sufficient conditions to ensure good theoretical convergence properties, we do not require them for the algorithms to still perform well, and significantly outperform non-stereographic counterparts. Indeed, neither Assumption \ref{assump-boundedness-condition} nor Assumption \ref{assump-derivatives-condition} holds in the example discussed in Section \ref{section-robustness}, in which we see that the stereographic methods significantly outperform HMC.

With these assumptions, we can obtain the following result on the behaviour of $\widehat{f}_t$.
\begin{theorem}[Asymptotics of $\widehat{f}_t$]
    Consider either of the estimators $\widehat{f}_t$ as given in Equation \eqref{equ-estimator} for a bounded function $f:\mathbb{R}^d \rightarrow \mathbb{R}$, applied to the AIR SRW, SSS, SBPS, or any discrete skeleton of the SBPS, targeting the distribution $\pi$ on $\mathbb{R}^d$.
    
    Under Assumption \ref{assump-parameter-space-compact} for any $r<R \in (0,\infty)$, Assumption \ref{assump-adaptation-times} for $\beta >0$, and either Assumption \ref{assump-boundedness-condition} if we are working with the SRW or SSS, or Assumption \ref{assump-derivatives-condition} for the SBPS, then for any initial position or parameters:
    \begin{itemize}
        \item For any $\beta >0$, and $0 \leq \epsilon < \min\left(\frac{1}{2},\frac{\beta}{1+\beta}\right)$, then as $t \rightarrow \infty$,
        \begin{equation*}
            t^\epsilon\left( \widehat{f}_t - \pi(f) \right) \rightarrow 0, \qquad \text{a.s.\ and in } \mathcal{L}^2.
        \end{equation*}
        In particular, a SLLN holds (by taking $\epsilon = 0$).
        \item For $\beta>1$, if $\gamma_i \xrightarrow{P} \gamma_\infty$, for some constant $\gamma_\infty$, such that the asymptotic variance $\sigma^2(\gamma_\infty)>0$, and that $\sigma^2(\gamma)$ is a continuous function of $\gamma$, a CLT holds, i.e.\ as $t \rightarrow \infty$,
        \begin{equation*}
            \sqrt{t} \left(\widehat{f}_t - \pi(f) \right) \rightarrow \mathcal{N}\left(0, \sigma^2(\gamma_\infty) \right), \qquad \text{in distribution}.
        \end{equation*}
    \end{itemize}
    
    \label{thm-fhat-asymptotics}
\end{theorem}

This theorem combines elements from \cite{chimisov2018air} and \cite{hofstadler2024almost} to obtain the best of both results.

The asymptotic variance $\sigma^2(\gamma)$ is exactly the same as the asymptotic variance of a non-adaptive version of the algorithm. Thus, if the parameters do indeed converge to optimal values, we retrieve the optimal asymptotic variance for our estimators.

Note also that although we express this in a unified way, the asymptotic variance $\sigma^2(\gamma)$ will depend on the choice of algorithm. Additionally, convergence of $\gamma_i$ to some constant $\gamma_\infty$ is given by part (a) of the theorem if we are using bounded empirical estimates for the optimal values of $\mu$ and $\Sigma$, as guided by Equation \eqref{equ-optimal-parameters}. If one uses standard empirical estimators for $(\mu^*, \Sigma^*)$, the a.s.\ and $\mathcal{L}^2$ convergence of those parameter estimators is also given by the theorem. Thus, the adaptive stereographic algorithms can find the optimal parameters, then obtain the same mixing rates as the optimised non-adaptive versions. 

However, it is much harder to guarantee convergence of estimators adapting $h$ or $\lambda_\text{ref}$, as these do not typically take the form of empirical averages.

Comparing this result to other similar results, such as those discussed in \cite{roberts2007coupling}, we see that despite significantly weaker assumptions on the adaptation scheme, we obtain almost sure and $\mathcal{L}^2$ convergence. The additional requirement on the convergence of the parameter and continuity of the asymptotic variance in the CLT result is comparable to those found in any similar result. We do not present a result on the ergodicity of the algorithm, i.e.\ whether $X_t \xrightarrow{D} \pi$ as $t \rightarrow \infty$, as these would bring identical conditions to more general results already found in the literature \cite{roberts2007coupling}.

Although the sampling mechanisms for our three algorithms are very different, we obtain essentially identical convergence results. This is an artefact of our proof, in which we use the simultaneous uniform ergodicity of the processes to create a unifying auxiliary Markov chain which is much easier to work with than any of our original processes.

\section{The Segment Chain for Uniformly Ergodic Markov Processes}
\label{section-segment-chain}

We have already discussed the AIR framework and its intuitive appeal. However, the theory presented in \cite{chimisov2018air} assumes a discrete-time chain, with a 1-step small set condition and simultaneous geometric ergodicity. Instead, the stereographic algorithms are all simultaneously uniformly ergodic, the SBPS lives in continuous time, and we do not have a drift towards a 1-step small set for the SRW. \cite{hofstadler2024almost} does give a.s.\ convergence results for the simultaneously uniformly ergodic case for discrete-time algorithms, but also does not discuss either $\mathcal{L}^2$ convergence or a CLT.

In this section, we present a novel auxiliary process, the segment chain, which gives a unified framework for analysing any uniformly ergodic Markov process. This framework is particularly useful for the study of continuous-time chains, where notions of splitting, excursions, and regenerations are not as well studied. We use it to prove the asymptotic results in Theorem \ref{thm-fhat-asymptotics}. To ensure that the assumption of uniform ergodicity is satisfied, we give explicit uniform ergodicity results for our stereographic algorithms in Section \ref{section-uniform-ergodicity}.

\subsection{A Markov chain in the space of paths}

Consider a continuous-time Markov process $\{ X_t \}_{t \geq 0}$ on a sample space $\chi$ and its parametrised associated transition semigroup $\{ P^t_\gamma \}_{t \geq 0}$, where $\gamma \in \Gamma$ is a parameter. Let $ \pi_\gamma$ be the stationary measure of $P_\gamma$ (these need not be equal). Unfortunately, this notation clashes with the $\pi_\gamma$ from the stereographic projection, and these two do not necessarily equate in our setting. For example, for the SBPS, this stationary distribution would be the joint distribution of $X$ in $\mathbb{R}^d$ and some latent velocity component whose distribution depends on $\gamma$. More on this in Section \ref{section-uniform-ergodicity}.

Suppose the process satisfies a simultaneous minorisation condition of the form
\begin{equation}
    P^T_\gamma(x, \cdot) \geq \epsilon \nu(\cdot), \qquad \forall x \in \chi, \gamma \in \Gamma,
    \label{equ-minorisation-general}
\end{equation}
with $T>0$, $\epsilon > 0$, and $\nu$ a probability measure on $\chi$, all independent of $\gamma$. This implies that regardless of the values of $X_t$ or $\gamma$, we have probability $\epsilon$ to get $X_{t+T} \sim \nu$ independently of $\{X_s\}_{0\leq s \leq t}$. Such a condition is equivalent to the following, more traditional expression for uniform ergodicity of the process
\begin{equation}
    \lVert P^t_\gamma(x, \cdot) - \pi \rVert_\text{TV} \leq C \rho^t, \qquad \forall x \in \chi, \gamma \in \Gamma,
\end{equation}
where $t >0$, and $C>0$, $\rho \in (0,1)$ are constants independent of $\gamma$.

We then define the segment chain $\{ \Phi_n \}_{n \in \mathbb{N}}$ to be a discrete-time Markov chain with $\Phi_n~:~[0,T] \rightarrow \chi$ such that
\begin{equation}
    \Phi_n(t) = X_{nT + t}.
    \label{equ-segment-chain}
\end{equation}

The segment chain $\{\Phi_n \}_{n \in \mathbb{N}}$ is then a Markov chain in the space of functions from $[0,T]$ to $\chi$, which we shall simply call $\Omega$.

If we have a discrete-time Markov chain satisfying the minorisation condition \eqref{equ-minorisation-general}, we can instead take
\begin{equation}
    \Phi_n(t) = X_{nT + \lfloor t \rfloor},
    \label{equ-segment-chain-discrete}
\end{equation}
to obtain the segment chain.

The crucial observation is that, using the minorisation condition \eqref{equ-minorisation-general}, we can ``split'' the chain $\{ \Phi_n \}_{n \in \mathbb{N}}$ such that, in a way we shall make rigorous later, for each $n$, with probability $\epsilon$ we have $\Phi_{n+1}(0) \sim \nu$ independently of $\Phi_0, \dots, \Phi_{n-1}$. This will allow us to divide sample paths into weakly dependent, identically distributed blocks, to which we can then apply standard techniques to get our LLNs and CLT results.

\subsection{Simultaneous Uniform Ergodicity of Stereographic Algorithms}
\label{section-uniform-ergodicity}

Before going further, we must ensure that the stereographic MCMC algorithms each satisfy a minorisation condition as described in Equation \eqref{equ-minorisation-general}.

Since these Markov processes live in $\mathbb{S}^d$, which is compact, it is natural for them to exhibit uniform ergodicity properties similar to those of other algorithms targeting bounded densities with compact supports. As discussed alongside assumptions \ref{assump-boundedness-condition} and \ref{assump-derivatives-condition}, the situation becomes slightly more complicated at the North Pole $N$, since $\pi_\gamma(z)$ or $v \cdot \nabla_z \log\pi_\gamma(z)$ may not remain bounded as $z_{d+1} \rightarrow 1$.

Indeed, we have the following results for the SRW and SSS:
\begin{lemma}[SRW Minorisation Condition]
    Suppose that $\gamma \in \Gamma_{r,R}$ for $0<r<R<+\infty$. For $P$ the Markov transition kernel for the SRW targeting $\pi$, and assuming $\pi$ satisfies assumption \ref{assump-boundedness-condition}, then $\exists \epsilon >0$ and a probability measure $\nu$ on $\mathbb{R}^d$ such that $\forall x \in \mathbb{R}^d$,
    \begin{equation*}
        P^3(x, \cdot) \geq \epsilon \nu(\cdot).
    \end{equation*}
    Furthermore, $\epsilon$ and $\nu$ can be chosen to be independent of $\gamma$.
    \label{lemma-srw-uniform-ergodicity}
\end{lemma}

\begin{lemma}[SSS Minorisation Condition]
    Suppose that $\gamma \in \Gamma_{r,R}$ for $0<r<R<+\infty$. For $P$ the Markov transition kernel for the SSS targeting $\pi$, and assuming $\pi$ satisfies assumption \ref{assump-boundedness-condition}, then $\exists \epsilon >0$ and a probability measure $\nu$ on $\mathbb{R}^d$ such that $\forall x \in \mathbb{R}^d$,
    \begin{equation*}
        P(x, \cdot) \geq \epsilon \nu(\cdot).
    \end{equation*}
    Furthermore, $\epsilon$ and $\nu$ can be chosen to be independent of $\gamma$.
    \label{lemma-sss-uniform-ergodicity}
\end{lemma}

Note that $\epsilon$ and $\nu$ are implicitly different for the SRW and SSS. The proofs of these results involve constructing suitable sequences of events which allow the process to hit any given open ball of arbitrarily small radius, then lower-bounding the probability of this sequence of events. We can then extend this lower bound to a measure on $(\mathbb{R}^d, \mathcal{B}(\mathbb{R}^d))$ by writing any set $A \in \mathcal{B}(\mathbb{R}^d)$ as a union of open balls. Most of the details can be found in \cite{yang2022stereographic} for the SRW or \cite{habeck2023geodesic} for the SSS.

One should note that the minorisation condition we have for the SRW is a 3-step minorisation, not a 1-step minorisation. To apply many of the existing theoretical results on adaptive algorithms, we would need to first construct a 1-step small set and a simultaneous drift condition towards it, which are both highly non-trivial tasks. This further motivates explicitly using the uniform ergodicity properties of the process to unify the theoretical results.

For the equivalent result in the case of the SBPS, we project the sample path $\{ Z_t, V_t \}_{t \geq 0}$ onto Euclidean space to obtain a Markov process $\{X_t,W_t\}_{t \geq 0}$ with invariant distribution $\pi(x) \times p_\gamma(w \mid x)$ over $\mathbb{R}^d \times \mathbb{S}^{d-1}$. $W_t$ is a unit vector related to the direction of the particle in $\mathbb{R}^d$, and is only necessary because $\{X_t\}_{t \geq 0}$ alone is not a Markov process.

It is noteworthy that here we are indeed considering a stationary distribution which changes with the parameter.

\begin{lemma}[SBPS Minorisation Condition]
    Suppose that $\gamma \in \Gamma_{r,R}$ for $0<r<R<+\infty$. For $P$ the Markov transition kernel for the projected SBPS process $\{ (X_t,W_t) \}_{t \geq 0}$ targeting $\pi$, and assuming $\pi$ satisfies assumption \ref{assump-derivatives-condition}, then $\exists T^*>0$, $\epsilon >0$ and a probability measure $\nu$ on $\mathbb{R}^d \times \mathbb{S}^{d-1}$ such that, $\forall T \geq T^*$, $\forall (x,w) \in \mathbb{R}^d \times \mathbb{S}^{d-1}$,
    \begin{equation*}
        P^T((x,w),\cdot) \geq \epsilon \nu(\cdot).
    \end{equation*}
    Furthermore, $T^*$, $\epsilon$ and $\nu$ can be chosen to be independent of $\gamma$.
    \label{lemma-sbps-uniform-ergodicity}
\end{lemma}

The proof of Lemma \ref{lemma-sbps-uniform-ergodicity} follows a Lyapunov function and small set proof. Most of the details can be found in \cite{yang2022stereographic}. Note that Lemma \ref{lemma-sbps-uniform-ergodicity} gives a minorisation condition for the continuous-time SBPS, as well as any discrete-time skeleton of the SBPS.

With these minorisation conditions, we can map any of our algorithms' sample paths onto a segment chain with similar properties. We therefore only need to prove results in the general setting to obtain results for each of our algorithms, under appropriate conditions on $\pi$.

\subsection{Splitting and Regenerations for the segment chain}
\label{section-excursions-setup}

We now return to our segment chain $\{\Phi_n \}_{n \in \mathbb{N}}$, as described in either Equation \eqref{equ-segment-chain} or \eqref{equ-segment-chain-discrete} for a general (continuous or discrete-time) Markov process $X$ on a state space $\chi$.

We let $\mathbb{Q}_\mu( \cdot \; ; \gamma)$ be the probability measure on $(\Omega, \mathcal{B}(\Omega))$, the space of paths of length $T$, induced by the dynamics of $\{ X_s \}_{0 \leq s \leq T}$ under parameter $\gamma$, subject to $X_0 \sim \mu$ for some probability measure $\mu$ on $\chi$. If $X_0 =x$ a.s., we write this measure as $\mathbb{Q}_x(\cdot \; ;\gamma)$. With this setup, we have that $\forall A \in \mathcal{B}(\Omega), n \in \mathbb{N}$,
\begin{equation}
    \mathbb{P} \left( \Phi_n \in A \mid \Phi_{0:(n-1)}\right) = \mathbb{Q}_{\Phi_{n-1}(T)}(A \; ;\gamma).
    \label{equ-segment-chain-probs}
\end{equation}
We write
\begin{equation}
    P_{\Phi;\gamma}( \phi, \cdot) = \mathbb{Q}_{\phi(T)}(\cdot \; ;\gamma),
    \label{equ-segment-kernel}
\end{equation}
for its transition kernel. If $\pi_\gamma$ is the unique stationary distribution of $X$, this kernel admits $\mathbb{Q}_{\pi_\gamma}(\cdot \; ;\gamma)$ as its stationary distribution.

We hope to use minorisation condition \eqref{equ-minorisation-general} to extend the state space $\Omega$ to a new space $\widecheck{\Omega} = \Omega \times \{0,1\}$, such that the Markov chain $(\Phi_n,Y_n)_{n \in \mathbb{N}}$ on $\widecheck{\Omega}$ possesses an ergodic atom whilst retaining the marginal transition probabilities of the $\Phi$ component as given in Equation \eqref{equ-segment-chain-probs}.

Mimicking the split chain constructions from \cite{bednorz2008regeneration} or \cite[Section 17.3]{meyn2012markov}, we start by considering the $T$-skeleton chain for the original process $\{ X_{nT} \}_{n \in \mathbb{N}}$. Given the minorisation \eqref{equ-minorisation-general}, we can define $Y_n \in \{0,1\}$, and obtain a Markov chain $\{(X_{nT}, Y_n)\}_{n \in \mathbb{N}}$ with associated probability measure $\widecheck{\mathbb{P}}_\gamma$ by setting
\begin{equation}
    \begin{gathered}
        \widecheck{\mathbb{P}}_\gamma\left( X_{(n+1)T} \in A \mid Y_n =1, X_{nT} \right) = \nu(A), \\[4pt]
        \widecheck{\mathbb{P}}_\gamma\left( X_{(n+1)T} \in A \mid Y_n =0, X_{nT} \right) = \eta_\gamma\left( X_{nT},A \right),
    \end{gathered}
    \label{equ-split-skeleton-transition}
\end{equation}
where
\begin{equation}
    \eta_\gamma(x, A) = \frac{P^T_\gamma(x,A) - \epsilon\nu(A)}{1-\epsilon},
    \label{equ-split-remainder-measure}
\end{equation}
as well as
\begin{equation}
    \widecheck{\mathbb{P}}_\gamma\left( Y_n = 1 \mid X_{nT} \right)= \epsilon.
    \label{equ-split-skeleton-Y-update}
\end{equation}
It is a standard result that this chain has the correct marginal distributions for $\{X_{nT}\}_{n\in\mathbb{N}}$, and that the set $\chi \times \{1\}$ is a regenerative atom for the chain, i.e.\ $A \subset \chi, i \in \{0,1\}$,
\begin{equation}
    \widecheck{\mathbb{P}}_\gamma\left( X_{(n+1)T} \in A, Y_{n+1} = i \mid X_{nT}, Y_n = 1 \right) = \nu(A) \epsilon^{i}(1-\epsilon)^{1-i}.
\end{equation}
In other words, conditionally on $\{Y_n = 1\}$, the processes $\{ X_{kT}, Y_k \}_{k \geq n+1}$ and $\{ X_{kT}, Y_k \}_{k \leq n}$ are independent.

To transfer these properties over to the segment chain $\Phi$, we bridge the paths from $X_{nT}$ to $X_{(n+1)T}$ conditionally on the endpoints. We define the Radon-Nykodym derivatives $\frac{\dif \nu}{\dif P^T_\gamma}$ and $\frac{\dif \eta_\gamma}{\dif P^T_\gamma}$, which are functions of $x$ and $x^\prime$, and extend the measure $\widecheck{\mathbb{P}}_\gamma$ to be a transition kernel from $\Phi_n$ to $\Phi_{n+1}$. For the $Y_n$ updates, we get
\begin{equation}
    \widecheck{\mathbb{P}}_\gamma\left( Y_n = 1 \mid \{ \Phi_k \}_{k=0}^n, \{ Y_k \}_{k=0}^{n-1} \right)= \epsilon,
    \label{equ-split-segment-Y-update}
\end{equation}
and for the $\Phi_{n+1}$ updates, we get
\begin{equation}
    \begin{gathered}
        \widecheck{\mathbb{P}}_\gamma\left( \Phi_{n+1} \in A \mid Y_n =1, \Phi_n \right) = \int_A \frac{\dif\nu}{\dif P^T_\gamma}\left( \Phi_n(T), \phi(T) \right) \; \dif \mathbb{Q}_{\Phi_n(T)}(\phi;\gamma), \\[6pt]
        \widecheck{\mathbb{P}}_\gamma\left( \Phi_{n+1} \in A \mid Y_n =0, \Phi_n \right) = \int_A \frac{\dif \eta_\gamma}{\dif P^T_\gamma}\left( \Phi_n(T), \phi(T) \right) \; \dif \mathbb{Q}_{\Phi_n(T)}(\phi;\gamma).
    \end{gathered}
    \label{equ-split-segment-transition}
\end{equation}

One can check that these give the appropriate marginal transition kernels to recover $\Phi_{n+1} \sim \mathbb{Q}_{\Phi_n(T)}(\cdot \; ;\gamma)$, and that $\Phi_{n+1}(T) \mid \Phi_n, \{Y = 1\} \sim \nu$. It is easy to check that the invariant distribution for the split segment chain is
\begin{equation}
    \widecheck{\pi}_\gamma(A \times \{i\}) = \epsilon^i (1-\epsilon)^{1-i} \mathbb{Q}_{\pi_\gamma}(A \; ; \gamma),
    \label{equ-segment-invariant}
\end{equation}
for $A \in \mathcal{B}(\Omega)$.

Let $\widecheck{P}_\gamma( (\phi,y), \cdot)$ be the transition kernel of the split segment chain $\{ \Phi_n, Y_n \}_{n \in \mathbb{N}}$, as given by Equations \eqref{equ-split-segment-Y-update} and \eqref{equ-split-segment-transition}. Since $Y_n = 1$ implies that $X_{(n+1)T} \sim \nu$ independently of the history of the process up to time $Tn$, we can show that the split segment chain admits $\widecheck{\alpha} = \Omega \times \{1\}$ as an atom:
\begin{lemma}[Atom for the Split Segment Chain]
    For $\{ \Phi_n, Y_n \}_{n \in \mathbb{N}}$ the split segment chain with transition law $\widecheck{P}_\gamma$, we have that $\forall \phi \in \Omega$, and $\forall A \in \mathcal{B}(\Omega)$, $i \in \{0,1\}$,
    \begin{equation*}
        \widecheck{P}^2_\gamma \left( (\phi,1), A \times \{i\} \right) = C^i (1-C)^{1-i} \, \mathbb{Q}_{\nu}(A \; ; \gamma).
    \end{equation*}
    In particular, conditionally on the event $\{Y_n =1\}$, the pre-$n$ process $\{ \Phi_k, Y_k\}_{k \leq n}$ and the post-$(n+2)$ process $\{\Phi_k, Y_k\}_{k \geq n+2}$ are independent. Given $\{Y_n =1\}$, $\{\Phi_k, Y_k\}_{k \geq n+2}$ has the same law as $\{\Phi_k, Y_k\}_{k \geq 1}$ with initial distributions $\Phi_0(T) \sim \nu$ and $Y_0 \sim \text{Bern}(C)$.
    \label{lemma-atom-1-dependence}
\end{lemma}

With Lemma \ref{lemma-atom-1-dependence}, we can consider an excursion process between hitting times of the atom $\widecheck{\alpha}$, the details of which can be found in Appendix \ref{section-excursions-details}. These are identically distributed, 1-dependent sequences of random ($\text{Geom}(\epsilon)$ distributed) length. This structure can be visualised in Figure \ref{fig-excursions}.

As a follow-up remark, since $Y_n =1$ implies that $\Phi_n$ and $\Phi_{n+2}$ are independent, we note that the ``bridge'' path $\Phi_{n+1}$ has the property that $\Phi_{n+1}(0) = \Phi_n(T)$ and $\Phi_{n+1}(T) = \Phi_{n+2}(0)$ are independent.

\begin{figure}[b!]
    \centering
    \includegraphics[width=4.7in]{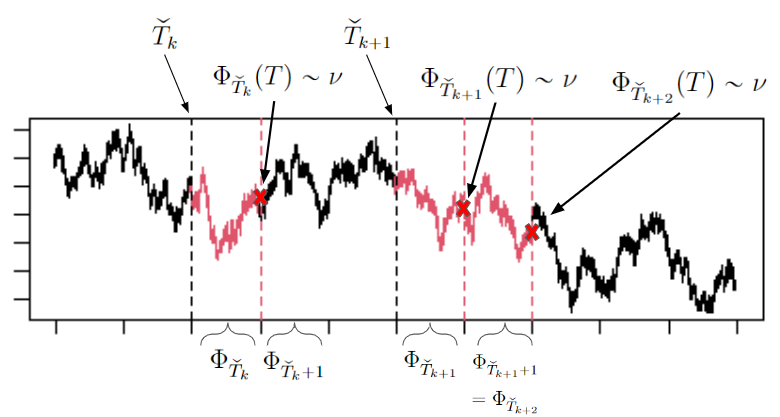}
    \caption{Example path including excursions (path is for illustration purposes and does not reflect a real sampling scheme). The regeneration times $\widecheck{T}_k$ (defined in appendix \ref{section-excursions-details}) are shown, along with associated random variables: the regeneration segments $\Phi_{\widecheck{T}_k}$ are shown in red. These segments are independent of neither the immediate past nor the immediate future. The random length excursions $\Psi_k$ go between red crosses and are identically distributed and 1-dependent.}
    \label{fig-excursions}
\end{figure}

With this framework, any uniformly ergodic Markov process can be transformed into a discrete-time chain with an ergodic atom.

\subsection{Proving the Asymptotic Results}

We now have enough setup to prove the results in Theorem \ref{thm-fhat-asymptotics}. Consider a bounded function $f:\mathbb{R}^d \rightarrow \mathbb{R}$, with $\lvert f \rvert \leq M$, and assume that $\pi_\gamma f = 0$ for all $\gamma \in \Gamma$. We are interested in the behaviour of the estimator $\widehat{f}_t$ as given in Equation \eqref{equ-estimator}. If the original process $X$ is a continuous-time process, we write
\begin{equation}
    F(\phi) = \frac{1}{T} \int_{s=0}^T f(x_s) \, \dif s,
    \label{equ-F(phi)}
\end{equation}
for $\phi = \{ x_s\}_{s=0}^T \in \Omega$, to get that
\begin{equation}
    \widehat{f}_t = \frac{1}{n} \sum_{i=0}^{n-1} F(\Phi_i) + \mathcal{O}(t^{-1}),
    \label{equ-fhat-truncation}
\end{equation}
where $n = \lfloor \frac{t}{T} \rfloor$. An analogous statement can be made for discrete-time processes. The asymptotic properties of $\widehat{f}_t$ will therefore be the same as those of
\begin{equation}
    \widehat{f}^\Phi_n = \frac{1}{n} \sum_{i=0}^{n-1} F(\Phi_i),
\end{equation}
as $n \rightarrow +\infty$.

Recall that we will be working with an AIR MCMC scheme where the lags between adaptation times are polynomially increasing. Let $\beta >0, c \geq 1$ such that $\forall k \in \mathbb{N}$,
\begin{equation}
    \frac{1}{c}k^\beta \leq n_k \leq ck^\beta,
    \label{equ-poly-lags-segment}
\end{equation}
for $n_k \in \mathbb{N}$. The adaptation times are then given by $N_k = \sum_{i=1}^k n_k$, with $N_0 = n_0 = 0$ for consistency of notation.

We consider the AIR segment chain process as follows: let $\{ (\Phi_n,Y_n) \}_{n \in \mathbb{N}}$ evolve according to $\widecheck{P}_{\gamma_k}$ for $N_k \leq n < N_{k+1}$, then update the parameter to $\gamma_{k+1}$ at $N_{k+1}$. In this setting, we obtain the following results:
\begin{theorem}[AIR Segment Chain $\mathcal{L}^2$ Convergence and CLT]
    Consider $(\Phi_n,Y_n)_{n \geq 0}$ the AIR segment chain process, where the original family of transition kernels $\{P_\gamma\}_{\gamma \in \Gamma}$ have invariant distributions $\pi_\gamma$ and satisfy the minorisation condition \eqref{equ-minorisation-general}. Assume that $f$ is bounded and $\mathbb{E}_{\pi_\gamma}(f(X))=0$ for all $\gamma \in \Gamma$. Then $\forall (\Phi_0,Y_0) \in \widecheck{\Omega}$, $\gamma_0 \in \Gamma$, and any adaptation scheme:
    \begin{itemize}
        \item For any $\beta >0$, and $\forall \epsilon < \min\left(\frac{1}{2}, \frac{\beta}{1+\beta}\right)$,
        \begin{equation*}
            n^\epsilon \widehat{f}^\Phi_n \xrightarrow{\mathcal{L}^2} 0, \qquad \text{as } n \rightarrow \infty.
        \end{equation*}
        \item For $\beta>1$, if $\gamma_i \xrightarrow{P} \gamma_\infty$, for some constant $\gamma_\infty$, such that $\sigma^2(\gamma_\infty)>0$, and that $\sigma^2(\gamma)$ is a continuous function of $\gamma$, then a CLT holds:
        \begin{equation*}
            \sqrt{n} \widehat{f}^\Phi_n \xrightarrow{D} \mathcal{N}(0, \sigma^2(\gamma_\infty)), \qquad \text{as } n \rightarrow \infty.
        \end{equation*}
    \end{itemize}
    \label{thm-air-segment-chain-1}
\end{theorem}

The proof of Theorem \ref{thm-air-segment-chain-1} is in Appendix \ref{section-air-proof}.

We can also apply the almost sure convergence results from \cite{hofstadler2024almost} to obtain almost sure convergence results:
\begin{theorem}[AIR Segment Chain Almost Sure Convergence]
    Consider $(\Phi_n,Y_n)_{n \geq 0}$ the AIR segment chain process, where the original family of transition kernels $\{P_\gamma\}_{\gamma \in \Gamma}$ have invariant distributions $\pi_\gamma$ and satisfy the minorisation condition \eqref{equ-minorisation-general}. Assume that $f$ is bounded and $\mathbb{E}_{\pi_\gamma}(f(X))=0$ for all $\gamma \in \Gamma$. Then $\forall (\Phi_0,Y_0) \in \widecheck{\Omega}$, $\gamma_0 \in \Gamma$, and any adaptation scheme, if $\epsilon < \min\left( \frac{1}{2}, \frac{\beta}{1+\beta} \right)$ then
        \begin{equation*}
            n^\epsilon\widehat{f}_n \xrightarrow{\text{a.s.}} 0.
        \end{equation*}
    \label{thm-air-segment-chain-2}
\end{theorem}

This result follows directly from Theorems 4.3 and 4.5 of \cite{hofstadler2024almost}.

Finally, since $\widehat{f}_t = \widehat{f}^\Phi_n + \mathcal{O}(t^{-1})$, we can use Theorem \ref{thm-air-segment-chain-1} and Theorem \ref{thm-air-segment-chain-2} to immediately obtain the results on $\widehat{f}_t$ in Theorem \ref{thm-fhat-asymptotics}, which concludes our discussion of the asymptotic results for the adaptive stereographic MCMC algorithms.

\section{Simulation Studies}
\label{section-simulations}

In this final section, we demonstrate the ability of our algorithms to perform well in challenging heavy-tailed, high-dimensional settings.

All code was run in Julia. We simulate the SBPS bounce events by combining the techniques from \cite{corbella2022automatic} and \cite{andral2024automated}, with the use of the ForwardDiff package from \cite{revels2016forward}. We use the adaptive window width from \cite{andral2024automated}, as it allows for significant flexibility in cases where the density may be highly concentrated, but we use Brent's method from \cite{corbella2022automatic}, rather than any methods involving Hessian information, since Hessian matrices can be very expensive in high dimensions.

\subsection{Impact of Parameters on Mixing}
\label{section-autocorrelation}

We present a comparison of autocorrelation plots targeting a MtD in $d$ dimensions with $d$ DoF. We compare the autocorrelation of $\lVert X \rVert^2$ for $d=100$ when taking $\gamma$ to be based on the mean and covariance of $n$ iid $\mathcal{N}(0_d,I_d)$ random vectors, reflecting what we might expect at different stages of a run of our AIR algorithms. As $n$ increases, the parameter estimators converge to the optimal parameters, which leads to $\pi_\gamma$ becoming uniform over $\mathbb{S}^d$. We start all algorithms in stationarity and do not perform adaptations to demonstrate the improved performance we obtain by updating the parameter estimates. For the SRW and SBPS, we tune $h$ and $\lambda_\text{ref}$ following the intuition given in the closing remarks of Section \ref{section-equator}.

\begin{figure}[t!]
    \begin{subfigure}{2.36in}
        \centering
        \includegraphics[width=2.35in]{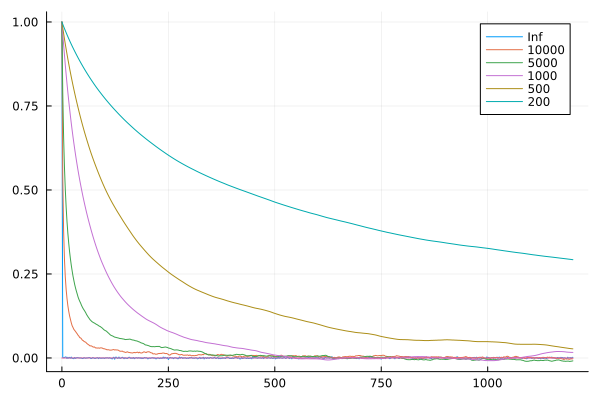}
        \caption{Autocorrelation for SRW}
    \end{subfigure}
    \hfill
    \begin{subfigure}{2.36in}
        \centering
        \includegraphics[width=2.35in]{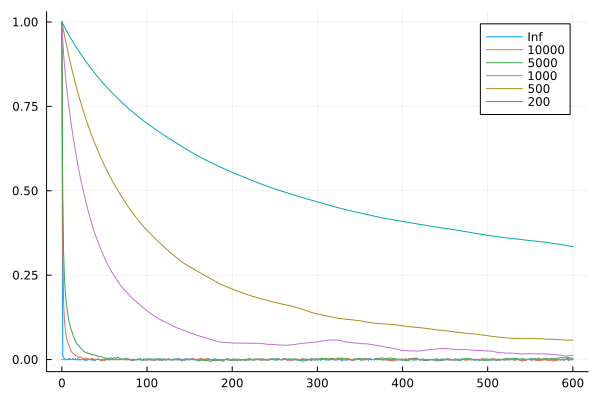}
        \caption{Autocorrelation for SSS}
    \end{subfigure}
    \\[10pt]
    \centering
    \begin{subfigure}{2.36in}
        \centering
        \includegraphics[width=2.35in]{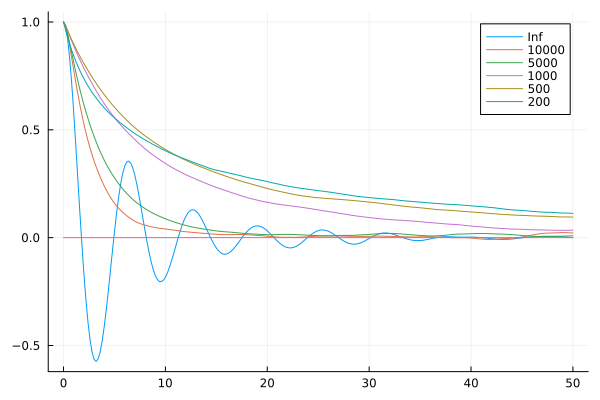}
        \caption{Autocorrelation for SBPS}
    \end{subfigure}
    \caption{Autocorrelation plots of $\lVert X \rVert^2$ for the stereographic algorithms targeting a MtD with $\nu = d = 100$. $\gamma$ is taken as the empirical mean and covariance estimator for $n$ iid $\mathcal{N}(0_d,I_d)$ with $n$ varying from $200$ to $\infty$.}
    \label{fig-autocor}
\end{figure}

In all cases, we see from Figure \ref{fig-autocor} that the autocorrelation dramatically improves as $n$ increases. In the optimal $n = \infty$ case, we see that the SRW and SSS achieve essentially uncorrelated moves after only one step. For the SBPS, we can observe the strength of the non-reversible dynamics in producing negatively correlated samples in the optimal setting. This clearly demonstrates the benefit of using adaptive methods to improve $\gamma$, since it yields such large improvements in autocorrelation.

When comparing the deterioration of the different algorithms as $n$ decreases, the SSS performs better than the SRW: the decay in autocorrelation for the SSS is more robust to poor parameters than the SRW. For the SBPS, the case $n=\infty$ seems very promising, but once the spherical symmetry is even slightly lost, then it begins to perform the worst out of all the algorithms. It is also worth noting that, in the optimal setting, the SBPS is significantly more expensive than either of the other two algorithms. For every 1 unit of time of the SBPS sample path, we could simulate roughly 200 steps of SRW or SSS. The SRW and SSS have comparable computational costs in the optimal regime, but the cost of one step of SRW does not increase as the quality of the parameters decreases. On the other hand, tuning the step size $h$ of the SRW can be challenging, whereas the SSS has no similar parameterisation issues.

\subsection{Robustness of Adaptive Algorithms}
\label{section-robustness}

Having looked at the improvements that adaptations can bring when sampling from stationarity, we now demonstrate the robustness of the adaptive algorithms to starting deep in the tails of the target with poor initial parameters. Consider targeting a MtD with $\nu = 2$, $d = 200$, which is far beyond the theoretical assumptions for Theorem \ref{thm-fhat-asymptotics} to guarantee convergence of estimators. We set the initial parameters $\mu = (1000, \dots, 1000)$, $\Sigma = d I_d$, and start the process on the equator to emulate a situation where the prior and posterior modes are very far apart, and the target is heavy-tailed.

We present plots of the sample paths for the adaptive SRW, SSS, SBPS, and a plot of a HMC sample path struggling to find the modes for comparison. These are Figures \ref{fig-asrw}, \ref{fig-asss}, \ref{fig-asbps}, and \ref{fig-hmc} respectively. We also include plots of the number of likelihood evaluations per step for the SSS and the number of gradient evaluations per unit time for the SBPS to show how the computational cost of the algorithms changes with the adaptations.

For our adaptation scheme, we use the latest quarter of the adaptive epochs to obtain an estimator for the target mean and covariance matrix via the standard empirical estimators. We also scale the covariance matrix to centre the latest adaptive epoch's sample path around the equator. This scaling is included to force the probability mass onto the equator and control our estimators in cases where the target covariance may be infinite.

\begin{figure}[t!]
    \centering
        \begin{subfigure}{2.36in}
            \centering
            \includegraphics[width=2.3in]{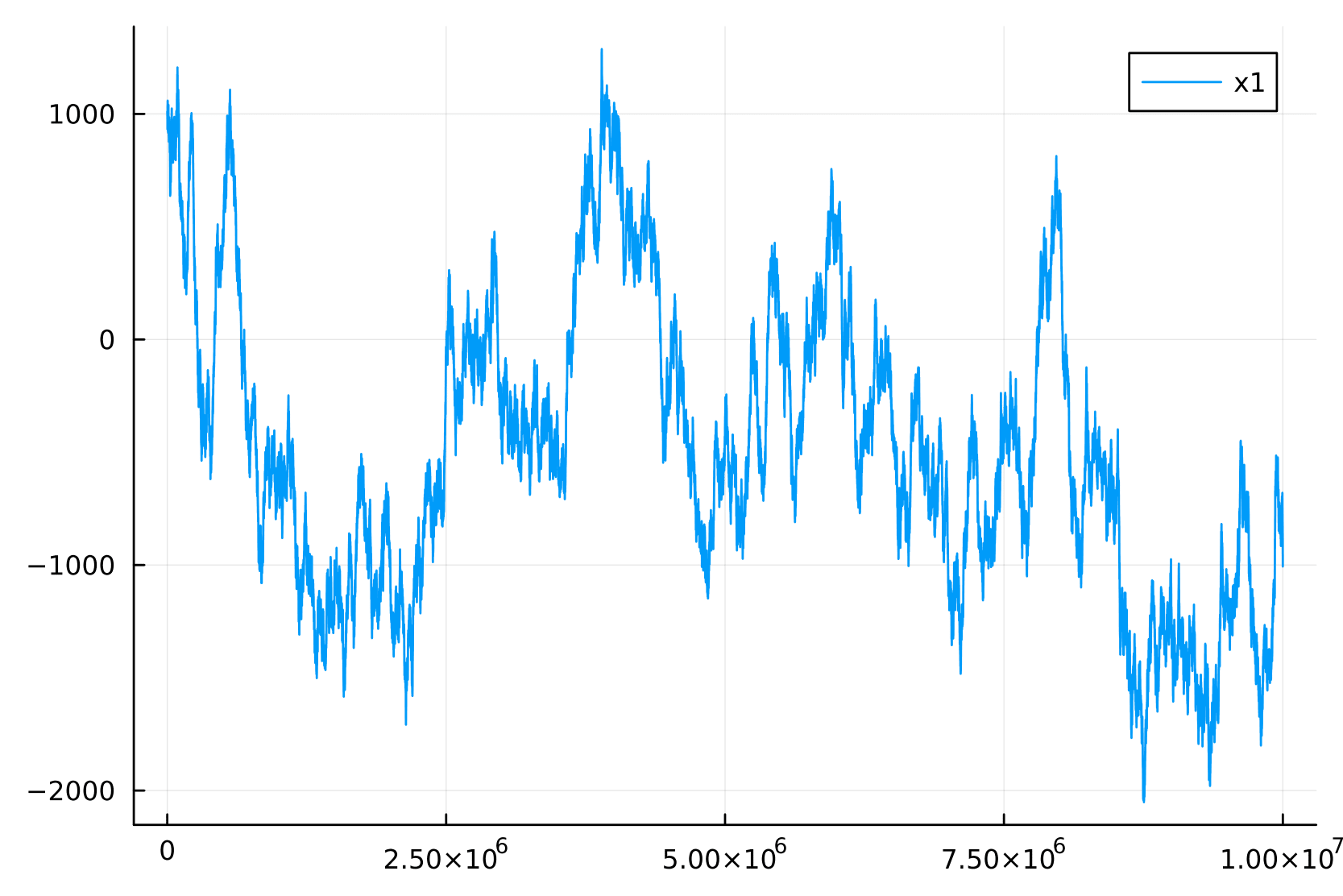}
            \caption{Plot of the first coordinate $X_1$}
            \label{fig-hmc-x}
        \end{subfigure}
        \hfill
        \begin{subfigure}{2.36in}
            \centering
            \includegraphics[width=2.3in]{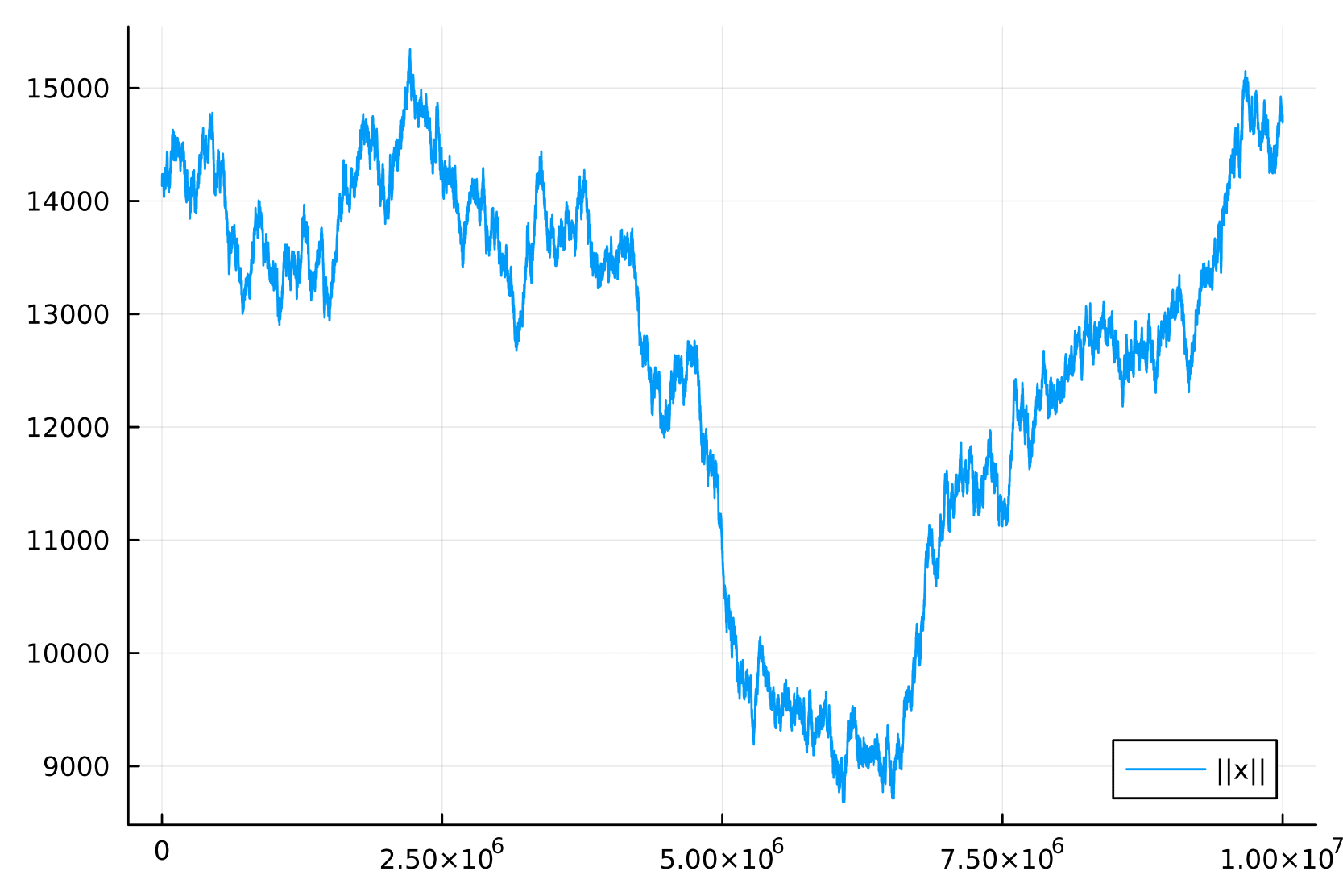}
            \caption{Plot of the norm $\lVert X \rVert$}
            \label{fig-hmc-norms}
        \end{subfigure}
    \caption{HMC sample paths targeting a MtD with $\nu = 2$, $d = 200$, started at $x_0 = (1000, \dots, 1000)$. This run took approximately 3.5 hours.}
    \label{fig-hmc}
\end{figure}

\begin{figure}[b!]
    \centering
        \begin{subfigure}{2.36in}
            \centering
            \includegraphics[width=2.3in]{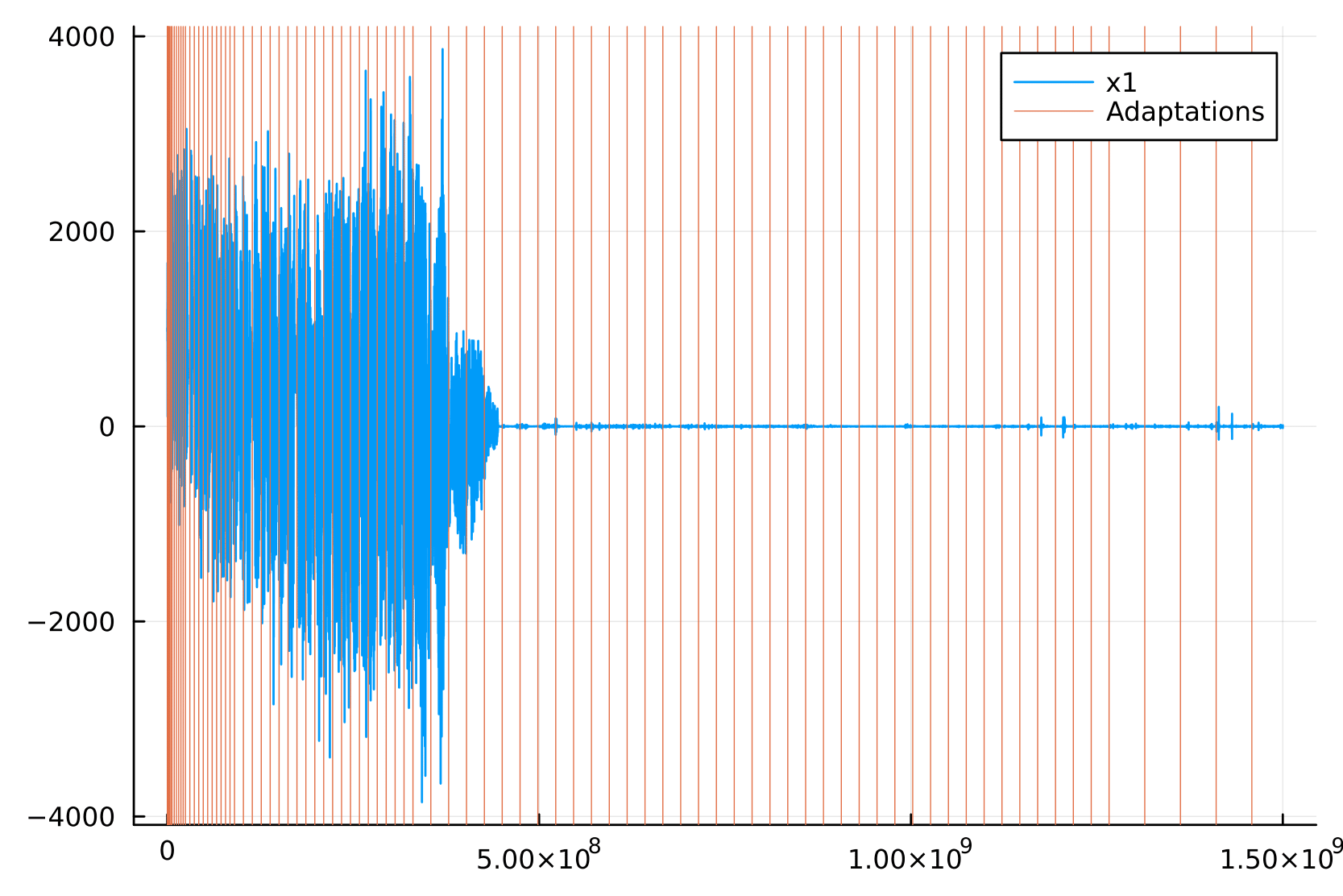}
            \caption{Plot of the first coordinate $X_1$}
            \label{fig-asrw-x}
        \end{subfigure}
        \hfill
        \begin{subfigure}{2.36in}
            \centering
            \includegraphics[width=2.3in]{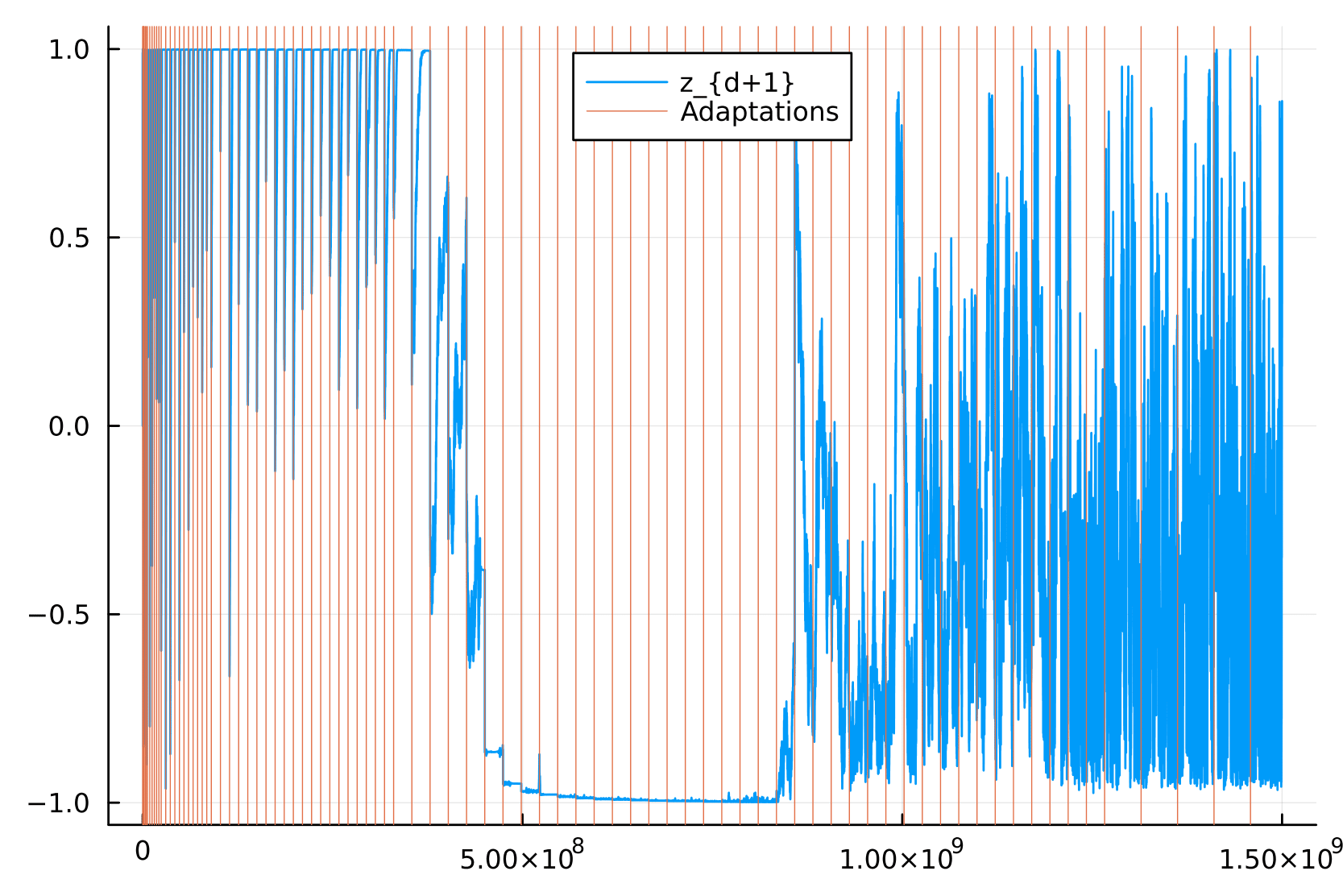}
            \caption{Plot of the latitude $Z_{d+1}$}
            \label{fig-asrw-z}
        \end{subfigure}
    \caption{Adaptive SRW sample paths targeting a MtD with $\nu = 2$, $d = 200$, and initial parameters $\mu = (1000, \dots, 1000)$, $\Sigma = d I_d$. We adaptively tune $h$ to target an average acceptance probability of $23.4\%$. This run took approximately 6 hours and 20 minutes.}
    \label{fig-asrw}
\end{figure}

\begin{figure}[t!]
    \centering
        \begin{subfigure}{2.36in}
            \centering
            \includegraphics[width=2.3in]{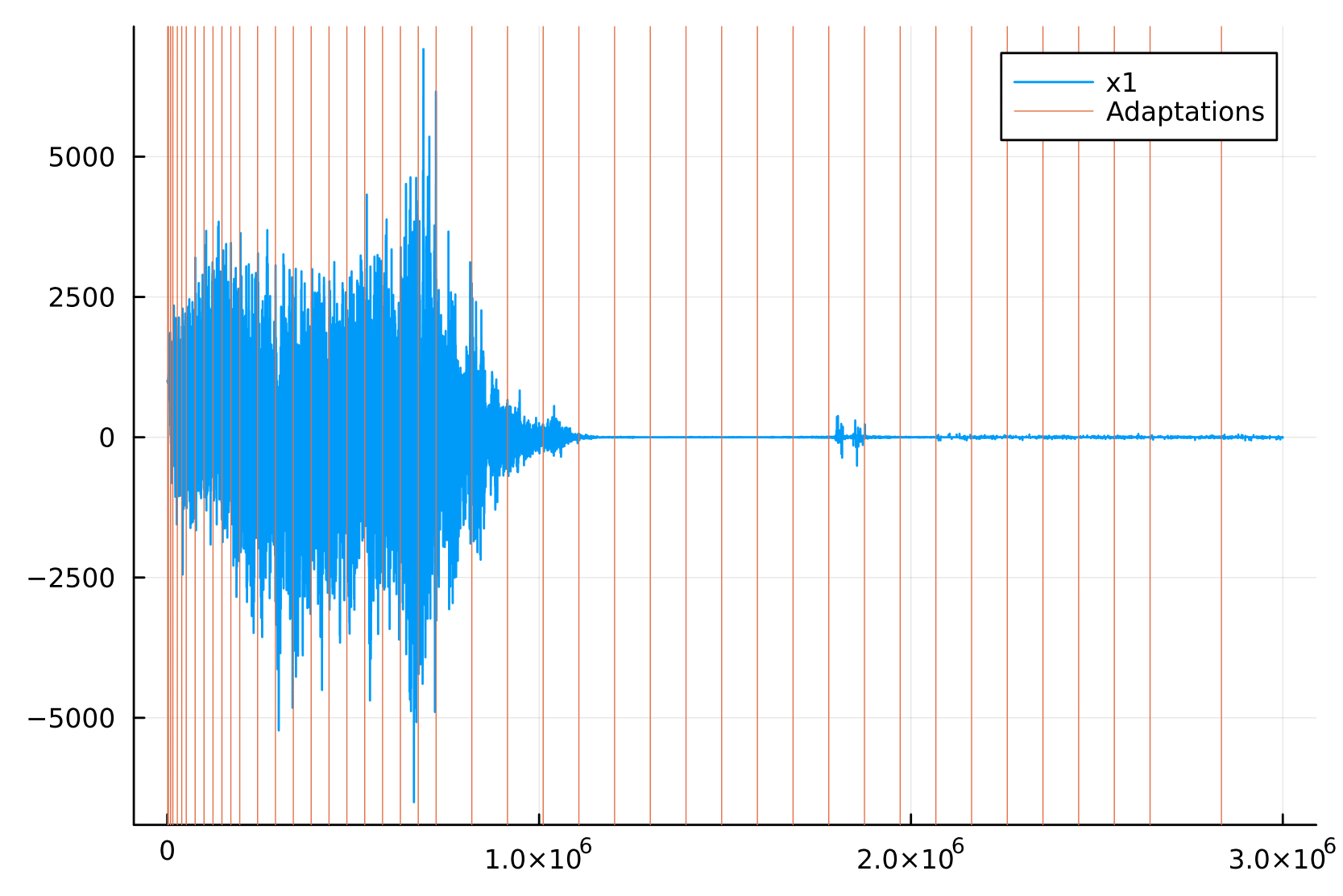}
            \caption{Plot of the first coordinate $X_1$}
            \label{fig-asss-x}
        \end{subfigure}
        \hfill
        \begin{subfigure}{2.36in}
            \centering
            \includegraphics[width=2.3in]{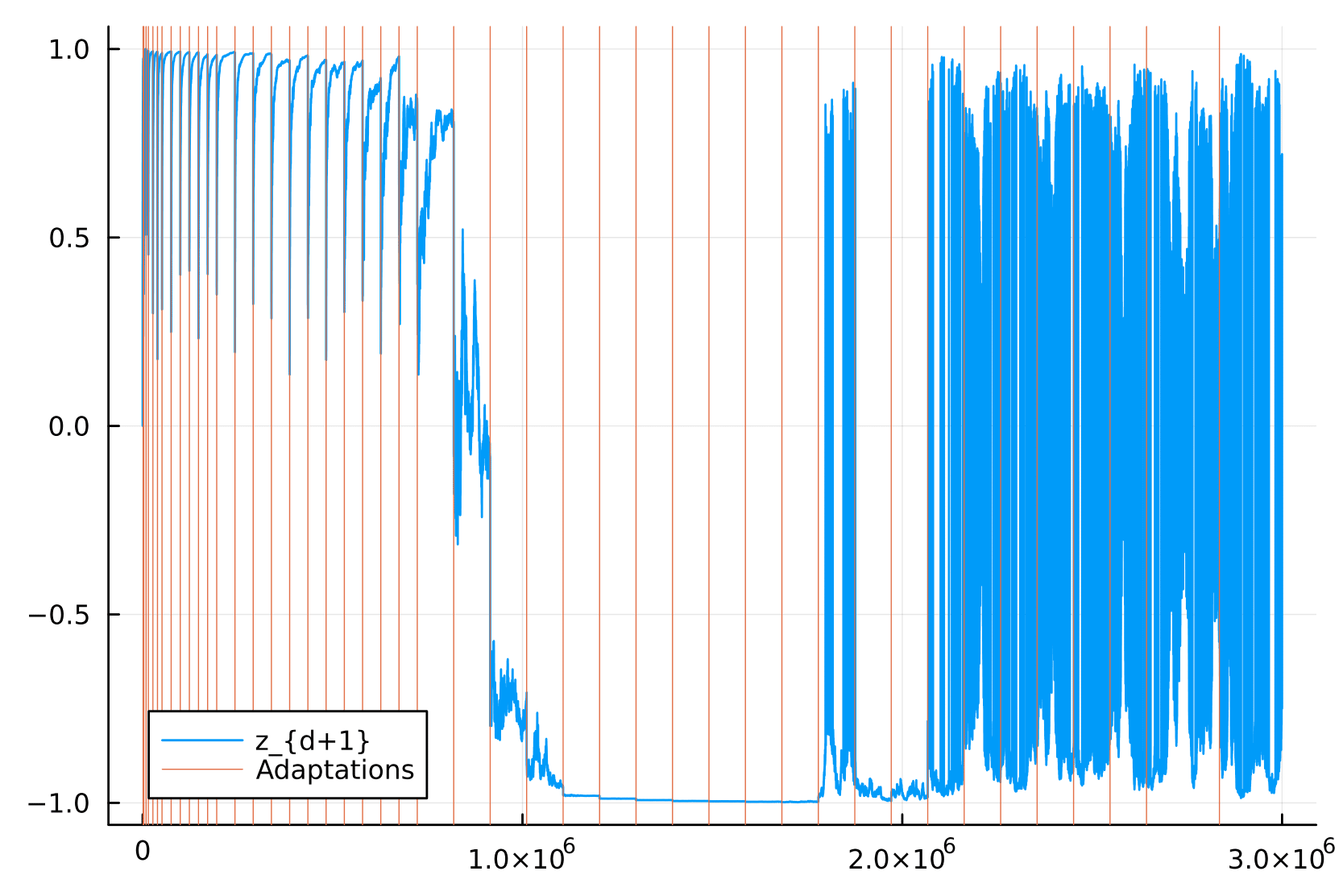}
            \caption{Plot of the latitude $Z_{d+1}$}
            \label{fig-asss-z}
        \end{subfigure}
        \\[10pt]
        \centering
        \begin{subfigure}{5in}
            \centering
            \includegraphics[width=2in]{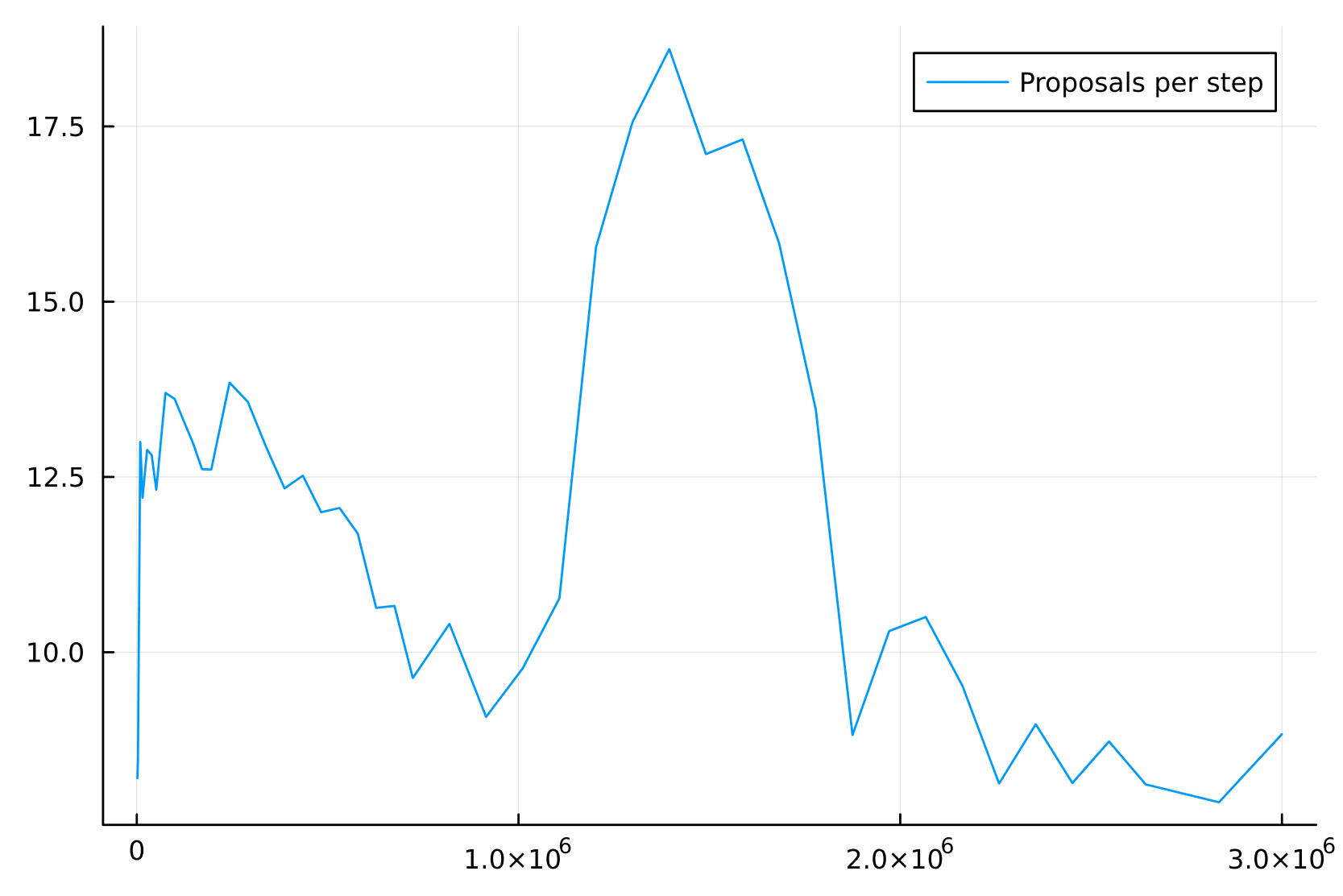}
            \caption{Average number of proposals per step over the run of the SSS}
            \label{fig-asss-prop}
        \end{subfigure}
    \caption{Adaptive SSS sample paths targeting a MtD with $\nu = 2$, $d = 200$, and initial parameters $\mu = (1000, \dots, 1000)$, $\Sigma = d I_d$. This run took approximately 15 minutes.}
    \label{fig-asss}
\end{figure}

\begin{figure}[t!]
    \centering
        \begin{subfigure}{2.36in}
            \centering
            \includegraphics[width=2.3in]{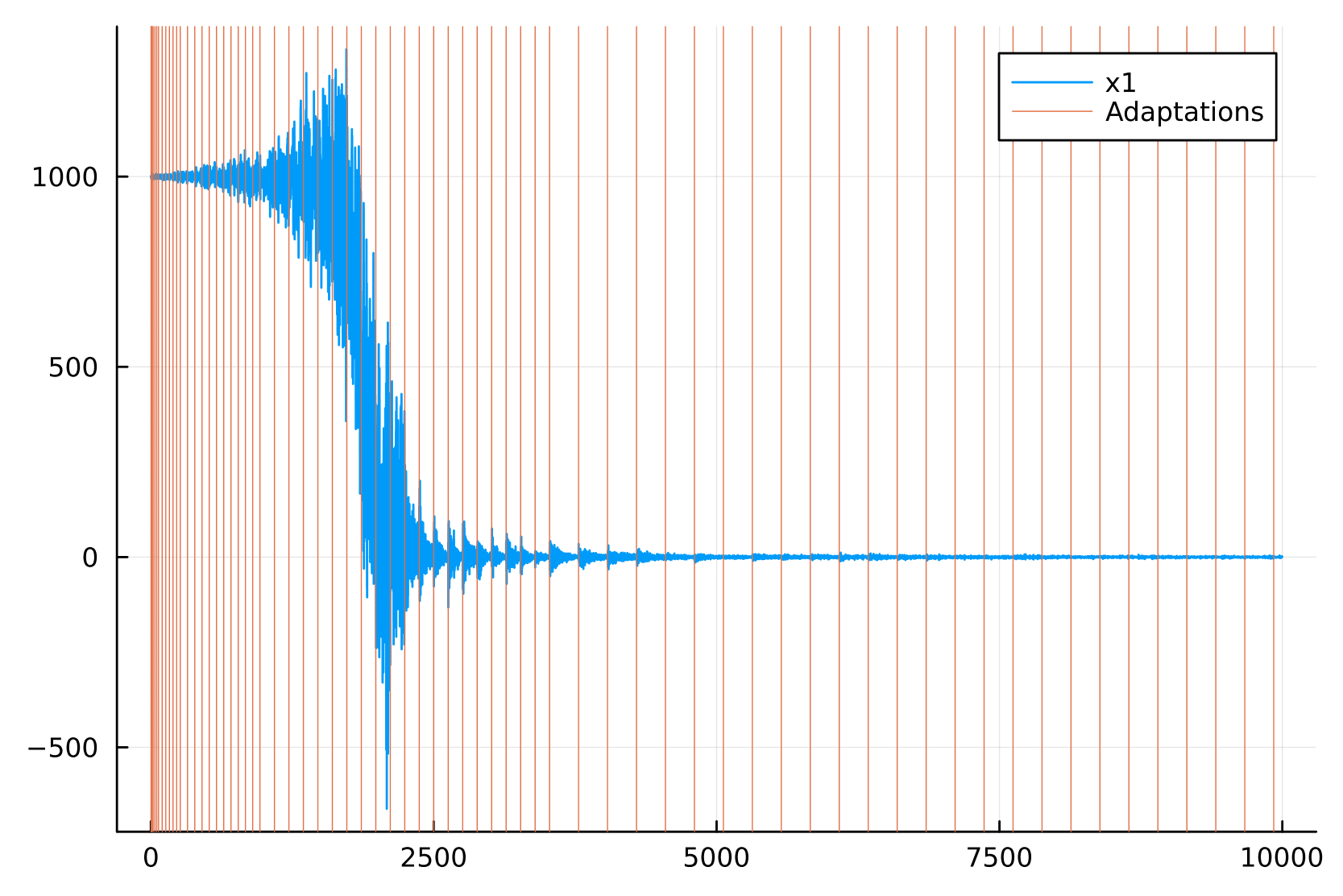}
            \caption{Plot of the first coordinate $X_1$}
            \label{fig-asbps-x}
        \end{subfigure}
        \hfill
        \begin{subfigure}{2.36in}
            \centering
            \includegraphics[width=2.3in]{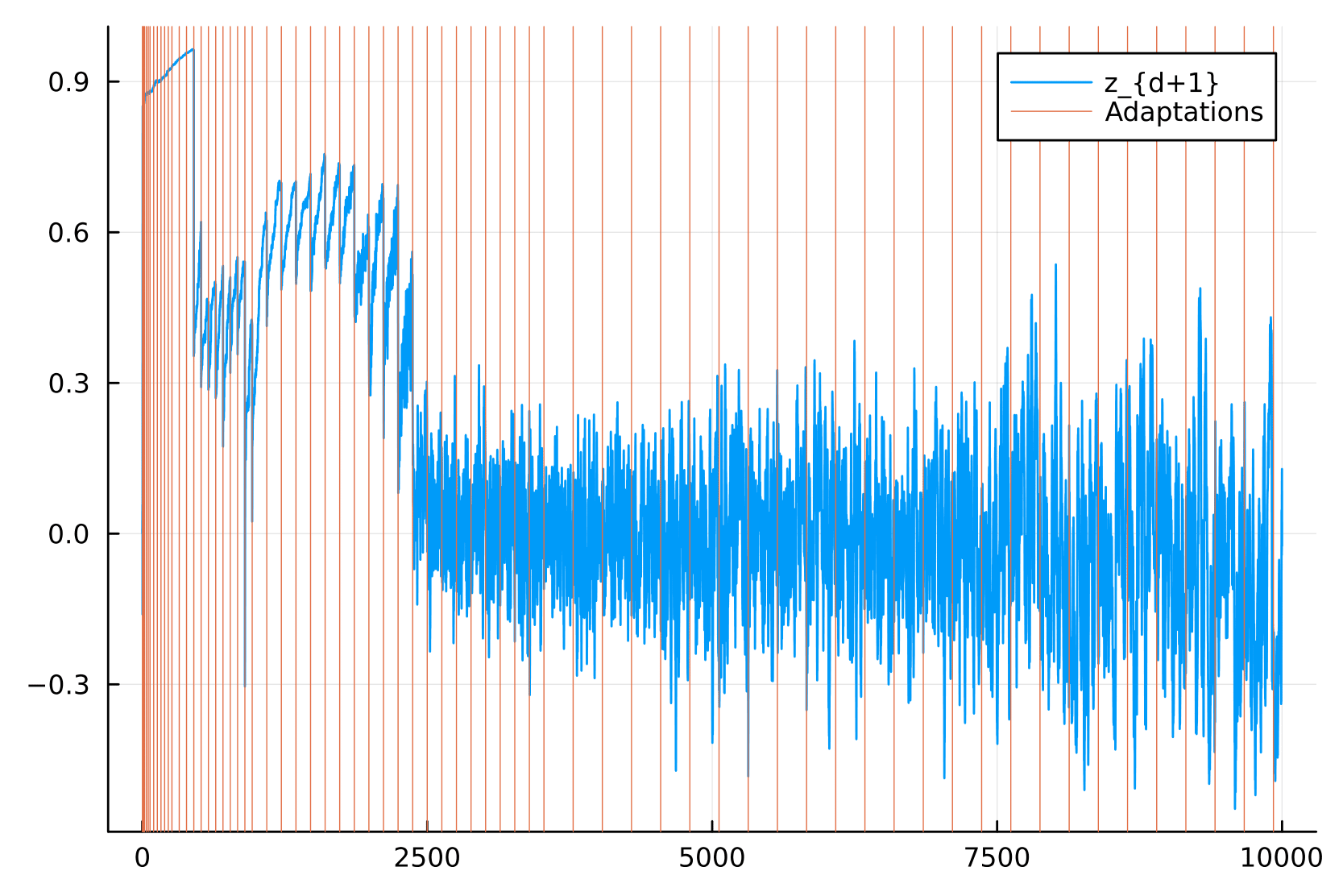}
            \caption{Plot of the latitude $Z_{d+1}$}
            \label{fig-asbps-z}
        \end{subfigure}
        \\[10pt]
        \centering
        \begin{subfigure}{5in}
            \centering
            \includegraphics[width=2in]{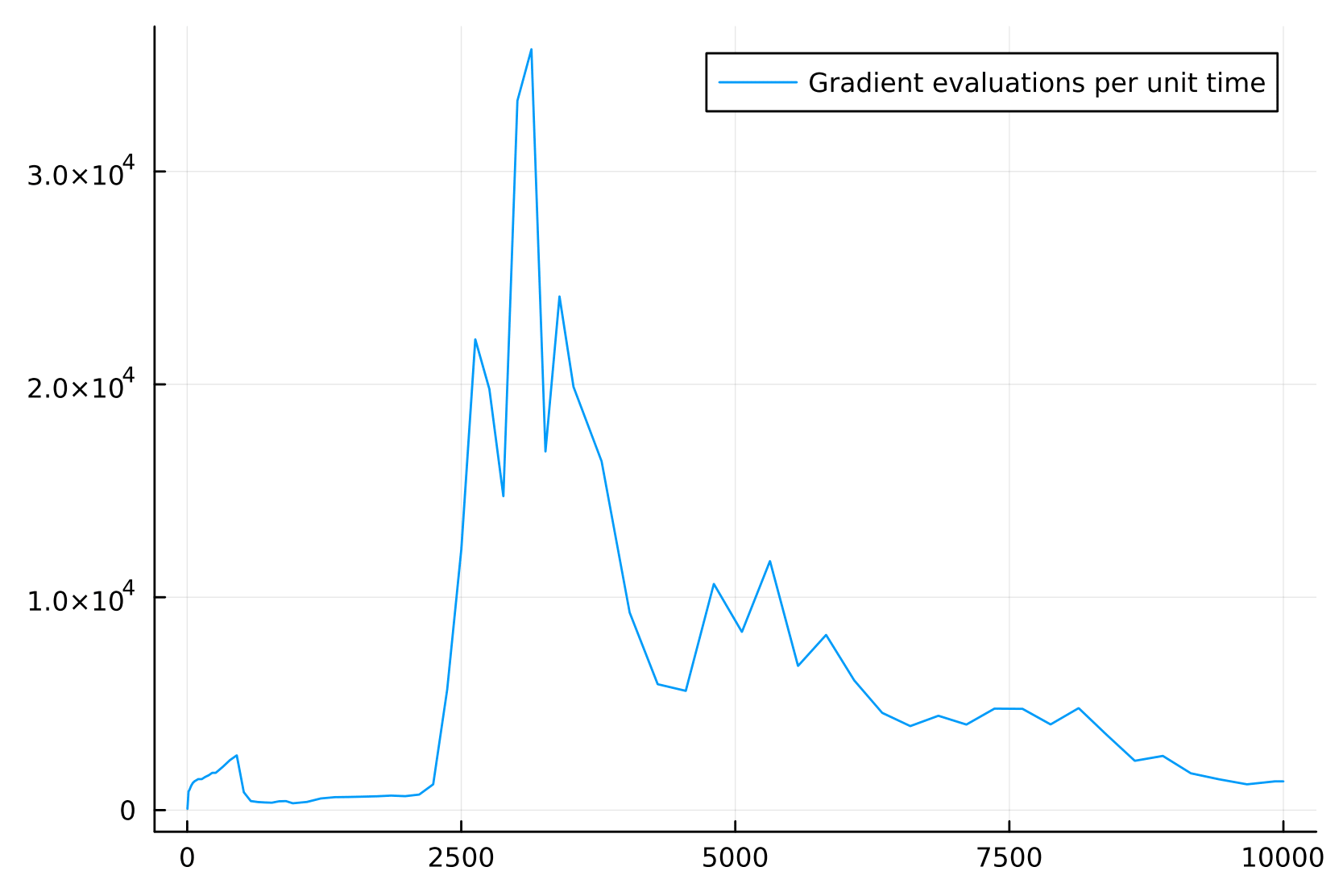}
            \caption{Average number of gradient evaluations per unit time over the run of the SBPS}
            \label{fig-asbps-eval}
        \end{subfigure}
    \caption{Adaptive SBPS sample paths targeting a MtD with $\nu = 2$, $d = 200$, and initial parameters $\mu = (1000, \dots, 1000)$, $\Sigma = d I_d$. We set $\lambda_\text{ref} = 1$ throughout. This run took approximately 1 hour and 40 minutes.}
    \label{fig-asbps}
\end{figure}

The HMC sample path in Figure \ref{fig-hmc} was essentially unable to make any progress in finding the mode of the target distribution. This is to be expected when the tails of the distribution are so heavy, and the dimension is relatively large.

The first thing the stereographic algorithms do in every case is to gradually make their way towards $N$, then expand the radius of the sphere through the adaptation scheme, then continue going towards $N$ (this can be seen by the plots of the trajectories of $Z_{d+1}$ in Figures \ref{fig-asrw-z}, \ref{fig-asss-z} and \ref{fig-asbps-z}). In the $X_1$ path, this exploration leads to widening oscillations more or less symmetrically around the initial parameter $\mu$.

The algorithms then take varying lengths of time to find the true mean of the target distribution. The SBPS takes longer to expand the sphere, but then finds the mean very quickly, as seen in Figure \ref{fig-asbps-x}. The SRW or SSS both expand the radius very quickly, then require more algorithm time to find the mean, as in Figure \ref{fig-asrw-x} or \ref{fig-asss-x}. 

From there, the behaviour of the algorithms is slightly different. For the SRW and SSS, once the mean is found, they are able to stay there and adapt the $\mu$ parameter to recenter the sphere around the true mean. The process then drops down from the North Pole, but may now need a long time to shrink the $\Sigma$ parameter back down to match the target, leading to a long time stuck at the South Pole. This is clearly seen in the plots of $Z_{d+1}$ in Figure \ref{fig-asrw-z} or \ref{fig-asss-z}.

For the SBPS, the latitude plot appears to suggest it skips this portion of the adaptations and immediately moves the mass onto the equator. However, it transpires that the algorithm may have successfully put the probability mass around $Z_{d+1} \approx 0$, but it had actually focused the target density onto a small corner of the equator: as the adaptations move $\mu$ from $(1000, \dots, 1000)$ to $0_d$, $\Sigma$ scales itself so that the probability mass becomes isolated onto a small subset of the equator. It takes a long time for the process to fully learn $\mu$, at which point $\Sigma$ can be updated to spread the probability mass evenly around the equator.

Since the estimators for the sample mean and covariance are still trying to balance a large number of points near $(1000, \dots, 1000)$ with the incoming samples near $0_d$, the estimators for $\Sigma$ become temporarily highly heterogeneous. This results in a spike in the computational cost of the SBPS just after the process gets near the true mode, as shown in Figure \ref{fig-asbps-eval}. A similar spike is present for the SSS, as seen in Figure \ref{fig-asss-prop}, but the computational cost only goes up by a factor of 2, not the factor of 30 seen in the SBPS. This once again highlights the computational cost of the SBPS when compared with the other algorithms, which was also seen to be a problem in the stationary regime in Section \ref{section-autocorrelation}.

Both of our simulation studies so far suggest that the SSS will outperform both the SRW and SBPS, partly because it is able to perform efficient global moves without the need for extra parametrisation, but also because it is very cheap to simulate: the overall runtime of the SSS was 15 minutes, as opposed to 1 hour 40 minutes for the SBPS or 6 hours 20 minutes for the SRW. For this reason, we currently recommend using the SSS over either of the other two algorithms.

However, these plots also suggest that the uninformed proposals in the SSS cause it to struggle to find the mode, whereas the SBPS is able to do so very quickly once it reaches the appropriate latitude. We therefore expect that certain scenarios must exist in which the SBPS's gradient-based dynamics allow it to outperform the SSS, though we have yet to find a setting in which this overcomes the computational cost attached to the SBPS.

\section{Discussion}

We have discussed the use of adaptive MCMC methods to automatically tune the parameters of our stereographic algorithms. The theoretical results also provide insight into other potential uniformly ergodic algorithms: the segment chain framework can be used going forward to prove convergence results for any adaptive MCMC algorithm based on a simultaneously uniformly ergodic collection of Markov kernels, particularly in the continuous-time setting. This addresses a gap in the current literature, where the focus has predominantly been on discrete-time algorithms.

Alongside theoretical convergence results, we have demonstrated our algorithms' ability to adapt to the target distribution and improve sampling efficiency in problems far more challenging than those covered by the theorems. These examples highlight the power of the stereographic projection as a way of helping sample from heavy-tailed distributions, since standard off-the-shelf methods (namely HMC) have great difficulty tackling these problems when they find themselves in the tails.

All three stereographic algorithms are able to target high-dimensional, heavy-tailed target distributions, and we demonstrated their ability to tune parameters when started in the tails. In terms of their relative mixing properties in the stationary regime, we see that the SSS performs best of the three algorithms. The SBPS is too computationally expensive to compete with it, whereas the SRW is not as robust as the SSS, requires extra tuning, and cannot as easily do global jumps outside of the perfect setting.

We conclude this paper with several further avenues of research:
\begin{itemize}    
    \item[\textit{SBPS Refreshment Rate:}] Several results exist discussing the behaviour of the Euclidean BPS algorithm as one varies the refreshment rate $\lambda_\text{ref}$ \cite{bierkens2017limit,bierkens2022high}. It would be of interest to analyse the SBPS using similar methods to better understand the implications of choices of $\lambda_\text{ref}$ on the mixing properties of the algorithm.
    
    \item[\textit{Other Stereographic Algorithms:}] We have seen that the stereographic projection naturally leads to uniformly ergodic algorithms, even when the target distribution is heavy-tailed. It is therefore natural to consider stereographic versions of other commonly used algorithms, namely MALA or HMC, potentially following the works of \cite{girolami2011riemann} on targeting distributions defined on Riemannian manifolds. However, early attempts at constructing these proposals led to many cases where the proposed moves ``fall off the sphere'' and must be rejected, indicating similar issues to cases where these algorithms have been used to target light-tailed distributions in Euclidean space. Table 1 of \cite{MR3911112} gives a good summary of these points. It would also be of interest to investigate more non-reversible stereographic algorithms, akin to the guided RWM from \cite{gustafson1998guided,kamatani2023non} or the discrete BPS from \cite{sherlock2022discrete}. These could be much cheaper than the SBPS whilst still inheriting the algorithmic benefits of non-reversibility.
    
    \item[\textit{Generalising the Segment Chain:}] Though the segment chain is an elegant way of proving results when the underlying Markov kernels are uniformly ergodic, this is generally too strong an assumption. It is therefore of interest to see whether a segment chain framework could be used under weaker assumptions of geometric or polynomial ergodicity. Initial attempts at doing similar strategies have yielded issues with bounding the return times of the T-skeleton to the small set, since results such as those in \cite{MR1379163} do not allow for bounds or constructions of small sets simultaneously over parameters.
    
    \item[\textit{More Simulation Studies:}] Finally, there is still much investigation required into the behaviour of the stereographic algorithms (adaptive or not) in cases where the target distributions are irregular. It would be of interest to do a more thorough overview of how the stereographic projection improves MCMC-related analysis in real-world applications where heavy-tailed distributions arise. We also have yet to find a case in which the SSS shows any reasonable flaw when compared to the other algorithms.
\end{itemize}

\begin{appendix}

\section{Further details}

\subsection{Excursions of the segment chain}
\label{section-excursions-details}

We discuss the regenerations and excursions of the split segment chain $\{ (\Phi_n, Y_n) \}_{n \in \mathbb{N}}$.

We will write $\nu^*$ for an arbitrary measure on $\widecheck{\Omega}$ such that, for $A \in \mathcal{B}(\mathbb{R}^d \times \mathbb{S}^{d-1})$ and $i =0,1$,
\begin{equation}
    \nu^*( \phi(T) \in A, Y =i ) = \epsilon^i (1-\epsilon)^{1-i} \nu(A),
    \label{equ-mu*}
\end{equation}
i.e., such that $\Phi \sim \nu^* \implies \Phi(T) \sim \nu$ and $Y \sim \text{Bern}(\epsilon)$.

Since the $Y_n$ updates are not affected by the current position $\Phi_n$, we can very easily identify the behaviour of our regeneration times. For $k \geq 1$, the hitting and return times of the atom are
\begin{equation}
    \begin{split}
        \widecheck{\sigma}_1 = \inf( n \geq 0 : Y_n =1), \qquad\widecheck{\sigma}_k = \inf( n > \widecheck{\sigma}_{k-1} : Y_n =1), \\
        \widecheck{\tau}_1 = \inf( n \geq 1 : Y_n =1), \qquad\widecheck{\tau}_k = \inf( n > \widecheck{\tau}_{k-1} : Y_n =1).
    \end{split}
    \label{equ-return-times}
\end{equation}
Alongside Lemma \ref{lemma-atom-1-dependence}, we can concisely write that $\{\Phi_l, Y_l\}_{l \geq \widecheck{\sigma}_k+2}$ has the same law as $\{\Phi_k, Y_k\}_{k \geq 1}$ started from $(\Phi_0,Y_0) \sim \nu^*$.

Thanks to uniform ergodicity, we can obtain the following explicit result on the distribution of the delayed renewal sequence $\{\widecheck{\sigma}_k\}_{k \in \mathbb{N}}$.
\begin{lemma}[Return Times to the Atom]
    For $\{ \Phi_n, Y_n \}_{n \in \mathbb{N}}$ the split segment chain with transition law $\widecheck{P}_\gamma$, let $\{ \widecheck{\sigma}_k \}_{k \in \mathbb{N}}$ be as given in Equation \eqref{equ-return-times}. Then the inter-arrival times $\widecheck{\sigma}_{n+1} - \widecheck{\sigma}_n$ for $n \geq 1$ are iid $\text{Geom}(\epsilon)$ random variables.
    \label{lemma-atom-return-times}
\end{lemma}
Note in particular that since $\epsilon$ does not depend on $\gamma$, neither does the distribution of the inter-arrival times.

Finally, we can define the notion of excursions from the atom, allowing us to divide sample paths into 1-dependent, identically distributed blocks. It is convenient to define the regeneration times
\begin{equation}
    \widecheck{T}_k = \widecheck{\sigma}_k +1,
    \label{equ-regeneration-times}
\end{equation}
such that $\Phi_{\widecheck{T}_k}(T) \sim \nu$ independently of $\{ \Phi_n \}_{n \leq \widecheck{\sigma}_k}$ (and also of $\Phi_{\widecheck{T}_k}(0) = \Phi_{\widecheck{\sigma}_k}(T)$), and so $\Phi_{\widecheck{T}_k +1} \sim \mathbb{Q}_\nu(\cdot \; ; \gamma)$ independently of the pre-$\widecheck{\sigma}_k$ process, as in Lemma \ref{lemma-atom-1-dependence}.

For $k \geq 0$, we define the excursion paths
\begin{equation}
    \Psi_k = \{ \Phi_n, Y_n \}_{n=\widecheck{T}_k +1}^{\widecheck{T}_{k+1}}.
    \label{equ-excursion-path}
\end{equation}
These are random sequences in $\Omega$ of a.s.\ finite length, starting from $\Phi_{\widecheck{T}_k+1}(0) = \Phi_{\widecheck{T}_k}(T)~\sim~\nu$, and ending with $\Phi_{\widecheck{T}_{k+1}}(T) \sim \nu$. By the above discussion, we can see that the variables $\{ \Psi_k \}_{k \in \mathbb{N}}$ are identically distributed and 1-dependent. See Figure \ref{fig-excursions} for a visualisation of how $\widecheck{T}_k$, $\Phi_n$, $\nu$, and $\Psi_k$ all fit together.

We can use \cite[Lemma 5.2]{bednorz2008regeneration}, alongside the minorisation condition \eqref{equ-minorisation-general}, to find the distribution of $\Phi_{\widecheck{\sigma}_k}(T)$.
\begin{lemma}[Stationarity at Renewal Times]
    For $\nu$ and $\eta_\gamma$ as defined in Equation \eqref{equ-minorisation-general} and Equation \eqref{equ-split-remainder-measure}, we have that the unique stationary distribution $\pi_\gamma$ of the original process $X$ is given by
    \begin{equation*}
        \pi_\gamma(\cdot) = \epsilon \sum_{n=0}^\infty (1-\epsilon)^n \nu \eta_\gamma^n(\cdot).
    \end{equation*}
    As an immediate consequence, we get that $\Phi_{\widecheck{\sigma}_k}(T) \sim \pi_\gamma$, for every $k \geq 2$.
    \label{lemma-stationarity-renewals}
\end{lemma}
The first statement is exactly Lemma 5.2 from \cite{bednorz2008regeneration}, but with our notation. The distribution of $\Phi_{\widecheck{\sigma}_k}(T)$ then follows from the geometric distribution of $\widecheck{\sigma}_k - \widecheck{\sigma}_{k-1}$.

\subsection{Assumptions on the Adaptation Times}
\label{section-times-assumptions}

Recall that in assumption \ref{assump-adaptation-times}, we included ``mild conditions'' on the adaptation lags $t_k$. These additional conditions arise via the construction of the segment chain $\{ \Phi_n \}_{n \in \mathbb{N}}$ featuring in our proofs. Since we wish to run the AIR framework on the discrete-time Markov process formed of segments of length $T$, it follows that we must adapt at integer multiples of $T$, so that $t_k = T \times n_k$. 

Thus, we see that we will require the lags in Theorem \ref{thm-fhat-asymptotics} to have a common divisor, at least for large $k$. This condition can be difficult to satisfy if $T$ is unknown, as is the case for the SBPS. Fortunately, since the minorisation condition \eqref{equ-minorisation-general} still holds if you increase $T$, we can consider a sequence of lags $t_k$ such that $\lim_{k \rightarrow \infty}( \text{gcd}(\{ t_j : j \geq k\}) ) = \infty$, such as
\begin{equation}
    t_k = \inf\left( 2^n : 2^n \geq k^\beta \right).
    \label{equ-lags-common-divisor}
\end{equation}
One can check that $t_k = \Theta(k^\beta)$, whilst having increasingly large common divisors. We can then apply the segment chain results once $t_k \geq T$, and consider the first part of the chain as asymptotically irrelevant.

\section{Proof of Theorem \ref{thm-air-segment-chain-1}}
\label{section-air-proof}

In this section, we prove our main result, Theorem \ref{thm-air-segment-chain-1}. We do this by bringing together the 1-dependent, identically distributed excursions $\Psi_k$ discussed in Appendix \ref{section-excursions-details} and controlling the variance of our estimators.

For $i \in \mathbb{N}$, we define the Markov chain $\{ H^{(i)}_j \}_{j \in \mathbb{N}}$ to have transition kernel $\widecheck{P}_{\gamma_i}$, such that
\begin{equation}
    H^{(i)}_j = (\Phi_{N_i+j}, Y_{N_i+j}),
    \label{equ-epoch-chains}
\end{equation}
for $j =0, \dots, n_{i+1} -1$, and $H^{(i)}_j$ evolves independently of $\{(\Phi_n,Y_n)\}_{n \geq N_{i+1}}$ for $j \geq n_{i+1}$.

As in \cite[Section 7]{chimisov2018air}, we let
\begin{equation}
    s_i = \sum_{j=N_i}^{N_{i+1} - 1} F(\Phi_j) = \sum_{j=0}^{n_{i+1}-1} F(H^{(i)}_j),
\end{equation}
where we abuse notation slightly and set $F(\phi,y) = F(\phi)$. Next, we let
\begin{equation}
    \begin{gathered}
        \sigma^{(i)}_1 = \inf( j \geq 0 : H^{(i)}_j \in \widecheck{\alpha}), \quad \sigma^{(i)}_{k+1} = \inf( j > \sigma^{(i)}_k : H^{(i)}_j \in \widecheck{\alpha}), \\
        T^{(i)}_k = \sigma^{(i)}_k +1,
    \end{gathered}
    \label{equ-refreshments-H}
\end{equation}
for $k \geq 1$ be the regeneration times for the $H^{(i)}$ process, defined analogously to $\widecheck{T}_k$ in \eqref{equ-regeneration-times}. We can therefore write these excursions as
\begin{equation}
    \Psi^{(i)}_j = \{ H^{(i)}_n \}_{n = T^{(i)}_j+1}^{T^{(i)}_{j+1}}.
    \label{equ-excursions-H}
\end{equation}
By Lemma \ref{lemma-atom-1-dependence}, these excursions are identically distributed for $i$ fixed and, conditionally on $\gamma_i$, they are 1-dependent.

We also define
\begin{equation}
    L_i = \inf(l \geq 1 : T^{(i)}_l \geq n_{i+1} -1) = 1+ \sum_{j=0}^{n_{i+1} -2} \mathbbm{1}(H^{(i)}_j \in \widecheck{\alpha}),
    \label{equ-leftover-excursion}
\end{equation}
such that the last excursion of $H^{(i)}$ to start before we decouple from $(\Phi,Y)$ ends at step $T^{(i)}_{L_i}$.

Finally, we can define the sums over excursions which we will be analysing
\begin{equation}
    \begin{split}
        \eta_i = \sum_{j=0}^{T^{(i)}_1} F(H^{(i)}_j), &\qquad \xi_i = \sum_{j = T^{(i)}_1 +1}^{T^{(i)}_{L_i}} F(H^{(i)}_j) \\
        \zeta_i = \sum_{j = n_{i+1}}^{T^{(i)}_{L_i}} F(H^{(i)}_j), &\qquad \xi_{i,j} = \sum_{m = T^{(i)}_j +1}^{T^{(i)}_{j+1}} F(H^{(i)}_m).
    \end{split}
    \label{equ-sum-terms}
\end{equation}
For fixed $i$, since $\xi_{i,j}$ are functions of $\Psi^{(i)}_j$ only, these are identically distributed, 1-dependent (conditionally on $\gamma_i$) random variables.

With these, we get that
\begin{equation}
    s_i = \sum_{j=N_i}^{N_{i+1} - 1} F(\Phi_j) = \eta_i + \xi_i - \zeta_i.
    \label{equ-sum-split}
\end{equation}
$\eta_i$ is a function of the $H^{(i)}$ path up until the first regeneration, and $\zeta_i$ is a function of the $H^{(i)}$ path from the point where we decouple from $(\Phi,Y)$ to the next regeneration. For large $i$, the vast majority of the sample path will be within the $\xi_i$ term. The majority of the asymptotic results will then follow by studying $\xi_i = \sum_{j=1}^{L_i-1} \xi_{i,j}$, with $\xi_{i,j}$ identically distributed, conditionally 1-dependent random variables.

For the empirical sums $S_N = \sum_{j=0}^{N-1} F(\Phi_j)$, we can find $k = k(N)$ such that $N_k \leq N < N_{k+1}$ and write
\begin{align}
    S_N =& \sum_{i=0}^k s_i + \sum_{j=N_k}^{N-1} F(\Phi_j) \notag \\
    =& \sum_{i=0}^{k-1} \eta_i + \sum_{i=0}^{k-1} \xi_i - \sum_{i=0}^{k-1} \zeta_i + \sum_{j=N_k}^{N-1} F(\Phi_j), \notag \\
    =& \Xi^{(1)}_k + \Xi^{(2)}_k + \Xi^{(3)}_k + \Xi^{(4)}_{k,N}.
    \label{equ-full-sum-split}
\end{align}

Since our estimator of interest is $\widehat{f}^\Phi_n = \frac{1}{n}S_n$, it only remains to study the asymptotic behaviour of $\Xi^{(m)}_k$, $m \in \{1,2,3\}$, and $\Xi^{(4)}_{k,N}$ as $N$ becomes large and the adaptive parameter $\gamma_i$ varies.

We follow the steps of the proofs found in \cite[Section 7]{chimisov2018air}. Our scenario has the advantage that, since our original minorisation condition \eqref{equ-minorisation-general} was uniform over the sample space, our ergodic atom has geometric return times. However, we have the difficulty that our $(\Phi,Y)$ regenerations are 2-step regenerations, rather than 1-step regenerations, which means that our excursions are 1-dependent, not independent. 

We write expectations of the form $\mathbb{E}^{\gamma_0}_{(\phi,y)}$ as expectations with initial value $(\phi,y)$ and initial adaptive parameter $\gamma_0$. If integrands are independent of the initial values $(\Phi_0, Y_0)$, we shall simply write $\mathbb{E}^\gamma$. The underlying adaptation scheme that updates $\gamma_k$ is then assumed to be included in the relevant expectations.

\subsection[Asymptotic Behaviour of Xi{(2)}k]{Asymptotic Behaviour of $\Xi^{(2)}_k$}

Since most of the sample path will be contained in the $\Xi^{(2)}_k$ term, we start by studying its properties. The other terms can be controlled relatively easily thanks to the boundedness of $f$ and the geometric tails of the distribution of the excursion lengths.

We define the asymptotic variance function as
\begin{equation}
    \sigma^2(\gamma) = \epsilon\mathbb{E}^\gamma(\xi_{1,1}^2) + 2\epsilon \mathbb{E}^\gamma ( \xi_{1,1} \xi_{1,2} ),
    \label{equ-chain-variance-def}
\end{equation}
not to be confused with the stopping times $\sigma^{(i)}_j$. We have that
\begin{equation}
    \begin{split}
        \mathbb{E}^\gamma(\xi_{1,1}^2) = \mathbb{E}^\gamma\left( \left[ \sum_{m = \widecheck{T}_1 + 1}^{\widecheck{T}_2} F(\Phi_m) \right]^2 \right) =& \mathbb{E}^\gamma\left( \left[ \sum_{m = \widecheck{\tau}_1 + 2}^{\widecheck{\tau}_2 +1} F(\Phi_m) \right]^2 \right) \\ 
        =& \mathbb{E}_{\nu^*}^\gamma\left( \left[ \sum_{m = 1}^{\widecheck{\sigma}_1+1} F(\Phi_m) \right]^2 \right),
    \end{split}
    \label{equ-xi^2-expectation}
\end{equation}
so since $f$ is bounded and $\widecheck{\sigma}_1 \sim \text{Geom}(\epsilon)$, we can use the Cauchy-Schwartz inequality to get
\begin{equation}
    \sup_{\gamma \in \Gamma} \sigma^2(\gamma) < \infty.
    \label{equ-bounded-variance-assumption}
\end{equation}

We start by giving explicit expressions for the moments of $\xi_i$.
\begin{lemma}[Moments of $\xi_i$]
    Consider $\xi_{i,j}$ and $\xi_i$ as given in Equation \eqref{equ-sum-terms}. Then, if $\mathbb{E}_{\pi_\gamma}(f(X)) = 0$,
    \begin{gather*}
        \mathbb{E}^{\gamma_0}( \xi_{i,j} \mid \gamma_i ) = \mathbb{E}^{\gamma_0}( \xi_i \mid \gamma_i ) = 0, \\[4pt]
        \mathbb{E}^{\gamma_0}( \xi_i^2 \mid \gamma_i) =(n_{i+1} -1) \sigma^2(\gamma_i) -2 \mathbb{E}^{\gamma_i} \left( \xi_{1,1} \xi_{1,2} \right), \\[4pt]
        \mathbb{E}^{\gamma_0}( \xi_i^4 \mid \gamma_i) = \mathcal{O}(n_{i+1}^2),
    \end{gather*}
    where the $\mathcal{L}^4$ bound holds uniformly over $\gamma_i$.
    \label{lemma-xi-2-moments}
\end{lemma}

The first statement follows from standard relations between excursions and stationary measures. The second statement follows from Wald's identity. The third statement follows from a martingale argument using Doob's inequality.

\begin{proof}[Proof of Lemma \ref{lemma-xi-2-moments}]
    For $i \geq 0$ fixed, recall that $\xi_{i,j}$ are identically distributed random variables. We therefore have that
    \begin{align*}
        \mathbb{E}^{\gamma_0}(\xi_{i,j} \mid \gamma_i) =& \mathbb{E}^{\gamma_0}\left( \sum_{m = T^{(i)}_j +1}^{T^{(i)}_{j+1}} F(H^{(i)}_m) \mathrel{} \middle| \mathrel{} \gamma_i \right) \\
        =& \mathbb{E}^{\gamma_i}\left( \sum_{m = \widecheck{T}_j +1}^{\widecheck{T}_{j+1}} F(\Phi_m) \right),
    \end{align*}
    must be constant over $j$. By Theorem 10.0.1 of \cite{meyn2012markov}, we have that
    \begin{equation*}
        \mathbb{E}_{\widecheck{\pi}_\gamma}^\gamma \left( \mathbbm{1}(Y_0=1) \sum_{m = 1}^{\widecheck{\tau}_1} F(\Phi_m) \right) = \mathbb{E}_{\widecheck{\pi}_\gamma}(F(\Phi)) = 0.
    \end{equation*}
    Writing $\widecheck{\pi} \rvert_{\widecheck{\alpha}}(\cdot) = \epsilon^{-1} \widecheck{\pi}_\gamma(\cdot \cap \{Y=1\})$, we rewrite this as
    \begin{equation}
        \mathbb{E}_{\widecheck{\pi} \rvert_{\widecheck{\alpha}}}^\gamma \left( \sum_{m = 1}^{\widecheck{\tau}_1} F(\Phi_m) \right) = \epsilon^{-1} \mathbb{E}_{\widecheck{\pi}_\gamma}^\gamma\left( \mathbbm{1}(Y_0=1) \sum_{m = 1}^{\widecheck{\tau}_1} F(\Phi_m) \right) = 0.
        \label{equ-centered-excursion-stationary}
    \end{equation}

    Exactly as in Equation \eqref{equ-xi^2-expectation}, for any initial distribution $\mu$,
    \begin{equation*}
        \mathbb{E}_\mu^\gamma\left( \sum_{m = \widecheck{T}_j + 1}^{\widecheck{T}_{j+1}} F(\Phi_m) \right)  = \mathbb{E}_{\nu^*}^\gamma\left( \sum_{m = 1}^{\widecheck{\sigma}_1+1} F(\Phi_m) \right).
    \end{equation*}
    Comparing this with Equation \eqref{equ-centered-excursion-stationary}, we see that
    \begin{align*}
        \mathbb{E}_{\nu^*}^\gamma\left( \sum_{m = 1}^{\widecheck{\sigma}_1+1} F(\Phi_m) \right) =& \mathbb{E}_{\widecheck{\pi} \rvert_{\widecheck{\alpha}}}^\gamma\left( \sum_{m = \widecheck{\sigma}_1 + 2}^{\widecheck{\sigma}_2 +1} F(\Phi_m) \right) \\[4pt]
        =& \mathbb{E}_{\widecheck{\pi} \rvert_{\widecheck{\alpha}}}^\gamma\left( \sum_{m = 1}^{\widecheck{\tau}_1} F(\Phi_m) \right) + \mathbb{E}_{\widecheck{\pi} \rvert_{\widecheck{\alpha}}}^\gamma\left( F(\Phi_{\widecheck{\sigma}_2 +1}) \right) - \mathbb{E}_{\widecheck{\pi} \rvert_{\widecheck{\alpha}}}^\gamma\left( F(\Phi_1) \right) \\[4pt]
        =& \mathbb{E}_{\widecheck{\pi} \rvert_{\widecheck{\alpha}}}^\gamma\left( F(\Phi_{\widecheck{T}_2}) \right) - \mathbb{E}_{\widecheck{\pi} \rvert_{\widecheck{\alpha}}}^\gamma\left( F(\Phi_1) \right).
    \end{align*}

    However, Lemma \ref{lemma-stationarity-renewals} tells us that, given $\Phi_0(T) \sim \pi_\gamma$ and $Y_0 = 1$, we have that $\Phi_1$ and $\Phi_{\widecheck{T}_2}$ have the same distribution. Hence,
    \begin{equation}
        \mathbb{E}_{\nu^*}^\gamma\left( \sum_{m = 1}^{\widecheck{\sigma}_1+1} F(\Phi_m) \right) = 0
        \quad \implies \quad \mathbb{E}^{\gamma_0}(\xi_{i,j} \mid \gamma_i) =0.
        \label{equ-xi_(i,j)-mean}
    \end{equation}

    For $\mathbb{E}^{\gamma_0}(\xi_i \mid \gamma_i) = \mathbb{E}^{\gamma_0}\left( \sum_{j=1}^{L_i -1} \xi_{i,j} \mathrel{}\middle|\mathrel{} \gamma_i \right)$, we use Wald's identity. For this, we require that
    \begin{equation}
        \mathbb{E}^{\gamma_0}( \xi_{i,j} \mathbbm{1}(L_i -1 \geq j) \mid \gamma_i ) = \mathbb{E}^{\gamma_0}(\xi_{i,j} \mid \gamma_i) \mathbb{P}^{\gamma_0}(L_i -1 \geq j),
        \label{equ-wald-condition}
    \end{equation}
    for every $j$ (recalling that $\sigma^{(i)}_j$ are independent of $\gamma_i$, so $L_i$ is also).

    However, $\{L_i -1 \geq j\} = \{ T^{(i)}_j < n_{i+1} -1 \} = \{ \sigma^{(i)}_j \leq n_{i+1} -3 \}$. Since $\xi_{i,j}$ is a function of the excursion $\Psi^{(i)}_j$, it is independent of the event $\{ \sigma^{(i)}_j \leq n_{i+1} -3 \}$, giving the required condition.
    
    Thus, we can apply Wald's identity to get that
    \begin{equation}
        \mathbb{E}^{\gamma_0}(\xi_i \mid \gamma_i) = \mathbb{E}^{\gamma_0}\left( \sum_{j=1}^{L_i -1} \xi_{i,j} \mathrel{}\middle|\mathrel{} \gamma_i \right) = \mathbb{E}^{\gamma_0}(L_i -1) \times \mathbb{E}^{\gamma_0}(\xi_{i,1} \mid \gamma_i) =0.
        \label{equ-wald-identity-xi}
    \end{equation}
    
    For $\mathbb{E}^{\gamma_0}( \xi_i^2 \mid \gamma_i)$, we follow an argument similar to the standard proof of Wald's identity:
    \begin{align*}
        \mathbb{E}^{\gamma_0}( \xi_i^2 \mid \gamma_i) =& \mathbb{E}^{\gamma_0} \left( \left[ \sum_{j=1}^{L_i-1} \xi_{i,j} \right]^2 \mathrel{}\middle|\mathrel{} \gamma_i \right) \\
        =& \mathbb{E}^{\gamma_0} \left( \left[ \sum_{j=1}^\infty \xi_{i,j} \mathbbm{1}(\sigma^{(i)}_j \leq n_{i+1} -3) \right]^2 \mathrel{}\middle|\mathrel{} \gamma_i \right) \\
        =& \sum_{j=1}^\infty \mathbb{E}^{\gamma_0} \left( \xi_{i,j}^2 \mathbbm{1}(\sigma^{(i)}_j \leq n_{i+1} -3) \mathrel{}\middle|\mathrel{} \gamma_i \right) \\
        &+ 2 \sum_{j <k} \mathbb{E}^{\gamma_0} \left( \xi_{i,j} \xi_{i,k} \mathbbm{1}(\sigma^{(i)}_j \leq n_{i+1} -3) \mathbbm{1}(\sigma^{(i)}_k \leq n_{i+1} -3) \mathrel{}\middle|\mathrel{} \gamma_i \right) \\
        =& \sum_{j=1}^\infty \mathbb{E}^{\gamma_0} \left( \xi_{i,j}^2 \mathbbm{1}(\sigma^{(i)}_j \leq n_{i+1} -3) \mathrel{}\middle|\mathrel{} \gamma_i \right) \\
        &+ 2 \sum_{j <k} \mathbb{E}^{\gamma_0} \left( \xi_{i,j} \xi_{i,k} \mathbbm{1}(\sigma^{(i)}_k \leq n_{i+1} -3) \mathrel{}\middle|\mathrel{} \gamma_i \right),
    \end{align*}
    since $j <k \implies \sigma^{(i)}_j < \sigma^{(i)}_k$. The first sum, using Wald's identity, can be simplified to
    \begin{equation}
        \sum_{j=1}^\infty \mathbb{E}^{\gamma_0} \left( \xi_{i,j}^2 \mathbbm{1}(\sigma^{(i)}_j \leq n_{i+1} -3) \mathrel{}\middle|\mathrel{} \gamma_i \right) = \mathbb{E}^{\gamma_0}(L_i -1) \times \mathbb{E}^{\gamma_0}(\xi_{i,1}^2 \mid \gamma_i ).
        \label{equ-wald-identity-xi^2}
    \end{equation}
    For the second term, if we condition on the $\sigma$-algebra $\mathcal{F}_{T^{(i)}_{j+1}}$ and $k > j+1$, then the 1-dependent structure of the excursions gives that $ \xi_{i,k}$ is conditionally independent of $\xi_{i,j} \mathbbm{1}(\sigma^{(i)}_k \leq n_{i+1} -3)$ and these expectations must be 0. Thus,
    \begin{equation}
        \sum_{j <k} \mathbb{E}^{\gamma_0} \left( \xi_{i,j} \xi_{i,k} \mathbbm{1}(\sigma^{(i)}_k \leq n_{i+1} -3) \mathrel{}\middle|\mathrel{} \gamma_i \right) = \sum_{j =1}^\infty \mathbb{E}^{\gamma_0} \left( \xi_{i,j} \xi_{i,j+1} \mathbbm{1}(\sigma^{(i)}_{j+1} \leq n_{i+1} -3) \mathrel{}\middle|\mathrel{} \gamma_i \right).
        \label{equ-xi_j-xi_k-covariance}
    \end{equation}

    Since $\xi_{i,j+1}$ is independent of $\{ H^{(i)}_m \}_{m =0}^{\sigma^{(i)}_{j+1}}$ conditionally on $\gamma_i$, we have that
    \begin{align*}
        \mathbb{E}^{\gamma_0} &\left( \xi_{i,j} \xi_{i,j+1} \mathbbm{1}(\sigma^{(i)}_{j+1} \leq n_{i+1} -3) \mathrel{}\middle|\mathrel{} \gamma_i \right) \\
        =& \mathbb{E}^{\gamma_0} \left( F\left( \Phi_{T^{(i)}_{j+1}} \right) \xi_{i,j+1} \mathbbm{1}(\sigma^{(i)}_{j+1} \leq n_{i+1} -3) \mathrel{}\middle|\mathrel{} \gamma_i \right) \\ 
        &+ \mathbb{E}^{\gamma_0}( \xi_{i,j+1} \mid \gamma_i ) \mathbb{E}^{\gamma_0}\left(\mathbbm{1}(\sigma^{(i)}_{j+1} \leq n_{i+1} -3)\left[ \xi_{i,j}- F\left( \Phi_{T^{(i)}_{j+1}} \right) \right] \mathrel{}\middle|\mathrel{} \gamma_i \right) \\
        =& \mathbb{E}^{\gamma_0} \left( F\left( \Phi_{T^{(i)}_{j+1}} \right) \xi_{i,j+1} \mathbbm{1}(\sigma^{(i)}_{j+1} \leq n_{i+1} -3) \mathrel{}\middle|\mathrel{} \gamma_i \right).
    \end{align*}
    However, if we condition on the $\sigma$-algebra $\mathcal{F}^Y_{\sigma^{(i)}_{j+1}}$, the $\sigma$-algebra generated by the $Y$ components of the sequence $H^{(i)}_m$ up to the stopping time $\sigma^{(i)}_{j+1}$, then from Lemma \ref{lemma-stationarity-renewals} we know that the distribution of $\Phi_{T^{(i)}_{j+1}}$ is a bridge segment starting according to $\pi_\gamma$ and ending according to $\nu$. This conditional expectation is therefore constant over $j$. Thus, we can write
    \begin{align*}
        \mathbb{E}^{\gamma_0} \left( \xi_{i,j} \xi_{i,j+1} \mathbbm{1}(\sigma^{(i)}_{j+1} \leq n_{i+1} -3) \mathrel{}\middle|\mathrel{} \gamma_i \right) =& \mathbb{E}^{\gamma_0} \left( F\left( \Phi_{T^{(i)}_{j+1}} \right) \xi_{i,j+1} \mathrel{}\middle|\mathrel{} \gamma_i \right) \\
        &\times \mathbb{P}^{\gamma_0}(\sigma^{(i)}_{j+1} \leq n_{i+1} -3) \\
        =& \mathbb{E}^{\gamma_0} \left( \xi_{i,j} \xi_{i,j+1} \mathrel{}\middle|\mathrel{} \gamma_i \right) \mathbb{P}^{\gamma_0}(L_i - 1 \geq j+1),
    \end{align*}
    since $\xi_{i,j+1}$ is conditionally independent of $\{ H^{(i)}_m \}_{m =0}^{\sigma^{(i)}_{j+1}}$.
    
    Thus, we get that
    \begin{equation}
        \sum_{j =1}^\infty \mathbb{E}^{\gamma_0} \left( \xi_{i,j} \xi_{i,j+1} \mathbbm{1}(\sigma^{(i)}_{j+1} \leq n_{i+1} -3) \mathrel{}\middle|\mathrel{} \gamma_i \right) = \mathbb{E}^{\gamma_0} \left( \xi_{i,1} \xi_{i,2} \mathrel{}\middle|\mathrel{} \gamma_i \right) \mathbb{E}^{\gamma_0}(L_i - 2).
        \label{equ-wald-identity-xi_j-xi_j+1}
    \end{equation}

    Combining Equations \eqref{equ-wald-identity-xi^2} and \eqref{equ-wald-identity-xi_j-xi_j+1} gives
    \begin{equation}
        \begin{split}
            \mathbb{E}^{\gamma_0}( \xi_i^2 \mid \gamma_i ) =& \mathbb{E}^{\gamma_0}(\xi_{i,1}^2\mid \gamma_i ) \times \mathbb{E}^{\gamma_0}(L_i -1) \\
            &+ 2 \mathbb{E}^{\gamma_0} \left( \xi_{i,1} \xi_{i,2} \mid \gamma_i \right) \times \mathbb{E}^{\gamma_0}(L_i - 2).
        \end{split}
        \label{equ-exp(xi^2)-decomp}
    \end{equation}
    For $\mathbb{E}^{\gamma_0}(L_i)$, we have that
    \begin{equation}
        \mathbb{E}^{\gamma_0}(L_i) = 1 + \sum_{j=0}^{n_{i+1} -2} \mathbb{P}^{\gamma_0}(H^{(i)}_j \in \widecheck{\alpha}) = 1+ \epsilon (n_{i+1} -1),
        \label{equ-L_i-expectation}
    \end{equation}
    for $i \geq 2$. The fixed value of $Y_0 =y$ will slightly skew the expectation of $L_1$, but this is asymptotically irrelevant, or can be removed by taking $Y_0 \sim \text{Bern}(\epsilon)$.

    We therefore get that
    \begin{align}
        \mathbb{E}^{\gamma_0}( \xi_i^2 \mid \gamma_i ) =& \epsilon (n_{i+1} -1) \mathbb{E}^{\gamma_0}(\xi_{i,1}^2 \mid \gamma_i ) \notag \\
        &+ 2 [\epsilon (n_{i+1} -1) - 1] \mathbb{E}^{\gamma_0} ( \xi_{i,1} \xi_{i,2} \mid \gamma_i ) \notag \\[4pt]
        =& (n_{i+1} -1) \sigma^2(\gamma_i) -2 \mathbb{E}^{\gamma_i} ( \xi_{1,1} \xi_{1,2} ).
        \label{equ-exp(xi^2)-decomp-2}
    \end{align}

    This leaves $\mathbb{E}^{\gamma_0}( \xi_i^4 \mid \gamma_i )$, for which we rely on a crude upper bound using Doob's $\mathcal{L}^p$ inequality. We have that
    \begin{equation}
        \begin{split}
            \mathbb{E}^{\gamma_0}( \xi_i^4 \mid \gamma_i ) =& \mathbb{E}^{\gamma_0}\left( \left[ \sum_{j=1}^{L_i-1} \xi_{i,j} \right]^4 \mathrel{}\middle|\mathrel{} \gamma_i \right) \\
            \leq& \mathbb{E}^{\gamma_0}\left( \max_{m = 0, \dots n_{i+1} -1} \left[ \sum_{j=1}^m \xi_{i,j} \right]^4 \mathrel{}\middle|\mathrel{} \gamma_i \right).
        \end{split}
        \label{equ-xi^4-max-bound}
    \end{equation}
    Since $\xi_{i,j}$ are conditionally 1-dependent, $\{ \sum_{j=1}^m \xi_{i,j} \}_{m \geq 0}$ is not a martingale, even conditionally on $\gamma_i$. As in \cite[Section 4]{bednorz2008regeneration}, we can divide this sum into odd and even terms since
    \begin{equation}
        \left[ \sum_{j=1}^m \xi_{i,j} \right]^4 \leq 2^3 \left( \left[ \sum_{\substack{{j=1} \\ {j \text{ odd}}}}^m \xi_{i,j} \right]^4 + \left[ \sum_{\substack{{j=1} \\ {j \text{ even}}}}^m \xi_{i,j} \right]^4 \right),
        \label{equ-xi^4-odds-evens}
    \end{equation}
    by Jensen's inequality. These two sums are martingales because, conditionally on $\gamma_i$, $\{ \xi_{i,2j} \}_{j \geq 1}$ and $\{ \xi_{i,2j-1} \}_{j \geq 1}$ are iid sequences of mean $0$ random variables. So we get that
    \begin{align}
        \mathbb{E}^{\gamma_0}( \xi_i^4 \mid \gamma_i ) \leq& 8\mathbb{E}^{\gamma_0}\left( \max_{m = 0, \dots n_{i+1} -1} \left[ \sum_{\substack{{j=1} \\ {j \text{ odd}}}}^m \xi_{i,j} \right]^4 \mathrel{}\middle|\mathrel{} \gamma_i \right) \notag \\
        &+ 8\mathbb{E}^{\gamma_0}\left( \max_{m = 0, \dots n_{i+1} -1} \left[ \sum_{\substack{{j=1} \\ {j \text{ even}}}}^m \xi_{i,j} \right]^4 \mathrel{}\middle|\mathrel{} \gamma_i \right) \notag \\
        \leq& 8 \left( \frac{4}{3} \right)^4 \times \mathbb{E}^{\gamma_0}\left( \left[ \sum_{\substack{{j=1} \\ {j \text{ odd}}}}^{n_{i+1}-1} \xi_{i,j} \right]^4 + \left[ \sum_{\substack{{j=1} \\ {j \text{ even}}}}^{n_{i+1}-1} \xi_{i,j} \right]^4 \mathrel{}\middle|\mathrel{} \gamma_i \right),
        \label{equ-doob-inequality}
    \end{align}
    by Doob's $\mathcal{L}^p$ inequality.

    By conditional independence, we have that
    \begin{align*}
        \mathbb{E}^{\gamma_0}\left( \left[ \sum_{\substack{{j=1} \\ {j \text{ odd}}}}^{n_{i+1}-1} \xi_{i,j} \right]^4 \mathrel{}\middle|\mathrel{} \gamma_i \right) =& \left\lceil \frac{n_{i+1} - 1}{2} \right\rceil \mathbb{E}^{\gamma_i}(\xi_{1,1}^4) \\
        +& \left\lceil \frac{n_{i+1} - 1}{2} \right\rceil \left( \left\lceil \frac{n_{i+1} - 1}{2} \right\rceil -1 \right) \mathbb{E}^{\gamma_i}(\xi_{1,1}^2)^2, \\
        \mathbb{E}^{\gamma_0}\left( \left[ \sum_{\substack{{j=1} \\ {j \text{ even}}}}^{n_{i+1}-1} \xi_{i,j} \right]^4 \mathrel{}\middle|\mathrel{} \gamma_i \right) =& \left\lfloor \frac{n_{i+1} - 1}{2} \right\rfloor \mathbb{E}^{\gamma_i}(\xi_{1,1}^4) \\
        &+ \left\lfloor \frac{n_{i+1} - 1}{2} \right\rfloor \left( \left\lfloor \frac{n_{i+1} - 1}{2} \right\rfloor -1 \right) \mathbb{E}^{\gamma_i}(\xi_{1,1}^2)^2.
    \end{align*}
    Substituting these into Equation \eqref{equ-doob-inequality} simplifies to give
    \begin{align}
        \mathbb{E}^{\gamma_0}( \xi_i^4 \mid \gamma_i ) \leq& \frac{2048}{81} \left( (n_{i+1} -1) \mathbb{E}^{\gamma_i}(\xi_{1,1}^4) + \frac{(n_{i+1} - 2)^2}{2} \mathbb{E}^{\gamma_i}(\xi_{1,1}^2)^2 \right) \notag \\
        =& \mathcal{O}(n_{i+1}^2),
        \label{equ-xi^4-expectation}
    \end{align}
    uniformly over $\gamma_i$, assuming that $\sup_{\gamma \in \Gamma} \left( \mathbb{E}^\gamma( \xi_{1,1}^4 ) \right) < +\infty$, which follows from the fact that $f$ is bounded and $\widecheck{\sigma}_1 \sim \text{Geom}(\epsilon)$.
\end{proof}

For the asymptotic results, the vast majority of our process is contained within the $\Xi^{(2)}_k$ terms, so we shall start by considering those.
\begin{lemma}[$\mathcal{L}^2$ Convergence of $\Xi^{(2)}$]
    Consider $\Xi^{(2)}_k$ as given in Equation \eqref{equ-full-sum-split}. Then, assuming $\mathbb{E}_{\pi_\gamma} (f(X)) = 0$ for all $\gamma \in \Gamma$, we have that $\forall (\phi,y) \in \widecheck{\Omega}$ and $\gamma_0 \in \Gamma$,
    \begin{equation*}
        \mathbb{E}^{\gamma_0}\left( \Xi^{(2)}_k \right) = 0,
    \end{equation*}
    and
    \begin{equation*}
        \mathbb{E}^{\gamma_0}\left( \left[\frac{1}{N_K} \Xi^{(2)}_k \right]^2 \right) = \mathcal{O}\left( \frac{1}{N_k} \right),
    \end{equation*}
    as $k \rightarrow +\infty$.
    \label{lemma-L2-Xi2}
\end{lemma}

The proof of this Lemma follows relatively easily from Lemma \ref{lemma-xi-2-moments}.

\begin{proof}[Proof of Lemma \ref{lemma-L2-Xi2}]
    The fact that $\mathbb{E}^{\gamma_0}\left( \Xi^{(2)}_k \right) = 0$ follows directly from the fact that $\mathbb{E}^{\gamma_0}( \xi_i) = 0$.
    
    For $\mathbb{E}^{\gamma_0}\left( \left[\frac{1}{N_k} \Xi^{(2)}_k \right]^2 \right)$, we start by noting that
    \begin{equation}
        \mathbb{E}^{\gamma_0}( \xi_i \xi_j ) =0,
        \label{equ-xi_i-uncorrelated}
    \end{equation}
    for every $i \neq j$, which follows from a conditioning argument. Thus,
    \begin{equation}
        \mathbb{E}^{\gamma_0}\left( \left( \Xi^{(2)}_k \right)^2 \right) = \sum_{i=0}^{k-1} \mathbb{E}^{\gamma_0}\left( \xi_i^2 \right).
        \label{equ-Xi2-sum-xi^2}
    \end{equation}
    
    Since $\sigma^2(\gamma)$ and $\mathbb{E}^\gamma(\xi_{1,1}^2)$ are bounded uniformly over $\gamma \in \Gamma$, we have that
    \begin{align}
        \mathbb{E}^{\gamma_0}( \xi_i^2 ) \leq& (n_{i+1} -1) \mathbb{E}^{\gamma_0}\left( \sigma^2(\gamma_i) \right) + 2 \left\lvert \mathbb{E}^{\gamma_0} \left( \xi_{i,1} \xi_{i,2} \right) \right\rvert \notag \\[4pt]
        \leq& (n_{i+1} -1) \sup_{\gamma \in \Gamma} \left( \sigma^2(\gamma) \right) + 2 \sup_{\gamma \in \Gamma} \left( \mathbb{E}^\gamma(\xi_{1,1}^2) \right).
        \label{equ-xi^2-expectation-bound}
    \end{align}
    By Equation \eqref{equ-Xi2-sum-xi^2}, this gives
    \begin{equation}
        \begin{split}
            \mathbb{E}^{\gamma_0}\left( \left( \Xi^{(2)}_k \right)^2 \right) = \sum_{i=0}^{k-1} \mathbb{E}^{\gamma_0}\left( \xi_i^2 \right) \leq& (N_k -k) \sup_{\gamma \in \Gamma} \left( \sigma^2(\gamma) \right) \\
            &+ 2k \sup_{\gamma \in \Gamma} \left( \mathbb{E}^\gamma(\xi_{1,1}^2) \right),
        \end{split}
        \label{equ-Xi^2-expectation-bound}
    \end{equation}
    so
    \begin{align}
        \mathbb{E}^{\gamma_0}\left( \left( \frac{1}{N_k} \Xi^{(2)}_k \right)^2 \right) \leq& \sup_{\gamma \in \Gamma} \left( \sigma^2(\gamma) \right) \left[ \frac{k}{N_k^2} +  \frac{1}{N_k} \right] \notag \\[4pt]
        &+ \frac{2k}{N_k^2} \sup_{\gamma \in \Gamma} \left( \mathbb{E}^\gamma(\xi_{1,1}^2) \right) \notag \\[4pt]
        =& \mathcal{O}\left( \frac{1}{N_k} \right), \label{equ-Xi^2/N-order}
    \end{align}
    as required, since $N_k \geq k$.
\end{proof}

With this bound on second moments, we can move on to a CLT for $\Xi^{(2)}_k$.
\begin{lemma}[Central Limit Theorem for $\Xi^{(2)}_k$]
    Consider $\Xi^{(2)}_k$ as given in Equation \eqref{equ-full-sum-split}. Assume $\mathbb{E}_\pi (f(X)) = 0$, and that $\sigma^2(\gamma)$ is a continuous function of $\gamma$. Suppose also that $\gamma_i \xrightarrow{P} \gamma_\infty$ as $i \rightarrow +\infty$, where $\gamma_\infty$ is constant. Then, if $\sigma^2(\gamma_\infty) >0$, we have that $\forall (\phi,y) \in \widecheck{\Omega}, \gamma_0 \in \Gamma$,
    \begin{equation*}
        \frac{1}{\sqrt{N_k}} \Xi^{(2)}_k \xrightarrow{D} \mathcal{N}\left( 0_d, \sigma^2(\gamma_\infty) \right),
    \end{equation*}
    as $k \rightarrow +\infty$.
    \label{lemma-clt-Xi2}
\end{lemma}

We will show later, using results following from Lemma \ref{lemma-L2-Xi2} and others, that the assumptions on the convergence of $\gamma_i$ are not particularly restrictive for well-chosen parameter estimators. 

We follow a similar proof to the CLT proof for $\Xi^{(2)}_k$ in \cite{chimisov2018air}, based on Theorem 2.2 from \cite{dvoretzky1972asymptotic}.

\begin{proof}[Proof of Lemma \ref{lemma-clt-Xi2}]
    We apply Theorem 2.2 from \cite{dvoretzky1972asymptotic} with $X_{n,k} = \frac{1}{\sqrt{N_n \sigma^2(\gamma_\infty)}} \xi_k$, though our indices do not quite line up with theirs.
    
    Let
    \begin{equation}
        \Tilde{\mathcal{F}}_{-1} = \sigma(\gamma_0), \quad \Tilde{\mathcal{F}}_i = \sigma\left( \Tilde{\mathcal{F}}_{i-1} \cup \{ \xi_{i,j} \}_{j \geq 1} \cup \{ \gamma_{i+1} \} \right),
        \label{equ-sigma-algebra}
    \end{equation}
    for $i \geq 0$. From Lemma \ref{lemma-xi-2-moments}, we have $\mathbb{E}^{\gamma_0}\left( \xi_i \mid \Tilde{\mathcal{F}}_{i-1} \right) = \mathbb{E}^{\gamma_0}\left( \xi_i \mid \gamma_i \right) =0$, which gives condition (2.3) from \cite{dvoretzky1972asymptotic}.

    Also by Lemma \ref{lemma-xi-2-moments}, we have that
    \begin{equation}
        \mathbb{E}^{\gamma_0}\left( \xi_i^2 \mid \Tilde{\mathcal{F}}_{i-1} \right) = \mathbb{E}^{\gamma_0}\left( \xi_i^2 \mid \gamma_i \right) = (n_{i+1} -1) \sigma^2(\gamma_i) -2 \mathbb{E}^{\gamma_i} \left( \xi_{1,1} \xi_{1,2} \right).
        \label{equ-xi-variance-conditional}
    \end{equation}
    This gives
    \begin{equation}
        \sum_{i=0}^{k-1} \mathbb{E}^{\gamma_0}\left( \xi_i^2 \mid \Tilde{\mathcal{F}}_{i-1} \right) = \sum_{i=0}^{k-1} (n_{i+1} -1) \sigma^2(\gamma_i) - 2 \sum_{i=0}^{k-1} \mathbb{E}^{\gamma_i} \left( \xi_{1,1} \xi_{1,2} \right).
        \label{equ-xi-variance-conditional-sum}
    \end{equation}
    
    Assuming that $\gamma_i \xrightarrow{P} \gamma_\infty$ as $i \rightarrow +\infty$, and assuming that $\sigma^2(\gamma)$ is continuous in $\gamma$, we get that $\sigma^2(\gamma_i) \xrightarrow{P} \sigma^2(\gamma_\infty)$ as $i \rightarrow +\infty$. Since $\sigma^2(\gamma)$ is uniformly bounded, it follows also that $\sigma^2(\gamma_i) \xrightarrow{\mathcal{L}^1} \sigma^2(\gamma_\infty)$ as $i \rightarrow +\infty$.

    From this, it is easy to show that
    \begin{equation}
        \frac{1}{N_k} \sum_{i=0}^{k-1} \mathbb{E}^{\gamma_0}\left( \xi_i^2 \mid \Tilde{\mathcal{F}}_{i-1} \right) \xrightarrow{\mathcal{L}^1} \sigma^2(\gamma_\infty), \qquad \text{as } k \rightarrow +\infty,
        \label{equ-martingale-variance-limit}
    \end{equation}
    which implies condition (2.4) from \cite{dvoretzky1972asymptotic}.

    It remains for us to show the following Lindeberg condition
    \begin{equation}
        \frac{1}{N_k} \sum_{i=0}^{k-1} \mathbb{E}^{\gamma_0}\left( \xi_i^2 \mathbbm{1}\left( \xi_i^2 > \delta N_k \right) \mathrel{}\middle|\mathrel{} \Tilde{\mathcal{F}}_{i-1} \right) \xrightarrow{P} 0,
        \label{equ-lindeberg}
    \end{equation}
    as $k \rightarrow +\infty$, for every $\delta >0$.

    Now, we can apply the Cauchy-Schwartz inequality and Markov's inequality to get that
    \begin{align*}
        \mathbb{E}^{\gamma_0}\left( \xi_i^2 \mathbbm{1}\left( \xi_i^2 > \delta N_k \right) \mathrel{}\middle|\mathrel{} \Tilde{\mathcal{F}}_{i-1} \right) =& \mathbb{E}^{\gamma_0}\left( \xi_i^2 \mathbbm{1}\left( \xi_i^2 > \delta N_k \right) \mathrel{}\middle|\mathrel{} \gamma_i \right) \\
        \leq& \mathbb{E}^{\gamma_0}\left( \xi_i^4 \mathrel{}\middle|\mathrel{} \gamma_i \right)^{1/2} \times \mathbb{P}^{\gamma_0}\left( \xi_i^2 > \delta N_k \mathrel{}\middle|\mathrel{} \gamma_i \right)^{1/2} \\
        \leq& \mathbb{E}^{\gamma_0}\left( \xi_i^4 \mathrel{}\middle|\mathrel{} \gamma_i \right)^{1/2} \left[ \frac{\mathbb{E}^{\gamma_0}\left( \xi_i^4 \mathrel{}\middle|\mathrel{} \gamma_i \right)}{\delta^2 N_k^2} \right]^{1/2} \\
        =& \frac{1}{\delta N_k} \mathbb{E}^{\gamma_0}\left( \xi_i^4 \mathrel{}\middle|\mathrel{} \gamma_{i-1} \right).
    \end{align*}

    By Lemma \ref{lemma-xi-2-moments}, we get that $\exists M>0$ such that
    \begin{equation}
        \frac{1}{N_k} \sum_{i=0}^{k-1} \mathbb{E}^{\gamma_0}\left( \xi_i^2 \mathbbm{1}\left( \xi_i^2 > \delta N_k \right) \mathrel{}\middle|\mathrel{} \Tilde{\mathcal{F}}_{i-1} \right) \leq \frac{M \sum_{i=0}^{k-1} n_{i+1}^2}{\delta N_k^2},
        \label{equ-xi-lindeberg}
    \end{equation}
    which goes to $0$ as $k \rightarrow +\infty$ by Lemma 1 from \cite{chimisov2018air}.

    Thus, by Theorem 2.2 from \cite{dvoretzky1972asymptotic}, we get that
    \begin{equation}
        \frac{1}{\sqrt{N_k \sigma^2(\gamma_\infty)}} \Xi^{(2)}_k \xrightarrow{D} \mathcal{N}(0,1),
        \label{equ-Xi-clt}
    \end{equation}
    which gives the required result.
\end{proof}

\subsection{Controlling the Other Terms}

We now turn to the remaining terms in Equation \eqref{equ-full-sum-split}.

\begin{lemma}[Moments of $\Xi^{(1)}_k$, $\Xi^{(3)}_k$ and $\Xi^{(4)}_{k,N}$]
    Consider $\Xi^{(1)}_k$, $\Xi^{(3)}_k$ and $\Xi^{(4)}_{k,N}$ as given in Equation \eqref{equ-full-sum-split}. Then, assuming $\mathbb{E}_\pi(f(X))=0$ for $f$ bounded, we have that $\forall (\phi,y) \in \widecheck{\Omega}$ and $\gamma_0 \in \Gamma$,
    \begin{equation*}
        \mathbb{E}^{\gamma_0}_{(\phi,y)}\left( \left[ \Xi^{(1)}_k \right]^2 \right) = \mathcal{O}(k^2), \qquad \mathbb{E}^{\gamma_0}_{(\phi,y)}\left( \left[ \Xi^{(3)}_k \right]^2 \right) =\mathcal{O}(k^2),
    \end{equation*}
    and
    \begin{equation*}
        \mathbb{E}^{\gamma_0}_{(\phi,y)}\left( \left[ \Xi^{(4)}_{k,N} \right]^2 \right) = \mathcal{O}(n_{k+1}).
    \end{equation*}
    \label{lemma-moments-xi-1,3,4}
\end{lemma}

The proof of this Lemma follows from Jensen's inequality for $\Xi^{(1)}_k$ and $\Xi^{(3)}_k$, and from a regeneration argument for $\Xi^{(4)}_{k,N}$.
\begin{proof}[Proof of Lemma \ref{lemma-moments-xi-1,3,4}]
    By Jensen's inequality, we have that
    \begin{align}
        \mathbb{E}^{\gamma_0}_{(\phi,y)}\left( \left[ \Xi^{(1)}_k \right]^2 \right) =& \mathbb{E}^{\gamma_0}_{(\phi,y)}\left( \left[ \sum_{i=0}^{k-1} \eta_i \right]^2 \right) \notag \\
        \leq& k \sum_{i=0}^{k-1} \mathbb{E}^{\gamma_0}_{(\phi,y)}\left(  \eta_i^2 \right) \notag \\
        =& k \sum_{i=0}^{k-1} \mathbb{E}^{\gamma_0}_{(\phi,y)} \left( \left[ \sum_{j=0}^{T^{(i)}_1} F(H^{(i)}_j) \right]^2 \right) \notag \\
        \leq& k^2 M^2 \mathbb{E}\left( (\widecheck{T}_1 + 1)^2 \right),
        \label{equ-xi-1-bound}
    \end{align}
    where $\lvert f \rvert \leq M$, and since $T^{(i)}_1 \sim \text{Geom}(\epsilon)$ for each $i$.
    
    Since the renewal sequence $\{T^{(i)}_j\}_{j \geq 1}$ has $\text{Geom}(\epsilon)$ increments and is therefore memoryless, an identical argument gives that
    \begin{equation}
        \mathbb{E}^{\gamma_0}_{(\phi,y)}\left( \left[ \Xi^{(3)}_k \right]^2 \right) \leq k^2 M^2 \mathbb{E}\left( (\widecheck{T}_1 + 1)^2 \right).
        \label{equ-xi-3-bound}
    \end{equation}

    Now, for $\mathbb{E}^{\gamma_0}_{(\phi,y)}\left( \left[ \Xi^{(4)}_{k,N} \right]^2 \right)$, we follow an identical argument to Lemma \ref{lemma-xi-2-moments}. Setting
    \begin{equation}
        L^{(4)} = L^{(4)}(k,N) = \inf(l \geq 1 : T^{(k+1)}_l \geq N - N_k -1),
    \label{equ-xi^4-leftover-refreshment}
    \end{equation}
    we write
    \begin{equation}
        \begin{split}
            \eta^{(4)} = \sum_{j=0}^{T^{(k+1)}_1} F(H^{(k+1)}_j), &\qquad \xi^{(4)} = \sum_{j = T^{(k+1)}_1 +1}^{T^{(k+1)}_{L^{(4)}}} F(H^{(k+1)}_j) \\
            \zeta^{(4)} = \sum_{j = N - N_k}^{T^{(k+1)}_{L^{(4)}}} F(H^{(k+1)}_j), &\qquad \xi_{i,j}^{(4)} = \sum_{m = T^{(k+1)}_j +1}^{T^{(k+1)}_{j+1}} F(H^{(k+1)}_m).
        \end{split}
        \label{equ-xi^4-sum-terms}
    \end{equation}
    We can therefore divide $\Xi^{(4)}_{k,N}$ as
    \begin{equation}
        \Xi^{(4)}_{k,N} = \eta^{(4)} + \xi^{(4)} - \zeta^{(4)},
        \label{equ-xi^4-sum-split}
    \end{equation}
    as we did for $s_i$ in Equation \eqref{equ-sum-split}. Now, by an identical argument to Equations \eqref{equ-xi-1-bound} and \eqref{equ-xi-3-bound}, we get that
    \begin{equation}
        \mathbb{E}_{(\phi,y)}^{\gamma_0}\left( \left[ \eta^{(4)} \right]^2 \right) = \mathcal{O}(1), \qquad \mathbb{E}_{(\phi,y)}^{\gamma_0}\left( \left[ \zeta^{(4)} \right]^2 \right) = \mathcal{O}(1),
        \label{equ-eta-zeta-4-order}
    \end{equation}
    and by an identical argument to Lemma \ref{lemma-xi-2-moments}, we get that
    \begin{equation}
        \mathbb{E}_{(\phi,y)}^{\gamma_0}\left( \left[ \xi^{(4)} \right]^2 \right) = \mathcal{O}(N - N_k) \implies \mathbb{E}_{(\phi,y)}^{\gamma_0}\left( \left[ \xi^{(4)} \right]^2 \right) = \mathcal{O}(n_{k+1}).
        \label{equ-xi-4-variance}
    \end{equation}
    Using these bounds and Equation \eqref{equ-xi^4-sum-split} therefore gives
    \begin{equation}
        \mathbb{E}_{(\phi,y)}^{\gamma_0}\left( \left[ \Xi^{(4)}_{k,N} \right]^2 \right) = \mathcal{O}(n_{k+1}),
        \label{equ-Xi-4-variance}
    \end{equation}
    which completes the proof.
\end{proof}

\subsection{Gathering all the Terms}

We can now conclude by gathering all the terms. This will yield the results given in Theorem \ref{thm-air-segment-chain-1}.

Recall first of all that, from Equation \eqref{equ-poly-lags}, we are assuming that $n_k = \Theta(k^\beta)$, which implies $N_k = \Theta(k^{\beta+1})$. We can start by gathering our bounds from Lemma \ref{lemma-xi-2-moments} and \ref{lemma-moments-xi-1,3,4} to show the following:

\begin{theorem}[Segment Chain $\mathcal{L}^2$ convergence]
    Consider $(\Phi_n,Y_n)_{n \geq 0}$ the AIR segment chain process, where the original family of transition kernels $\{P_\gamma\}_{\gamma \in \Gamma}$ have invariant distributions $\pi_\gamma$ and satisfy the minorisation condition \eqref{equ-minorisation-general}. Assume that $f$ is bounded and $\mathbb{E}_{\pi_\gamma}(f(X))=0$ for all $\gamma \in \Gamma$. Then $\forall (\Phi_0,Y_0) \in \widecheck{\Omega}$, $\gamma_0 \in \Gamma$, any $\beta>0$, and any adaptation scheme,
    \begin{equation*}
        \mathbb{E}_{(\Phi_0,Y_0)}^{\gamma_0}\left( \left[ \frac{1}{N} S_N \right]^2 \right) = \mathcal{O}\left(N^{-\min\left(1, \frac{2\beta}{1+\beta}\right)}\right),
    \end{equation*}
    so that $\mathcal{L}^2$ convergence holds.
    \label{thm-L2}
\end{theorem}

The proof of this theorem is simply a matter of gathering terms from previous results and is identical to the proof of Theorem 1 from \cite{chimisov2018air}. Rearranging the terms yields the expression in Theorem \ref{thm-air-segment-chain-1}.

\begin{proof}[Proof of Theorem \ref{thm-L2}]
    For the $\mathcal{L}^2$ convergence, note that by Lemma \ref{lemma-xi-2-moments} and \ref{lemma-moments-xi-1,3,4}, we have
    \begin{equation*}
        \mathbb{E}_{(\Phi_0,Y_0)}^{\gamma_0}\left( \left[ \frac{1}{N_k} S_N \right]^2 \right) = \mathcal{O}\left(\frac{1}{N_k}\right) + \mathcal{O}\left(\frac{k^2}{N_k^2}\right) + \mathcal{O}\left(\frac{n_{k+1}}{N_k^2}\right).
    \end{equation*}
    Since $\frac{N_k}{N} \rightarrow 1$ as $N \rightarrow +\infty$, and $n_k$, $N_k$ are of order $k^\beta$ and $k^{\beta +1}$ respectively, we get that
    \begin{equation*}
        \mathbb{E}_{(\Phi_0,Y_0)}^{\gamma_0}\left( \left[ \frac{1}{N} S_N \right]^2 \right) = \mathcal{O}( N^{-1} ) + \mathcal{O}\left( N^{-\frac{2\beta}{\beta +1}} \right),
    \end{equation*}
    which gives the $\mathcal{L}^2$ convergence.
\end{proof}

Finally, we prove the central limit theorem for $S_N$.

\begin{theorem}[Central Limit Theorem for AIR SBPS]
    Suppose $\beta>1$. Consider $(\Phi_n,Y_n)_{n \geq 0}$ the AIR segment chain process, where the original family of transition kernels $\{P_\gamma\}_{\gamma \in \Gamma}$ have invariant distributions $\pi_\gamma$ and satisfy the minorisation condition \eqref{equ-minorisation-general}. Assume that $f$ is bounded and $\mathbb{E}_{\pi_\gamma}(f(X))=0$ for all $\gamma \in \Gamma$. Then, if $\gamma_i \xrightarrow{P} \gamma_\infty$ such that $\sigma^2(\gamma_\infty)>0$, and that $\sigma^2(\gamma)$ is a continuous function of $\gamma$, then $\forall (\Phi_0,Y_0) \in \widecheck{\Omega}$, $\gamma_0 \in \Gamma$,
    \begin{equation*}
        \frac{1}{\sqrt{N}}S_N \xrightarrow{D} \mathcal{N}(0, \sigma^2(\gamma_\infty)).
    \end{equation*}
    \label{thm-CLT}
\end{theorem}

This theorem follows immediately from Lemma \ref{lemma-clt-Xi2}'s CLT for $\frac{1}{\sqrt{N_k}}\Xi^{(2)}_k$, then using the moment bounds in Lemma \ref{lemma-moments-xi-1,3,4} to show that the other terms go to $0$ in probability for large $N$, assuming $\beta >1$.

\end{appendix}

\begin{funding}
    Cameron Bell has been supported by the EPSRC on studentship grant EP/W523793/1 and by a PR[AI]RIE-PSAI Chair funded by the Agence Nationale de la Recherche (ANR-23-IACL-0008).
    
    Krzysztof {\L}atuszy{\'n}ski has been supported by the Royal Society through the Royal Society University Research Fellowship.

    Gareth O. Roberts has been supported by the UKRI grant EP/Y014650/1 as part of the ERC Synergy project OCEAN,  EPSRC grants Bayes for Health (R018561), CoSInES (R034710), PINCODE (EP/X028119/1), and EP/V009478/1.
\end{funding}

\begin{supplement}
All code used for simulations can be found in our GitHub repository:

https://github.com/tamarock143/Julia-Stereographic.
\end{supplement}


\bibliographystyle{imsart-nameyear} 
\bibliography{Bibliography.bib}       


\end{document}